%% file: main-final-final.tex
\documentclass{lmcs}
\pdfoutput=1
\usepackage[utf8]{inputenc}

\usepackage{lastpage}
\lmcsdoi{19}{3}{1}
\lmcsheading{}{\pageref{LastPage}}{}{}%
{Jun.~01,~2021}{Jul.~06,~2023}{}

\usepackage{pinlabel}
\usepackage{latexsym}
\usepackage{amssymb}
\usepackage{amsmath}
\usepackage{amsfonts}
\usepackage{graphicx}
\usepackage{array, longtable}
\newcolumntype{C}{>{$}c<{$}}
\newcolumntype{L}{>{$}l<{$}}
\newcolumntype{R}{>{$}r<{$}}
\usepackage{amsthm}
\usepackage{bussproofs}

\allowdisplaybreaks

\title[The Gray Code Case]{Computing with Infinite Objects: the Gray Code Case}

\author[D.~Spreen]{Dieter Spreen\lmcsorcid{0000-0002-2773-7323}}[a]
\address{Department of Mathematics, University of Siegen, 57068 Siegen, Germany}
\email{spreen@math.uni-siegen.de}
\thanks{
\protect\includegraphics[width=1.1em]{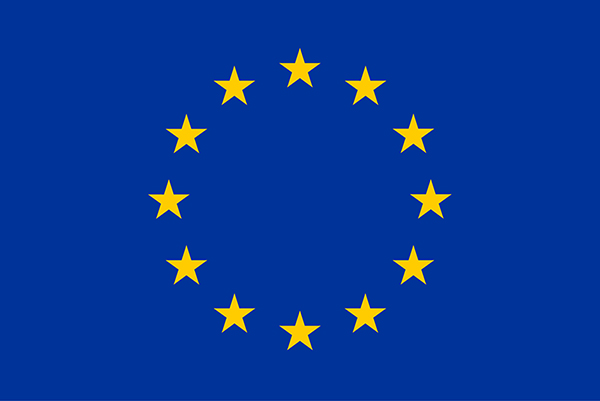}
This project has received funding from the European Union's Horizon 2020 research and innovation programme under the Marie Sk\l{}odowska-Curie grant agreement No 731143.}

\author[U.~Berger]{Ulrich Berger\lmcsorcid{0000-0002-7677-3582}}[b]
\address{Department of Computer Science, Swansea University, The Computational Foundry,
Swansea University Bay Campus, Fabian Way, Swansea, SA1 8EN, UK}
\email{u.berger@swansea.ac.uk}

\newenvironment{ncase}{\emph{Case\,}}{}

\newcommand{\set}[2]{\{\,#1 \mid #2\,\}}
\newcommand{\fun}[3]{#1\colon #2 \to #3}
\newcommand{\llbracket}{[\![}
\newcommand{\rrbracket}{]\!]}
\newcommand{\range}{\mathrm{range}}
\newcommand{\card}[1]{|\!|#1|\!|}
\newcommand{\llparenthesis}{(\!|}
\newcommand{\rrparenthesis}{|\!)}

\newcommand{\alphaeq}{=_{\alpha}}

\newcommand{\CC}{\mathbb{C}}

\newcommand{\II}{\mathbb{I}}

\newcommand{\NN}{\mathbb{N}}

\newcommand{\RR}{\mathbb{R}}
\newcommand{\TT}{\mathbb{T}}

\newcommand{\AAA}{\mathcal{A}}

\newcommand{\KKK}{\mathcal{K}}

\newcommand{\PPP}{\mathcal{P}}

\newcommand{\TTT}{\mathcal{T}}

\newcommand{\bone}{\mathbf{1}}
\newcommand{\btwo}{\mathbf{2}}
\newcommand{\bthree}{\mathbf{3}}
\newcommand{\bfalse}{\mathbf{False}}
\newcommand{\bfix}{\mathbf{fix}\,}

\newcommand{\bfA}{\mathbf{A}} 
\newcommand{\bB}{\mathbf{B}}

\newcommand{\bD}{\mathbf{D}}
\newcommand{\bG}{\mathbf{G}} 
\newcommand{\bGC}{\mathbf{GC}}
\newcommand{\bH}{\mathbf{H}}
\newcommand{\bII}{\mathbf{I}\!\mathbf{I}}

\newcommand{\bK}{\mathbf{K}}
\newcommand{\bN}{\mathbf{N}}
\newcommand{\br}{\,\mathbf{r}\,}
\newcommand{\brk}{\,\mathbf{rk}}
\newcommand{\bP}{\mathbf{P}}
\newcommand{\bQ}{\mathbf{Q}}
\newcommand{\bR}{\,\mathbf{R}}
\newcommand{\bS}{\mathbf{S}} 
\newcommand{\bSD}{\mathbf{SD}}
\newcommand{\bZ}{\mathbf{Z}}
\newcommand{\bt}{\mathbf{t}}
\newcommand{\bnil}{\mathbf{Nil}}
\newcommand{\bleft}{\mathbf{Left}}
\newcommand{\bright}{\mathbf{Right}}
\newcommand{\bpair}{\mathbf{Pair}}
\newcommand{\bfun}{\mathbf{Fun}}
\newcommand{\bcase}{\mathbf{case}\,}

\newcommand{\bof}{\mathbf{of}\,}
\newcommand{\brec}{\mathbf{rec}\,}
\newcommand{\amb}{\mathbf{Amb}}
\newcommand{\amblr}{\amb_{\mathbf{LR}}}

\newcommand{\mamb}{\mathbf{mapamb}\,}
\newcommand{\mlr}{\mathbf{mapLR}\,}

\newcommand{\AV}{\mathrm{AV}}
\newcommand{\GC}{\mathrm{GC}}
\newcommand{\GF}{\mathrm{GF}}
\newcommand{\SD}{\mathrm{SD}}

\newcommand{\av}[1]{\mathrm{av}_{#1}}
\newcommand{\bav}[1]{\mathbf{av}_{#1}}

\newcommand{\ball}[3]{\mathrm{B}_{#1}(#2,#3)}

\newcommand{\hdm}{\mu_{\mathrm H}}
\newcommand{\treehdm}{\delta_{\mathrm H}}
\newcommand{\Def}{\overset{\mathrm{Def}}{=}}

\newcommand{\ccase}[2]{\bcase #1\, \bof \{ #2 \}}

\newcommand{\ddown}{\mathord{\downdownarrows}}
\newcommand{\itdown}{\mathord{\stackrel{\ast}{\ddown}}}
\newcommand{\low}{\mathopen\downarrow}
\newcommand{\stc}[2]{#1\low #2}

\newcommand{\val}[1]{\llbracket #1 \rrbracket}
\newcommand{\treeval}[1]{\llparenthesis #1 \rrparenthesis}

\newcommand{\fin}{\mathbf{fin}}
\newcommand{\nil}{\mathbf{Nil}}

\newcommand{\acc}{\mathbf{Acc}}
\newcommand{\prog}{\textbf{Prog}}
\newcommand{\path}{\textbf{Path}}
\newcommand{\wfi}{\textbf{WFI}}
\newcommand{\bmin}{\mathbf{min}\,}
\newcommand{\bmax}{\mathbf{max}\,}

\newcommand{\pfun}{\mathbf{fun}}
\newcommand{\ssp}{\rightsquigarrow}
\newcommand{\sspc}{\overset{\mathrm{c}}{\ssp}}

\newcommand{\dom}{\mathrm{dom}}
\newcommand{\FV}{\mathrm{FV}}
\newcommand{\btrue}{\mathbf{True}}

\newcommand{\bnat}{\mathbf{nat}}

\newcommand{\extra}[1]{#1}
\newcommand{\ignore}[1]{}

\newcommand{\mon}{\mathbf{mon}}
\newcommand{\iter}{\mathbf{it}}
\newcommand{\coiter}{\mathbf{coit}}
\newcommand{\hscoiter}{\mathbf{hscoit}}
\newcommand{\chscoiter}{\mathbf{chscoit}}
\newcommand{\scoiter}{\mathbf{scoit}}

\newcommand{\absorb}{\mathbf{absorb}}
\newcommand{\id}{\mathbf{id}}
\newcommand{\caibs}{\mathbf{caibs}}

\newcommand{\eps}{\mathrel{\varepsilon}}
\newcommand{\noteps}{\mathrel{\not\!\varepsilon}}

\def\int{\mathop{\mathstrut\rm int}\nolimits}

\newcommand{\rt}[2]{#2\mathord\upharpoonright_{#1}}
\newcommand{\rtbt}[2]{#2|_{#1}}
\newcommand{\rtu}[2]{#2\mathord\upharpoonright^{\rm u}_{#1}}
\newcommand{\rest}{\mathord\restriction}
\newcommand{\restu}{\mathord\upharpoonright^{\rm u}}

\newcommand{\mytextcolor}[2]{#2}

\begin{document}

\begin{abstract}
Infinite Gray code has been introduced by Tsuiki~\cite{ts} as a redundancy-free representation of the reals. In applications the signed digit representation is mostly used which has maximal redundancy. Tsuiki presented a functional program converting signed digit code into infinite Gray code. Moreover, he showed that infinite Gray code can effectively be converted into signed digit code, but the program needs to have some non-deterministic features (see also \cite{tsug}). Berger and Tsuiki~\cite{btifp,bt} reproved the result in a system of formal first-order intuitionistic logic extended by inductive and co-inductive definitions, as well as some new logical connectives capturing concurrent behaviour. The programs extracted from the proofs are exactly the ones given by Tsuiki. In order to do so, co-inductive predicates $\bS$ and $\bG$ are defined and the inclusion $\bS \subseteq \bG$
is derived. For the converse inclusion the new logical connectives are used to introduce a concurrent version $\bS_{2}$ of $\bS$ and $\bG \subseteq \bS_{2}$ is shown. What one is looking for, however, is an equivalence proof of the involved concepts. One of the main aims of the present paper is to close the gap. A concurrent version $\bG^{*}$ of $\bG$ and a modification $\bS^{*}$ of $\bS_{2}$ are presented such that $\bS^{*} = \bG^{*}$. A crucial tool in \cite{btifp} is a formulation of the Archimedean property of the real numbers as an induction principle. We introduce a concurrent version of this principle which allows us to prove that $\bS^{*}$ and $\bG^{*}$ coincide. 
A further central contribution is the extension of the above results to the hyperspace of 
non-empty compact subsets of the reals.
\end{abstract}

\ACMCCS{D.2.4.; F.4.1; I.2.2; I.2.3; I.2.4}

\maketitle

\setcounter{tocdepth}{1}
 \tableofcontents

\section{Introduction}\label{sec-intro}
\input{sec-intro-final-extended-final}

\section{Digit Spaces}\label{sec-dig}
\input{sec-dig-final-final}

\section{Inductive and co-inductive definitions}\label{sec-indef}
\input{sec-indef-final-final}

\section{Extracting algorithmic content from co-inductive proofs}\label{sec-progex}

\input{sec-progex-final-final}

\section{Computationally motivated logical connectives}\label{sec-concon}
\input{sec-concon-final-final}

\section{Concurrent Archimedean induction}\label{sec-conarch}
\input{sec-conarch-final-final}

\section{Concurrent signed digit and Gray codes}\label{sec-sdgc}
\input{sec-sdgc-final-final}

\section{The compact sets case}\label{sec-comset}

\input{sec-comset-final-final}

\section{Archimedean induction for compact sets}\label{sec-arcp}

\input{sec-arcp-final-final}

\section{Signed digit and Gray code for non-empty compact sets}\label{sec-cpcode}
\input{sec-cpcode-final-final}

\section{Concurrent Gray code for non-empty compact sets}\label{sec-concgray}
\input{sec-concgray-final-final}

\section{Conclusion}\label{sec-concl}
\input{sec-concl-final-final}

\section*{Acknowledgement}

This research has  been started during the Hausdorff trimester ``Types, Sets and Construction'' at the Hausdorff Research Institute for Mathematics, Bonn, 2018. The authors are grateful to the organisers of the trimester for having arranged this inspiring meeting and to the Hausdorff Institute for providing such excellent working conditions.

Thanks are due to the referees for their careful reading of the paper. They did a wonderful job: errors in results could be eliminated and the overall presentation of the paper improved.

\end{document}

%% file: sec-intro-final-extended-final.tex
In investigations on exact computations with continuous objects such as the real numbers, objects are usually represented by streams of finite data. This is true for theoretical studies in the Type-Two Theory of Effectivity approach (cf.\ e.g.\ \cite{wei}) as for practical research, where prevalently the signed digit representation is used (cf.\ \cite{cg,em,bh}), but also others~\cite{es,eh,ts}. In \cite{beu} it is shown how to use the method of program extraction from proofs to extract certified algorithms working with the signed digit representation in a semi-constructive logic allowing inductive and co-inductive definitions. In addition to producing correct algorithms, this approach allows reasoning in a representation-free way, as in usual mathematical practice. Concrete representations of the objects needed in computations are generated automatically by the extraction procedure. A detailed description of the logic (i.e.\ \emph{Intuitionistic Fixed Point Logic (IFP)})  and the  realisability approach used for extracting programs can be found in~\cite{btifp}.

In order to generalise from the different finite objects used in the various stream representations, the present authors~\cite{bs} used the abstract framework of what was coined \emph{digit space}, i.e.\ a bounded complete non-empty metric space $X$ enriched with a finite set $D$ of contractions on $X$, called \emph{digits}, that cover the space, that is
\[
X = \bigcup \set{d[X]}{d \in D},
\]
where $d[X] = \set{d(x)}{x \in X}$. 

Digit spaces are compact and weakly hyperbolic, where the latter property means that for every infinite sequence $d_{0}, d_{1}, \ldots$ of digits the intersection $\bigcap_{n \in \NN} d_{0} \circ \cdots \circ d_{n}[X]$ contains \emph{at most} one point~\cite{ed}. Compactness on the other hand, implies that each such intersection contains \emph{at least} a point. By this way every stream of digits denotes a uniquely determined point in $X$. Because of the covering property it follows conversely that each point in $X$ has such a code. 

The framework has been generalised in~\cite{sp}. In both papers the proof of the main results required a strengthening of the covering condition in such way that 
\[
X = \bigcup \set{\int(d[X])}{d \in D},
\]
where for a subset $A$ of $X$, $\int(A)$ is the topological interior of $A$. Spaces with this property were called \emph{well-covering}. The usual spaces occurring in applications are of this kind, in particular the space $(\II, \SD)$ consisting of the real interval $\II \Def [-1, 1]$ and the set $\SD \Def \set{\lambda x.\ (x + d)/2}{(d = -1 \vee d = 1) \vee d = 0}$ whose streams of digits are used in the \emph{signed digit representation}.

An important example of a non-well-covering digit space is the space $(\II, \GC)$ with $\GC \Def \set{\lambda x.\ -d \cdot (x - 1)/2}{d = -1 \vee d= 1}$ leading to an extension of finite Gray code to infinite words over the alphabet $\{\, -1, 1 \,\}$, by which each real number in $\II$ except the dyadic rationals in $(-1, 1)$ is represented by exactly one word. Dyadic rationals are represented by two words that differ in only one place. It follows that the corresponding cell contains no information. Tsuiki~\cite{ts} suggested to identify both codes and to fill the cell in which they differ with the symbol $\bot$ for `\emph{unknown}'. Note that the symbol $\bot$ is of a different kind than -1 or 1. It is like the symbol for `\emph{blank}' on a Turing tape which can also be re-written in the course of the computation.
By this way a redundancy-free representation of the interval [-1, 1] is obtained, also called \emph{infinite Gray code}.

There is, however, a price to be paid for getting rid of redundancy. Tsuiki~\cite{ts} proved that the computability notion for real numbers that is obtained with respect to the new representation is equivalent to the widely accepted computability notion based on the Type-Two Theory of Effectivity approach. To this end he showed that there are computable translations from streams of digits of a real number with respect to the signed digit representation into a stream representing the same number in infinite Gray code, and vice versa. As turned out, the translation of infinite Gray code into signed digit representation cannot be computed purely sequentially: one must have access not only to the head of the input stream, but also to the entry next to it, similarly when writing. To this end, the algorithm has to work non-deterministically.  

The representation of elements of a digit space $(X, D)$  by streams of digits can be characterised co-inductively. Let $\CC_{X} \subseteq X$ be co-inductively defined by
\[
x \in \CC_{X} \overset{\nu}{=} (\exists d \in D) (\exists y \in \CC_{X})\, x = d(y),
\]
that is, $\CC_{X}$ is the largest subset of $X$ satisfying the equation (see Section~\ref{sec-indef} for the theory of inductive and co-inductive definitions). Then from a constructive proof that $x \in \CC_{X}$ one can extract a stream of digits representing $x$. Note that from the covering property of digit spaces it follows (by co-induction) that $X \subseteq \CC_{X}$. However, in general this is only true in classical logic since for an arbitrary element $x \in X$ one can usually not determine constructively a digit $d \in D$ whose image contains $x$. Hence, constructively, $\CC_{X}$ is normally a proper subset of $X$. 
However, if the digit space has an effective basis and is well-covering, then $\CC_{X}$ yields a representation of $X$ that is constructively equivalent to the standard Cauchy representation. Since this is the case for the digit space $(\II, \SD)$, we let $\bS = \CC_{(\II, \SD)}$, and  obtain $\bS$ as the largest set of real numbers in $[-1, 1]$ that (constructively) has a signed digit representation.

The same approach does not work for infinite Gray code since, as pointed out earlier, its digit space, $(\II, \GC)$, is non-well covering. Nevertheless, Berger and Tsuiki presented in \cite{btifp}, the following co-inductive characterisation $\bG$ of the interval $\II$ that does allow for the extraction of infinite Gray code: 
\[\bG(x) \overset{\nu}{=} (-1 \le x \le 1) \wedge \bD(x) \wedge \bG(\bt(x)). \]
Here, $\bD(x) \Def x \not= 0 \to (x \le 0 \vee x \ge 0)$, and $\bt$ is the tent function
$\bt(x) \Def 1 - 2|x|$ which is the continuous join of the inverses of the digits in $\GC$.
Note that, if $x\not= 0$, then a realiser of $\bD(x)$ is 
a digit deciding the disjunction $x \le 0 \vee x \ge 0$. 
However, in the case $x=0$ the realiser may be undefined. 
This provides a logical explanation why an infinite Gray code may contain an undefined digit.
In \cite{btifp} it is shown that that $\bS \subseteq \bG$ is provable in IFP and that the extracted algorithm is exactly the translation from signed digit representation into infinite Gray code given in \cite{ts}. 

When trying to extract Tsuiki's translation in the opposite direction from a proof of $\bG \subseteq \bS$, one faces the obstacle that IFP, being based on traditional realisability, can only extract algorithms that are deterministic and sequential, while, as discussed above, an effective translation from infinite Gray code to the signed digit \mytextcolor{red}{representation} is necessarily non-deterministic and concurrent.
To overcome this limitation of IFP, in \cite{bt} an extension of IFP, \emph{Concurrent Fixed Point Logic (CFP)}, is developed. Its main novelty is a \emph{concurrency modality} $\ddown(A)$ indicating that realisers of $A$ may
be computed by two concurrent threads, where the result of the thread terminating first is taken as realiser and the other thread is discarded. 
More precisely, realisability of $\ddown(A)$ is defined using a version of McCarthy's Amb~\cite{mc}:
\begin{align*}
c \br \ddown(A) \Def\, &c = \amb(a,b) 
\wedge \mbox{} \\
&(a \neq\bot \vee b \neq \bot) \wedge \mbox{} \\
&(a \neq\bot \to a \br A) \wedge (b \neq\bot \to b \br A).
\end{align*}
Besides solving the above translation problem, the motivation for introducing this modality is the wish to provide a constructive interpretation of the law of excluded middle 
\[\dfrac{B \to A \quad \neg B \to A}{A}\,\, \text{(lem)}\]
where $B$ is a formula without computational content. The idea is that (lem) should be realised by $\amb$, since a realiser of the conclusion should be computable by running the given realisers of the two premises in parallel. In turns out that, to make this work, it is not enough to add the concurrency modality to the conclusion: one must also modify the premises to avoid false positives.
This results in the rule 
\[\dfrac{\rt{B}{A} \quad \rt{\neg B}{A}}{\ddown(A)}\,\, \text{($\ddown$-lem)} \]
where $\rt{B}{A}$ is a strengthening of the implication $B \to A$, called \emph{restriction}\footnote{In\cite{bt} the notation $\rtbt{B}{A}$ is used instead of $\rt{B}{A}$.}, that guarantees that all its defined realisers are in fact realisers of $A$, independently of the truth value of $B$. More precisely, realisability for restriction is defined as
\[
a \br \rt{B}{A} \Def 
(
B 
\to a \neq \bot) \wedge
(a \neq\bot \to a \br A).
\]
whereas $a \br (B \to A) \Def B \to a \br A$. The definitions of realisability for the concurrency modality and restriction shown above are slightly simplified; for full definitions, see Section~\ref{sec-concon}.

In \cite{bt} the definition of $\bS$ is modified by making use of the new modality for concurrency:
\[
\bS_{2}(x) \overset{\nu}{=} \ddown((\exists d \in \bSD)\, \bII(d, x) \wedge \bS_{2}(2x - d)),
\]
where $\bII(d,x) \Def |2x - d| \le 1$.
In terms of realisability, $\bS_{2}(x)$ means that a signed digit representation of $x$ is obtained through the concurrent computation of two threads. 
Now, the inclusion $\bG \subseteq \bS_{2}$  can be derived (\cite{bt}, see also Theorem~\ref{thm-StoGtoS2}), and it turns out that the extracted algorithm is the one given in \cite{ts}. 

So far, we reviewed the results in \cite{btifp} and \cite{bt} which our work builds on. In the following we give an overview of the main new result of the present paper.

As we have seen, from the results in \cite{btifp} and \cite{bt} it follows that $\bS \subseteq \bG \subseteq \bS_{2}$. What one is looking for, however, is a proof of the equivalence of the involved concepts. In this paper we present a concurrent version $\bG^{*}$ of $\bG$ and a modification $\bS^{*}$ of $\bS_{2}$ so that $\bS^{*} = \bG^{*}$. $\bS^{*}$ and $\bG^{*}$ are defined with the following iterated form of the concurrency operator~(cf.\ Section~\ref{sub-monad-conc}):
\[
\itdown(A) \overset{\mu}{=} \ddown(A \vee \itdown(A)).
\]
where $\overset{\mu}{=}$ indicates that $\itdown(A)$ is the \emph{least} (i.e.\ logically strongest)
proposition satisfying the equation. A realiser of $\itdown(A)$ consists of two concurrent threads of computations, $\amb(a,b)$, where each thread, if terminating, either provides a realiser of $A$,
or else a new concurrent computation. Since the least fixed point is taken, the first alternative is guaranteed to happen, eventually.

$\bS^{*}$ is now defined like $\bS_{2}$, but with \mytextcolor{red}{$\itdown$ instead} of $\ddown$~(see Section~\ref{sec-sdgc}):
\[ \bS^{*}(x) \overset{\nu}{=} \itdown((\exists d \in \bSD)\, \bII(d, x) \wedge \bS^{*}(2x - d)) \]
The modification of the predicate for infinite Gray code is even simpler since only the decision
$x\le 0 \lor x \ge 0$ is subject to the operator $\itdown$:
\[ \bG^{*}(x) \overset{\nu}{=} (-1 \le x \le 1) \wedge (x \not= 0 \to \itdown(x \le 0 \vee x \ge 0)) \wedge \bG^{*}(\bt(x))\]
This means that extracted realisers are ordinary streams, however with concurrently computed digits.

Another central objective of the present paper is to do a similar thing for the hyperspace of non-empty compact subsets of $\II$. That is, we give a co-inductive characterisation of this space from which a Gray code-like representation of the non-empty compact subsets of $\II$ can be extracted and compare it with the characterisation of the hyperspace of non-empty compact subsets of digit spaces investigated in \cite{bs,sp}, now applied to the digit space $(\II, \SD)$.

The analogue of the signed digit representation for compact sets is
\[
\bS_{\bK}(K) \overset{\nu}{=} \bK(K)\land(\exists E \in \bP_{\fin}(\bSD))\,
K \subseteq \bII_E \wedge (\forall d\eps E) (K_d \not= \emptyset \wedge 
\bS_{\bK}(\bav{d}^{-1}[K_d]))
\]
where $\bK$ is an atomic predicate characterising (axiomatically) the non-empty compact subsets of $\II$, $\bP_{\fin}(\bSD)$ is the set of non-empty subsets of the signed digit set $\bSD$, $\bII_E$ is the set of all points that are $1/2$ close to the half of some digit in $E$, and $K_d$ is the set of points in $K$ that are $1/2$ close to $d/2$ (e.g.~$K_1=K \cap [0,1]$); finally, $\bav{d}^{-1}$ is the inverse of the digit function $\lambda x.(x+d)/2$~(see Definition~\ref{def-sk}).

The generalisation of infinite Gray code of compact sets is trickier. Here, we compute the Gray codes of the minimum and maximum of $K$ and then, recursively, narrow down the set (\mytextcolor{red}{Definition}~\ref{def-gk}):
\[
\bG_{\bK}(K) \overset{\nu}{=} 
\bK(K) \land \bG(\bmin K) \land \bG(\bmax K) \land 
 (\forall d \in  \bGC)\, (K_{d} \ne \emptyset \to  \bG_{\bK}(\bt[K_{d}])).
\]
It is not hard to see that both definitions are generalisations of the point case, that is, 
$\bS(x)$ exactly if $\bS_{\bK}(\{x\})$, and $\bG(x)$ exactly if $\bG_{\bK}(\{x\})$.
We give a constructive proof that $\bS_{\bK} \subseteq \bG_{\bK}$ from which a translation between the respective representations of compact sets can be extracted, thus lifting the corresponding result in~\cite{btifp} from points to the compact sets.

Finally, we also lift the equation $\bS^{*} = \bG^{*}$ from points to compact sets. The definitions of the concurrent versions of $\bS_{\bK}(K)$ and $\bG_{\bK}(K)$ are obtained by putting stars at appropriate places (see Definition~\ref{def-skstar} and the beginning of Section~\ref{sec-concgray}):
\begin{gather*}
\bS^{*}_{\bK}(K) \overset{\nu}{=} \bK(K)\land\itdown((\exists E \in \bP_{\fin}(\bSD))\,
K \subseteq \bII_E \wedge (\forall d \eps E) (K_d \not= \emptyset \wedge \bS^{*}_{\bK}(\bav{d}^{-1}[K_d])))  \\
\bG_{\bK}^{*}(K) \overset{\nu}{=} \bK(K) \land \bG^{*}(\bmin K) \land \bG^{*}(\bmax K) \land (\forall d \in \bGC)\, (K_d \ne \emptyset \rightarrow \bG_{\bK}^{*}(\bt[K_d]))
\end{gather*}
Our final result is the equation $\bS^{*}_{\bK} = \bG^{*}_{\bK}$ which provides an equivalence of the concurrent
signed digit and Gray code representations of non-empty compact sets.

An important proof tool in this work is \emph{Archimedean Induction (AI)}, a formulation of the Archimedean property for real numbers as an induction principle introduced in~\cite{bt}. In the present paper we introduce different version of this principle which are vital for all our main results. Let us briefly discuss the main ideas.

Common formulations of the Archimedean property are either not realisable (e.g.\ $(\forall x) (\exists n\in\NN)\, |x| \le n$), or don't have computational content (e.g.\ $(\forall x)\, ((\forall n\in\NN)\, |x|<2^{-n}) \to x=0$). In contrast, the classically equivalent principle of Archimedean Induction (cf.\ Section~\ref{sec-conarch})) 
\[
\dfrac{(\forall x \ne 0)\ (|x| \leq 1/2 \to P(2x)) \to P(x)}
{(\forall x \ne 0)\ P(x)}\,\, (\mathrm{AI})
\]
inherits computational content from the (arbitrary) predicate $P$ and is realised by general recursion.  It is crucial that the variable $x$ is not relativised to a predicate such as $\bS$ that would yield a representation of $x$.
A simple example of an application of (AI) is $(\forall x,y\in\bS)\,(x+y\neq 0 \to (\exists n\in\NN)\, (|x|>2^{-n} \lor |y|>2^{-n}))$, which \mytextcolor{red}{one} would normally prove using the Archimedean property (in the form without computational content), countable choice ($\mathrm{AC}^\omega$) and  Markov's principle (MP). Using (AI) one needs neither ($\mathrm{AC}^\omega$) nor (MP). Roughly speaking, (AI) can be viewed as a combination of all those three principles that avoids speaking about infinite sequences.

In this paper, we will use variants of (AI) and of the following classically equivalent but constructively slightly weaker form of (AI) which has another predicate $B$ as parameter:
\[
\dfrac{(\forall x \in B \setminus \{ 0 \})\, P(x) \vee (|x| \le 1/2 
\wedge B(2x) \wedge (P(2x) \to P(x)))}
{(\forall x \in B \setminus \{ 0 \})\, P(x)} \, \,(\mathrm{AIB})
\]
An example is a variant where both premise and conclusion are made concurrent (cf.\ Definition~\ref{def-caibstar}):
\[
\dfrac{(\forall x \in B \setminus \{ 0 \})\, \itdown(P(x) \vee (|x| \le 1/2 
\wedge B(2x) \wedge (P(2x) \to P(x))))}
{(\forall x \in B\setminus \{ 0 \})\, \itdown(P(x))} \,\, (\mathrm{CAIB^*})
\]
We also introduce versions of (AI) or (AIB) for compact sets (cf.\ Definition~\ref{def-arcp}), signed digit represented compact sets (cf.\ Definition~\ref{def-aicsd}), and the restriction operator $\rest$ (cf.\ Definition~\ref{def-aicr}). 

The paper is organised as follows: In Section~\ref{sec-dig} the definition of a digit space is recalled and extended Gray code introduced. Section~\ref{sec-indef} contains a short introduction to inductive and
co-inductive definitions and the proof methods they come equipped with. The application to digit spaces is discussed as well.

The next two sections give brief introductions to the logical systems used for program extraction. Section~\ref{sec-progex} deals with \emph{Intuitionistic Fixed Point Logic (IFP)} and the kind of realisability used to generate programs. 
We follow~\cite{btifp} except that in the case-construct of the programming language we permit clauses with overlapping patterns (see~Section~\ref{sub-prog}).
In Section~\ref{sec-concon}, the extension of IFP to \emph{Concurrent Fixed Point Logic (CFP)}  is discussed. The logic contains two new connectives from \cite{bt} and corresponding proof rules. The rules are all realisable.  We will derive further rules. 

In Section~\ref{sec-conarch} and \ref{sec-sdgc}, respectively, concurrent versions of Archimedean induction and the predicates $\bS$ and $\bG$ are introduced. These predicates are such that the realisers of their elements are signed digit and/or Gray code representations of the elements. The concurrent versions of both predicates are shown to coincide. From the proofs computable translations between the two representations can be extracted.

The remaining sections deal with non-empty compact subsets of the interval $\II$. In Section~\ref{sec-comset} some facts presented in \cite{bs} about the representation of the non-empty compact subsets of a digit space by digit trees are recalled. These are finitely branching infinite trees. Their nodes are labelled with digits. The words along the infinite paths are codes of the elements of the represented compact set. In the special case of the non-empty compact subsets of $\II$, the representation is one-to-one, if the elements of $\II$ are represented by infinite Gray code.
Archimedean induction for the non-empty compact subsets of $\II$ is discussed in Section~\ref{sec-arcp}. 

In Section~\ref{sec-cpcode} a predicate $\bG_{\bK}$ is co-inductively defined the realisers of which are Gray code representations of the non-empty compact subsets of $\II$. A similar predicate $\bS_{\bK}$ was defined in the previous section with respect to the signed digit representation. It is shown that $\bS_{\bK} \subseteq \bG_{\bK}$. Just as in the point case, for the converse inclusion a concurrent version $\bS^{*}_{\bK}$ of the predicate $\bS_{\bK}$ has to be considered. In Section~\ref{sec-concgray}, finally, a concurrent version $\bG^{*}_{\bK}$ of the predicate $\bG_{\bK}$ is introduced and the equality $\bS^{*}_{\bK} = \bG^{*}_{\bK}$ is derived. Again computable translations between the digital trees based on signed digit representation and Gray code representation, respectively, can be extracted from the proof.

Realisers are an important ingredient of the approach delineated so far: Results are derived by applying the logical rules of Concurrent Fixed Point logic as well as new rules provided in the paper. But the algorithms used in applications are obtained by following the proof rules and combining their realisers accordingly. For each of the results derived in the real number case, that is in Section~\ref{sec-sdgc}, we will present the realisers obtained in this way. In the compact sets case we leave this to the reader as the proofs follow a pattern very similar to the point case.

\mytextcolor{green}
{
A further crucial aspect of this work is abstraction: The logical language and proof calculus 
do not refer to the operational semantics of programs. Instead, operational soundness
is guaranteed through a general computational adequacy theorem that
applies to any concurrent operational semantics satisfying certain fairness 
requirements\footnote{Committing to a fixed operational semantics would make
concurrent logical rules and programming constructs reduncant since they could be
sequentialised by familiar scheduling/dove-tailing techniques.}~\cite{bt}. 
This means that extracted programs can be executed in 
any efficient concurrent execution model. 
}

%% file: sec-dig-final-final.tex
We review the concept of a digit space \cite{bs,sp} as a general model of computation with infinite streams of digits.
\begin{defi}
Let $(X, \mu)$ be a non-empty compact metric space and $E$ be a finite collection of contracting self-maps $\fun{e}{X}{X}$. Then
$(X, E)$ is a \emph{digit space}, if 
\begin{equation}\label{eq-ds}
X = \bigcup \set{e[X]}{e \in E}.
\end{equation}
\end{defi}
Here, $e[X] = \set{e(x)}{x \in X}$. The maps $e$ will be called \emph{digits} in this context. 

Note that, being a continuous map on a compact set, the metric $\mu$ is bounded.

We identify a finite sequence of digits $\vec e = [e_0, \ldots, e_{n-1}] \in E^n$ with the composition $e_0 \circ \cdots \circ e_{n-1}$ and a digit $e$ with the singleton sequence $[e] \in E^1$. The set of all finite sequences of digits will be denoted by $E^{<\omega}$. Moreover, we let $E^\omega$ be the set of all infinite sequences of elements of $E$ and set for $\alpha \in E^\omega$,
\[
\alpha^{<n} \Def [\alpha_0, \ldots, \alpha_{n-1}].
\]

\begin{lemC}[{\cite[Lemma 2.3]{bs}}]
Let $(X, E)$ be a digit space. Then $\bigcap_{n \in \NN} \alpha^{<n}[X]$ contains exactly one point which we denote by $\val{\alpha}$, for every $\alpha \in E^\omega$.
\end{lemC}
The mapping $\fun{\val{\cdot}}{E^\omega}{X}$ is called the \emph{coding map}.

As is well known, $E^\omega$ is a compact bounded metric space with metric
\[
\delta(\alpha, \beta) = \begin{cases}
					0 & \text{ if $\alpha = \beta$,}\\
					2^{-\min \{\, n \mid \alpha_n \not= \beta_n \,\}} & \text{ otherwise.}
				\end{cases}
\]

\begin{propC}[{\cite[Proposition 2.7]{bs}}]\label{prop-val}\hfill
\begin{enumerate}

\item The coding map $\val{\cdot}$ is onto and uniformly continuous.\label{prop-val-1}

\item The metric topology in $X$ is equivalent to the quotient topology induced by the coding map.\label{prop-val-2}

\end{enumerate}
\end{propC}

Set
\[
\alpha \sim \beta \Longleftrightarrow \val{\alpha} = \val{\beta},
\]
for $\alpha, \beta \in E^\omega$. Then $\sim$ is an equivalence relation. The equivalence class associated with $\alpha \in E^\omega$ will be denoted by $[\alpha ]_\sim$. Furnish the quotient $E^\omega /\mathord{\sim}$ with the quotient topology and let $\fun{q_\sim}{E^{\omega}}{E^{\omega}/\mathord{\sim}}$ be the quotient map. Moreover let $\fun{\widehat{\val{\cdot}}}{E^{\omega}/\mathord{\sim}}{X}$ be the uniquely determined continuous map with $\widehat{\val{\cdot}} \circ q_\sim = \val{\cdot}$. 

\begin{prop}\label{prop-valhom}
The map $\widehat{\val{\cdot}}$ is a homeomorphism between the quotient $E^\omega / \mathord{\sim}$ and the metric space $X$.
\end{prop}
\begin{proof}
By construction $\widehat{\val{\cdot}}$ is a bijection. It remains to show that its inverse $\widehat{\val{\cdot}}^{-1}$ is continuous as well. Let $A$ be a closed set in $E^\omega / \mathord{\sim}$. 
Since $X$ is compact and $q_\sim$ continuous, it follows that $E^\omega / \mathord{\sim}$ is compact as well. Therefore, $A$ is also compact and so is its continuous image $\widehat{\val{A}}$. As $X$ is Hausdorff, we obtain that $\widehat{\val{A}}$ is closed, i.e., $((\widehat{\val{\cdot}})^{-1})^{-1}[A]$ is closed.
\end{proof}

In what follows we will be interested in two sets of digits on the interval $\II \Def [-1, +1] \subseteq \RR$ furnished with the usual Euclidean metric. 

Let $\AV \Def \set{\av{i}}{i \in \SD}$ with $\SD \Def \{ -1, 0, +1\}$ and
\[
\av{i}(x) = (x + i)/2.
\]
Then $(\II, \AV)$ is a digit space. 
Note that for $[i_{0}, \ldots, i_{r-1}] \in \SD^{r}$, 
\[
\range(\av{i_{0}} \circ \cdots \circ \av{i_{r-1}}) = [\sum_{\nu <r} i_{\nu} \cdot 2^{-(\nu+1)} -1, \sum_{\nu <r} i_{\nu} \cdot 2^{-(\nu+1)} +1].
\]
 Hence, for $w \in \SD^{\omega}$ and $\alpha_{w} \in \AV^{\omega}$ with $\alpha_{w}(\nu) \Def \av{w_{\nu}}$, 
 \[
 \val{\alpha_{w}} = \sum_{\nu \ge 0} w_{\nu} \cdot 2^{-(\nu+1)},
 \] 
 that is $w$ is a \emph{signed digit representation} of $\val{\alpha_{w}}$. Therefore, we call $(\II, \AV)$ \emph{signed digit space}. It satisfies a stronger covering condition than (\ref{eq-ds}), which is needed in the development of most of the theory presented in \cite{bs,sp}: $(\II, \AV)$ is well-covering.

\begin{defi}\label{dn-wc}
A digit space $(X, E)$ is \emph{well-covering} if 
\[
X = \bigcup\set{\int(e[X])}{\mytextcolor{red}{e} \in E},
\]
where $\int(e[X])$ denotes the interior of $e[X]$.
\end{defi}

The other digit set we are going to consider leads to an important example of a digit space that is not well-covering.

Let $\GF \Def \set{g_i}{i \in \GC}$ with $\GC \Def \{ -1, 1 \}$ and 
\[
g_i(x) \Def -i \cdot (x-1)/2. 
\]

Then $g_{-1}$ and $g_1$ are contractions with $g_{-1}[\II] = [-1, 0]$ and $g_1[\II] = [0,+1]$. However, $0 \not\in [-1, 0) \cup (0, +1]$. Therefore,

\begin{lem}\label{lem-gnwc}
$(\II, \GF)$ is a digit space that is not well-covering.
\end{lem}

Note that $\GC^{\omega}$ is an extension of finite Gray code to infinite words.

By Proposition~\ref{prop-val}(\ref{prop-val-2}) we know that the metric topology on $\II$ is equivalent to the quotient topology induced by the coding map $\val{\cdot}_G$ associated with $(\II, \GF)$. Set
\[
\alpha \sim_G \beta \Longleftrightarrow \val{\alpha}_G = \val{\beta}_G,
\]
for $\alpha, \beta \in \GF^{\omega}$.

\begin{lem}\label{lem-simG}
For $\alpha, \beta \in \GF^{\omega}$,  $\alpha \sim_G \beta$ if, and only if, either $\alpha = \beta$, or the following Properties (\ref{cond-1}-\ref{cond-3}) hold for some $i \ge 0$:
\begin{enumerate}

\item\label{cond-1} For all $j < i$,  $\alpha_j = \beta_j$.

\item\label{cond-2} $\alpha_i = g_{-1}$ and $\beta_i = g_1$, or conversely, $\alpha_i = g_1$ and $\beta_i = g_{-1}$.

\item\label{cond-3} $\alpha_{i+1} = \beta_{i+1} = g_1$ and $\alpha_j = \beta_j = g_{-1}$, for all $j > i+1$.

\end{enumerate}
\end{lem}
\begin{proof}
Without restriction assume that $\alpha \not= \beta$ and let $i \ge 0$ be such that $\alpha^{<i} = \beta^{<i}$, $\alpha_i = g_{-1}$ and $\beta_i = g_1$. Moreover, let $\alpha_{i+1} = g_1 = \beta_{i+1}$ as well as $\alpha_j = g_{-1} = \beta_j$, for all $j > i+1$. Then we have for $n>0$ that
\begin{gather*}
g_{-1}^n[\II] = [-1, 2^{-(n-1)}-1], \quad (g_1 \circ g_{-1}^n)[\II] = [1-2^{-n}, 1], \\
(g_{-1} \circ g_1 \circ g_{-1}^n)[\II] = [-2^{-(n+1)}, 0], \quad\text{and}\quad (g_1 \circ g_1 \circ g_{-1}^n)[\II] = [0, 2^{-(n+1)}].
\end{gather*}
Therefore,
\[
\{ \val{\alpha}_G \} 
= \bigcap_{j \ge 0} \alpha^{< j}[\II] 
= \bigcap_{n \ge i} \alpha^{< i}[-2^{-(n+1)}, 0] 
\supseteq \alpha^{< i}[\bigcap_{n \ge i} [-2^{-(n+1)}, 0]] 
= \{ \alpha^{< i}(0) \},
\]
from which it follows that $\val{\alpha}_G = \alpha^{< i}(0)$. In the same way we obtain that $\val{\beta}_G = \beta^{< i}(0)$. Hence, $\val{\alpha}_G = \val{\beta}_G$.

For the verification of the converse implication assume that $\alpha \not= \beta$. Then there is a smallest $i \ge 0$ so that $\alpha_i \not= \beta_i$. It follows that either $\alpha_i = g_{-1}$ and $\beta_i = g_1$, or conversely, $\alpha_i = g_1$ and $\beta_i = g_{-1}$. Without restriction we only consider the first case.

Assume that $\alpha_{i+1} = g_{-1}$. Then 
\[
\val{\alpha}_G \in \bigcap_{n > i} \alpha^{< n}[\II] \subseteq \alpha^{< i}[g_{-1}[g_{-1}[\II]]] = \alpha^{< i}[[-1, -1/2]].
\]  

If $\beta_{i+1} = g_{-1}$, we similarly obtain that 
\[
\val{\beta}_G \in \beta^{< i}[g_1[g_{-1}[\II]]] = \beta^{< i}[[1/2, 1]].
\]
Since $\val{\alpha}_G = \val{\beta}_G$,  $\alpha^{< i} = \beta^{< i}$, and the functions $g_{-1}$ and $g_1$ are both one-to-one, it follows that 
\[
\val{\alpha}_G \in \alpha^{< i}[[-1, -1/2] \cap [1/2, 1]],
\]
which is impossible.

On the other hand, if $\beta_{i+1} = g_1$, we have that $\val{\beta}_G \in \beta^{< i}[g_1[g_1[\II]]] = \beta^{< i}[[0, 1/2]]$, and hence that $\val{\alpha}_G \in \alpha^{< i}[[-1, -1/2] \cap [0, 1/2]]$, which is impossible as well.

It follows that $\alpha_{i+1} = g_1$, which means that $\val{\alpha}_G \in \alpha^{< i}[g_{-1}[g_1[\II]]]= \alpha^{< i}\mytextcolor{red}{[[-1/2, 0]]}$. Thus, $\beta_{i+1} = g_1$ as well.

Finally, suppose that there is a minimal $j > i+1$ such that $\alpha_j = g_1$ and $\beta_{i+2} = \cdots = \beta_{j-1} = g_{-1}$, or $\beta_j = g_1$ and $\alpha_{i+2} = \cdots = \alpha_{j-1} = g_{-1}$. Again, we only consider the first case. Then 
\[
\alpha = \alpha_0 \ldots \alpha_{i-1} g_{-1} g_1 g_{-1} \ldots g_{-1} g_1 \alpha_{j+1} \ldots \quad\text{and}\quad
\beta = \beta_0 \ldots \beta_{i-1} g_1 g_1 g_{-1} \ldots g_{-1} \beta_j \beta_{j+1} \ldots,
\]
with $\alpha^{< i} = \beta^{< i}$. It follows that
\begin{align*}
\val{\alpha}_G \in \alpha^{< i}[g_{-1}[g_1[g_{-1}^{j-(i+2)}[g_1[\II]]]]]
&= \alpha^{< i}[g_{-1}[g_1[g_{-1}^{j-(i+2)}[[0,1]]]]] \\
&= \alpha^{< i}[g_{-1}[g_1[[-1+2^{i+2-j}, -1+2^{i+3-j}]]]] \\
&= \alpha^{< i}[g_{-1}[[1-2^{i+2-j}, 1-2^{i+1-j}]]] \\
&= \alpha^{< i}[[-2^{i+1-j}, -2^{i-j}]]
\end{align*}
and $\val{\beta}_G \in \alpha^{< i}[g_1[g_1[g_{-1}^{j-(i+2)}[\beta_j[\II]]]]]$.

Let us first consider the case that $\beta_j = g_1$. Then we have that
\[
\val{\beta}_G \in \alpha^{< i}[g_1[[1-2^{i+2-j}, 1-2^{i+1-j}]]] = \alpha^{< i}[[2^{i-j}, 2^{i+1-j}]].
\]
Hence, $\val{\alpha}_G \in \alpha^{< i}[[-2^{i+1-j}, -2^{i-j}] \cap [2^{i-j}, 2^{i+1-j}]]$, which is impossible.

If $\beta_j = g_{-1}$, we obtain that 
\begin{align*}
\val{\beta}_G \in \alpha^{< i}[g_1[g_1[g_{-1}^{j-(i+2)}[g_1[\II]]]]] 
&= \alpha^{< i}[g_1[g_1[g_{-1}^{j-(i+2)}[[-1, 0]]]]] \\
&= \alpha^{< i}[g_1[g_1[[-1, 2^{i+2-j}-1]]]] \\
&= \alpha^{< i}[g_1[[1-2^{i+1-j}, 1]]] \\
&= \alpha^{< i}[[0, 2^{i-j}]].
\end{align*}
Thus, $\val{\alpha}_G \in \alpha^{< i}[[-2^{i+1-j}, -2^{i-j}] \cap [0, 2^{i-j}]]$, which is impossible again.

By symmetry we obtain similar contradictions in the other cases.
\end{proof}

It follows that each equivalence class $[\alpha ]_{\sim_G}$ contains at most two elements, and if so, then the two differ in exactly one place, which means that the information coming with this entry is not used in the computation of the coding map.

If $[\alpha ]_{\sim_G} = \{\alpha \}$, then $\widehat{\val{[\alpha] _{\sim_G}}} = \val{\alpha}_G$. In the other case $[\alpha]_{\sim_G} = \{ \alpha, \beta\}$ and there is some uniquely determined index $i \ge 0$ (also denoted by $i(\alpha)$) so that $\alpha = \alpha^{< i} g_{-1} g_1 g_{-1}^\omega$ and $\beta= \alpha^{< i} g_1 g_1 g_{-1}^\omega$, or vice versa. Then $\widehat{\val{[\alpha]_{\sim_G}}_G} = \val{\alpha}_G = \val{\beta}_G$ and hence
\begin{align*}
\{\widehat{\val{[\alpha]_{\sim_G}}_G}\} 
&= \{\val{\alpha}_G\} \cup \{\val{\beta}_G\} \\
&= \bigcap_{n \ge 0} \alpha^{< i}[g_{-1}[g_1[g_{-1}^n[\II]]]] \cup \bigcap_{n \ge 0} \alpha^{< i}[g_1[g_1[g_{-1}^n[\II]]]] \\
&= \alpha^{< i}[g_{-1}[\bigcap_{n \ge 0} g_1[g_{-1}^n [\II]]] \cup g_1[\bigcap_{n \ge 0} g_1[g_{-1}^n[\II]]]] \\
&= \alpha^{< i}[(g_{-1} \cup g_1)[\bigcap_{n \ge 0} g_1[g_{-1}^n[\II]]]],
\end{align*}
where the multi-valued function $g_{-1} \cup g_1$ is defined by $(g_{-1} \cup g_1)(x) \Def \{ g_{-1}(x), g_1(x) \}$. Note that in this case for $Y \subseteq \II$,
\[
(g_{-1} \cup g_1)[Y] = \bigcup\set{(g_{-1} \cup g_1)(x)}{x \in Y} = \bigcup\set{ \{ g_{-1}(x), g_1(x) \}}{x \in Y} = g_{-1}[Y] \cup g_1[Y].
\]

Let $\bot \not\in \GC$ be a new symbol ($\bot$ for \emph{unspecified}) and $g_\bot \Def g_{-1} \cup g_1$. It follows that 
\[
\widehat{\val{[\alpha]_{\sim_G}}_G} = \bigcap_{n \ge 0} \alpha^{< i}[g_\bot[g_1[g_{-1}^n[\II]]]].
\]
Set $\overline{\GF} = \{ g_\bot, g_{-1}, g_1 \}$ and define $\fun{\Phi}{\GF^\omega}{\overline{\GF}^\omega}$ by 
\[
\Phi(\alpha) = \begin{cases}
                                \alpha^{<i(\alpha)} g_\bot g_1 g_{-1}^\omega & \text{ if $\card{[\alpha]_\sim} = 2$,}\\
                                \alpha & \text{ otherwise.}
                      \end{cases}
\]                                
Then 
\[
\Phi(\alpha) = \Phi(\beta) \Longleftrightarrow \alpha \sim_G \beta.
\]
The elements of $\widehat{G} \Def \range({\Phi})$ are called \emph{modified Gray code expansions} of the real numbers in $\II$, or just \emph{Gray code} 
\cite{ts}\footnote{Note that in recent research also words $\alpha \in \GF^{\omega}$ with $\card{[\alpha]_{\sim}} = 2$ are considered as valid Gray code \cite{bmst,btifp}.}. 
Topologise $\widehat{G}$ with the topology co-induced by $\Phi$. As we have seen earlier in this section, $\overline{\GF}^\omega$ also possesses a canonical metric. Its restriction to $\widehat{G}$ will be denoted by $\widehat{\delta}$,  whereas $\delta$ denotes the corresponding metric on $\GF^\omega$. For $n \ge 0$, $\widehat{\alpha} \in \widehat{G}$ and $\alpha \in \GF^\omega$, let $\ball{\widehat{\delta}}{\widehat{\alpha}}{2^{-n}}$ and $\ball{\delta}{\alpha}{2^{-n}}$ be the balls in $\widehat{G}$ and $\GF^\omega$, respectively, of radius $2^{-n}$ around $\widehat{\alpha}$ and $\alpha$.

\begin{lem}
$\Phi^{-1}[\ball{\widehat{\delta}}{\widehat{\alpha}}{2^{-n}}]
= \bigcup\set{\ball{\delta}{\beta}{2^{-n}}}{\beta \in \Phi^{-1}[\{ \widehat{\alpha} \}]}$.
\end{lem}
\begin{proof}
Three cases are to be considered.

\begin{ncase}
$\widehat{\alpha}_m = g_\bot$, for some $m > n$.
\end{ncase}
Then $\widehat{\alpha}_0, \ldots, \widehat{\alpha}_{n} \in \GF$ and hence, for any $\beta \in \Phi^{-1}[\{ \widehat{\alpha} \}]$, $\beta_i = \widehat{\alpha}_i$, for $i \le n$. Therefore, if $\gamma \in \Phi^{-1}[\ball{\widehat{\delta}}{\widehat{\alpha}}{2^{-n}}]$, i.e., if $\Phi(\gamma)^{< n+1} = \widehat{\alpha}^{< n+1}$, then $\gamma^{< n+1} = \beta^{< n+1}$, for any $\beta \in \Phi^{-1}[\{ \widehat{\alpha} \}]$, which means that $\gamma \in \ball{\delta}{\beta}{2^{-n}}$, for any such $\beta$.

Conversely, if $\gamma^{< n+1} = \beta^{< n+1}$, for some $\beta \in \Phi^{-1}[\{ \widehat{\alpha} \}]$, then $\Phi(\gamma)^{< n+1} = \Phi(\beta)^{< n+1} = \widehat{\alpha}^{< n+1}$, i.e., $\Phi(\gamma) \in \ball{\widehat{\delta}}{\widehat{\alpha}}{2^{-n}}$.

\begin{ncase}
$\widehat{\alpha}_m = g_\bot$, for some $m \le n$.
\end{ncase}
It follows that $\widehat{\alpha}_0, \ldots, \widehat{\alpha}_{m-1} \in \GF$. Moreover, if $\gamma \in \Phi^{-1}[\ball{\widehat{\delta}}{\widehat{\alpha}}{2^{-n}}]$, then $\gamma^{< m} = \widehat{\alpha}^{< m}$, $\gamma_{m+1} = g_1 = \widehat{\alpha}_{m+1}$, $\gamma_j = g_{-1} = \widehat{\alpha}_j$, for $m+1 < j \le n$, and $\gamma_m = g_{-1}$ or $\gamma_m = g_1$.
Set $\beta = \widehat{\alpha}^{< m} \gamma_m g_1 g_{-1}^\omega$. Then $\gamma \in \ball{\delta}{\beta}{2^{-n}}$ and $\beta \in \Phi^{-1}[\{ \widehat{\alpha} \}]$.

Conversely, if $\gamma^{< n+1} = \beta^{< n+1}$, for some $\beta \in \Phi^{-1}[\{ \widehat{\alpha} \}]$, then $\Phi(\gamma)^{< n+1} = \Phi(\beta)^{< n+1} = \widehat{\alpha}^{< n+1}$, which means that $\gamma \in \Phi^{-1}[\ball{\widehat{\delta}}{\widehat{\alpha}}{2^{-n}}]$.

\begin{ncase}
$\widehat{\alpha}_i \in \GF$, for all $i \ge 0$.
\end{ncase}
In this case we have for $\gamma \in \GF^\omega$ that $\Phi(\gamma)^{< n+1} = \widehat{\alpha}^{< n+1}$, exactly if $\gamma^{< n+1} = \widehat{\alpha}^{< n+1}$, i.e., $\Phi^{-1}[\ball{\widehat{\delta}}{\widehat{\alpha}}{2^{-n}}] = \ball{\delta}{\widehat{\alpha}}{2^{-n}}$.
\end{proof}

This shows that the topology on $\widehat{G}$ co-induced by $\Phi$ is finer than the metric topology. We will now derive the converse.

Let $U \subseteq \widehat{G}$ be open in the topology co-induced by $\Phi$ and $\widehat{\alpha} \in U$. Then $\Phi^{-1}[U]$ is open in the metric topology on $\GF^\omega$ and $\beta \in \Phi^{-1}[U]$, for all $\beta \in \Phi^{-1}[\{ \widehat{\alpha} \}]$. Note that the latter set is finite. Hence, there is some $n \ge 0$ so that for all $\beta \in \Phi^{-1}[\{ \widehat{\alpha} \}]$, $\ball{\delta}{\beta}{2^{-n}} \subseteq \Phi^{-1}[U]$.

\begin{lem}
$\ball{\widehat{\delta}}{\widehat{\alpha}}{2^{-n}} \subseteq U$.
\end{lem}
\begin{proof}
Similarly to the preceding proof we consider the following cases.

\begin{ncase}
$\widehat{\alpha}_i \in \GF$, for all $i \ge 0$.
\end{ncase}
Let $\beta \in \Phi^{-1}[\{ \widehat{\alpha} \}]$, then $\widehat{\alpha} = \Phi(\beta) = \beta$ in this case, and hence $\ball{\widehat{\delta}}{\widehat{\alpha}}{2^{-n}} \subseteq \Phi[\ball{\delta}{\beta}{2^{-n}}] \subseteq U$.

\begin{ncase}
$\widehat{\alpha}_m = g_\bot$, for some $m \ge 0$.
\end{ncase}
Let $\beta, \widetilde{\beta} \in \GF^\omega$ such that $\beta_j = \widetilde{\beta}_j = \widehat{\alpha}_j$, for all $j \ge 0$ with $j \not= m$, $\beta_m = g_{-1}$, and $\widetilde{\beta}_j = g_1$. Then $\{ \beta, \widetilde{\beta} \} = \Phi^{-1}[\{ \widehat{\alpha} \}]$. Hence, $\ball{\delta}{\beta}{2^{-n}} \cup \ball{\delta}{\widetilde{\beta}}{2^{-n}} \subseteq \Phi^{-1}[U]$. 
Now, let $\widehat{\gamma} \in \ball{\widehat{\delta}}{\widehat{\alpha}}{2^{-n}}$ and $\gamma \in \Phi^{-1}[\{ \widehat{\gamma} \}]$. Then $\gamma^{< n+1} =  \beta^{< n+1}$ or $\gamma^{< n+1} =  \widetilde{\beta}^{< n+1}$, i.e, $\gamma \in \ball{\delta}{\beta}{2^{-n}} \cup \ball{\delta}{\widetilde{\beta}}{2^{-n}}$, from which we obtain that $\widehat{\gamma} \in U$. Thus, $\ball{\widehat{\delta}}{\widehat{\alpha}}{2^{-n}} \subseteq U$.
\end{proof}

\begin{prop}
The metric topology on $\widehat{G}$ is equivalent to the topology co-induced by $\Phi$.
\end{prop}

%% file: sec-indef-final-final.tex
Let $X$ be a set and $\PPP(X)$ its powerset. An operator $\fun{\Phi}{\PPP(X)}{\PPP(X)}$ is \emph{monotone} if for all $Y, Z \subseteq X$,
\[
\text{
if $Y \subseteq Z$, then $\Phi(Y) \subseteq \Phi(Z)$;
}
\]
and a set $Y \subseteq X$ is \emph{$\Phi$-closed} (or a pre-fixed point of $\Phi$) if $\Phi(Y) \subseteq Y$. Since $\PPP(X)$ is a complete lattice, every monotone operator $\Phi$ has a least fixed point $\mu\Phi\in\PPP(X)$ by the Knaster-Tarski Theorem. We often write
\[
P(x) \overset{\mu}{=} \Phi(P)(x),
\]
instead of $P = \mu\Phi$. $\mu\Phi$ can be defined to be the least $\Phi$-closed subset of $X$. Thus, we have the \emph{induction principle} stating that for every $Y \subseteq X$,
\[\text{
If $\Phi(Y) \subseteq Y$ then $\mu\Phi \subseteq Y$.
}\] 

Dual to inductive definitions are \emph{co-inductive definitions}. A subset $Y$ of $X$ is called \emph{$\Phi$-co-closed} (or a post-fixed point of $\Phi$) if $Y \subseteq \Phi(Y)$. By duality, every monotone $\Phi$ has a largest fixed point $\nu\Phi$ which can be defined as the largest $\Phi$-co-closed subset of $\Phi$. So, we have the \emph{co-induction principle} stating that for all $Y \subseteq X$,
\[\text{If $Y \subseteq \Phi(Y)$ then $Y \subseteq \nu\Phi$.}\]

Note that for $P \subseteq X$ we also write
\[
P(x) \overset{\nu}{=} \Phi(P)(x)
\]
instead of $P = \nu\Phi$. 

For monotone operators $\fun{\Phi, \Psi}{\PPP(X)}{\PPP(X)}$ define
\[
\Phi \subseteq \Psi \Def (\forall Y \subseteq X)\, \Phi(Y) \subseteq \Psi(Y).
\]
It is easy to see that the operation $\nu$ is monotone, i.e., if $\Phi \subseteq \Psi$, then $\nu\Phi \subseteq \nu\Psi$. This allows to derive the following strengthening of the co-induction principle.

\begin{lem}[\textbf{Strong Co-induction Principle~\cite{btifp}}]\label{lem-strong}\sloppy
Let $\fun{\Phi}{\PPP(X)}{\PPP(X)}$ be a monotone operator. Then: 
\[\text{
If $Y \subseteq \Phi(Y \cup \nu\Phi)$ then $Y \subseteq \nu\Phi$.
}\]
\end{lem}
The proof is dual to the proof of the strong induction principle in \cite{btifp,sp}.

\begin{lem}[\textbf{Generalised Half-strong Co-induction Principle}]\label{lem-halfstrong}
Let $\fun{\Phi',\Phi}{\PPP(X)}{\PPP(X)}$ be monotone operators
such that $\Phi'$ is absorbed by $\Phi$, that is,
$\Phi'(\Phi(Y)) \subseteq \Phi(Y)$ for all $Y\subseteq X$.
Then: 
\[\text{
If $Y \subseteq \Phi'(\Phi(Y) \cup \nu \Phi)$ then $Y \subseteq \nu \Phi$.
}\]
\end{lem}
Note: If $\Phi'$ is the identity, then this is the half-strong co-induction
principle from~\cite{btifp}. For a proof of that special case see~\cite{sp}.
We will specialise generalised half-strong co-induction to a
concurrent setting in Section~\ref{sec-concon} and use it in Section~\ref{sec-sdgc}.
\begin{proof}
Assume $Y \subseteq \Phi'(\Phi(Y) \cup \nu \Phi)$.
Since $\Phi(Y) \cup\nu\Phi \subseteq \Phi(Y\cup\nu\Phi)$
(by the monotonicity of $\Phi$), we have
$Y \subseteq \Phi'(\Phi(Y\cup \nu \Phi))$ (by the monotonicity of $\Phi'$),
and therefore $Y \subseteq \Phi(Y\cup \nu \Phi)$
since $\Phi'$ is absorbed by $\Phi$.
With strong co-induction, it follows $Y\subseteq\nu\Phi$.
\end{proof}

The following example is taken from~\cite{be}.

\begin{exa}[\textbf{Natural numbers}]\label{ex-nat}
Define $\fun{\Phi}{\PPP(\RR)}{\PPP(\RR)}$ by
\[
\Phi(Y): = \{ 0 \} \cup \set{y + 1}{y \in Y}.
\]
Then $\mu\Phi = \NN = \{\, 0, 1, \ldots \,\}$. The induction principle is logically equivalent to the usual zero-successor-induction on $\NN$; if $0 \in Y$ and $(\forall y \in Y) (y \in Y \to y+1 \in Y)$, then $(\forall y \in \NN)\, y \in Y$.
\end{exa}

\begin{exa}[\textbf{The set of non-empty finite subsets of a set}]\label{ex-finset}
Let $Y$ be a subset of a set $X$.
Define $\fun{\Phi_Y}{\PPP(\PPP(X))}{\PPP(\PPP(X))}$ by
\[
\Phi_Y(Z) \Def
  \set{u \in \PPP(X)}
      {(\exists x \in Y)\, u = \{ x \} \vee
         (\exists v \in Z) (\exists y \in Y)\, u = v \cup \{ y \}}
\]
and let $\bP_{\textbf{fin}}(Y) \Def \mu \Phi_Y$. 
Then $\bP_{\textbf{fin}}(Y)$ is the set of all non-empty finite subsets of $Y$.
\end{exa}

\begin{exa}[\textbf{Digit Spaces}]
Digit spaces can be characterised co-inductively. Define $\CC_{X} \subseteq X$ by
\[
\CC_{X}(x) \overset{\nu}{=} (\exists e \in E)(\exists y \in X)\, x = e(y) \wedge \CC_{X}(y), 
\]
i.e.\ $\CC_{X} = \nu \Phi_{X}$, where for $Z \subseteq X$,
\[
\Phi_{X}(Z) \Def \set{x \in X}{(\exists e \in E)(\exists y \in X)\, x = e(y) \wedge Z(y)}.
\]
Note here that we may consider subsets $A \subseteq X$ as unary predicates and write $A(x)$ instead of $x \in A$.

\begin{lemC}[\cite{sp}]\label{lem-digco}
Let $(X, E)$ be a digit space. Then $X = \CC_{X}$.
\end{lemC}

If all digits $e \in E$ are invertible, a slightly more comfortable characterisation can be given. Define $\CC'_{X} \subseteq X$ by
\[
\CC'_{X}(x) \overset{\nu}{=} (\exists e \in E)\, x \in \range(e) \wedge \CC'_{X}(e^{-1}(x)).
\]
\begin{lem}\label{lem-diginvco}
Let $(X, E)$ be a digit space with only invertible digits. Then $\CC'_{X} = \CC_{X}$. 
\end{lem}
\begin{proof} Both inclusions follow by co-induction. Let $x \in \CC'_{X}$. Then there exists $e \in E$ so that $x \in \range(e)$ and $\CC'_{X}(e^{-1}(x))$. It follows for $y = e^{-1}(x)$ that $x = e(y)$ and $\CC'_{X}(y)$, which shows that $\CC'_{X} \subseteq \Phi_{X}(\CC'_{X})$. Hence, $\CC'_{X} \subseteq \CC_{X}$.

Conversely, let $x \in \CC_{X}$. Then there are $e \in E$ and $y \in X$ with $x = e(y)$ and $\CC_{X}(y)$. It follows that $x \in \range(e)$ and $\CC_{X}(e^{-1}(x))$. Thus, $\CC_{X} \subseteq \CC'_{X}$.
\end{proof}

Classically the set $\CC_{X}$ is rather uninteresting, but constructively it is significant, since from a constructive proof that $x \in \CC_{X}$ one can extract a stream $\alpha$ of digits such that $x = \val{\alpha}$.

For what follows let $S \Def \CC_{(\II, \AV)}$ and $G \Def \CC_{(\II, \GF)}$. Then 
\begin{equation}\label{eq-sdrep}
S(x) \overset{\nu}{=} (\exists i \in \SD)\, \II(i, x) \wedge S(2x - i),
\end{equation}
where for $i \in \SD$ and $x \in \II$, $\II(i, x) \Def | 2 x - i | \le 1$, and
\begin{equation}\label{eq-grepprime}
G(x) \overset{\nu}{=} (\exists j \in \GC)\, x \in \range(g_{j}) \wedge G(1 - j \cdot 2x).
\end{equation}
Note that $\range(g_{-1}) = [-1, 0]$ and $\range(g_{1}) = [0, 1]$. Moreover, the functions $1 - 2(-x)$ and $1 - 2x$, respectively, form the left and the right branch of the \emph{tent function} $t(x) \Def 1 - 2|x|$. Hence the right-hand side in (\ref{eq-grepprime}) is equivalent to 
\[
((x < 0 \wedge j = -1) \vee (x > 0 \wedge j = 1) \vee x = 0) \wedge G(t(x)).
\]
However, the last disjunction is not decidable as the test for 0 is not computable. Since we want to work in a logic that allows extracting computable content from disjunctions, a (classically) equivalent formula of what we have just obtained is preferable 
\begin{equation}\label{eq-grep}
G\mytextcolor{red}{(x)} \overset{\nu}{=} (x \not= 0 \to x \le 0 \vee x \ge 0) \wedge G(t(x)).
\end{equation}
\end{exa}

\begin{exa}[\textbf{Well-founded induction}]
The principle of \emph{well-founded induction} is an induction principle for elements in the accessible or well-founded part of a binary relation $\prec$. As shown in~\cite{btifp}, it is an instance of strictly positive induction. The \emph{accessible part} of $\prec$ is inductively defined by
\[
\acc_{\prec}(x) \overset{\mu}{=} (\forall y \prec x)\, \acc_{\prec}(y),
\]
that is, $\acc_{\prec} = \mu \Phi$ where $\Phi(X) \Def  \set{x}{(\forall y \prec x)\, X(y)}$.
A predicate $P$ is called \emph{progressive} if $\Phi(P) \subseteq P$, that is, $\prog_{\prec}(P)$ holds where
\[
\prog_{\prec}(P) \Def (\forall x)((\forall y \prec x)\, P(y) \to P(x)).
\]
Therefore, the principle of well-founded induction, which states that a progressive predicate holds on the accessible part of $\prec$, is a direct instance of the rule of strictly positive induction:
\[
\dfrac{\prog_{\prec}(P)}{\acc_{\prec} \subseteq P}\,\, \text{(WFI$_{\prec}(P)$).}
\]
In most applications $P$ is of the form $A \to Q$. The progressivity of $P \to Q$ can equivalently be written as progressivity of $P$ relativised to $A$,
\[
\prog_{\prec, A}(P) \Def (\forall x \in A)((\forall y \in A)\, (y \prec x \to P(y)) \to P(x)).
\]
and the conclusion becomes $\acc_{\prec} \cap A \subseteq P$
\[
\dfrac{\prog_{\prec, A}(P)}{\acc_{\prec} \cap A \subseteq P}\,\, \text{(WFI$_{\prec, A}(P)$)}.
\]

Dually to the accessibility predicate one can define for a binary relation a path predicate
\[
\path_{\prec}(x) \overset{\nu}{=} (\exists y \prec x)\, \path_{\prec}(y),
\]
that is, $\path_{\prec} = \nu \Psi$ where $\Psi(X) \Def \set{x}{(\exists y \prec x)\, X(y)}$.
Intuitively, $\path_{\prec}(x)$ states that there is an infinite descending path $\ldots x_{2} \prec x_{1} \prec x$.

With the axiom of choice and classical logic it can be shown that $\neg \path_{\prec}(x)$ implies $\acc_{\prec}(x)$. 
\end{exa}

%% file: sec-progex-final-final.tex
In this section we recast the theory of digit spaces in a constructive setting with the aim to extract programs that provide effective representations of certain objects or transformations between different representations. As the main results on this basis we will obtain effective transformations between the signed digit and the Gray code representations of  $\II$ and the hyperspace of non-empty compact subsets of $\II$, respectively, showing that the two representations are effectively equivalent. The method of program extraction is based on a version of realisability, and the main constructive definition and proof principles will be induction and co-induction. The advantage of the constructive approach lies in the fact that proofs can be carried out in a representation-free way. Constructive logic and the Soundness Theorem automatically guarantee that proofs are witnessed by effective and provably correct transformations on the level of representations.

\subsection{The formal system IFP}
\label{sub-ifp}
As basis for program extraction from proofs we use \emph{Intuitionistic Fixed Point Logic} (IFP)~\cite{btifp}, which is an extension of many-sorted first-order logic by inductive and co-inductive definitions, i.e., predicates defined as least and greatest fixed points of strictly positive operators. Here, an occurrence of an expression $E$ is \emph{strictly positive (s.p.)} in an expression $F$ if that occurrence is not within the premise of an implication, and a predicate $P$ is \mytextcolor{red}{\emph{strictly positive in a predicate variable}} $X$ if every occurrence of $X$ in $P$ is strictly positive. 
Strict positivity is a simple and sufficiently general syntactic condition that ensures monotonicity and hence the existence of these fixed points, as discussed in Section~\ref{sec-indef}.

Relative to the language specified the following kinds of expression are defined:
\begin{description}

\item[{\it Formulas} $A, B$]  \emph{Equations} $s = t$ ($s, t$ terms of the same sort), $P(\vec t)$ ($P$ a predicate which is not an abstraction, $\vec t$ a tuple of terms whose sorts fit the arity of $P$), \emph{conjunction} $A \wedge B$, \emph{disjunction} $A \vee B$, \emph{implication} $A \to B$, \emph{universal} and \emph{existential quantification} $(\forall x)\, A, (\exists x)\, A$.

\item[{\it Predicates} $P, Q$]  \emph{Predicate variables} $X, Y, \ldots$ (each of fixed arity), \emph{predicate constants, abstraction} $\lambda \vec{x}.\, A$ (arity given by the variable tuple $\vec{x}$), $\mu \Phi, \nu \Phi$ (arities = arity of $\Phi$).

\item[{\it Operators} $\Phi$]  $\lambda X.\, P$ where $P$ must be strictly positive in $X$ and the arities of $X$ and $P$ must coincide. The arity of $\lambda X.\, P$ is this common arity.

\end{description}
\emph{Falsity} is defined as $\bfalse \Def \mu (\lambda X.\, X)()$ where $X$ is a predicate variable of arity $()$.

Program extraction is performed via a `uniform' realisability interpretation 
(Section~\ref{sub-realisability}). Uniformity concerns the interpretation of 
quantifiers: A formula $(\forall x) \, A(x)$ is realised uniformly by one 
object $a$ that realises $A(x)$ for all $x$, so $a$ may not depend on $x$. 
Dually, a formula $(\exists x)\, A(x)$ is realised uniformly by one 
object $a$ that realises $A(x)$ for some $x$, so $a$ does not contain a 
witness for $x$. 
Expressions (formulas, predicates, operators) that contain no disjunction
and no free predicate variables are identical to their realisability 
interpretations and are called \emph{non-compu\-tational (nc)}.
A slightly bigger class of expressions are \emph{Harrop expressions}.
These may contain disjunctions and free predicate variables but not at strictly
positive positions. A Harrop formula may not be identical to its realisability
interpretation, however they have at most one realiser which is trivial
and which is represented by the program constant $\bnil$ 
(see Sections~\ref{sub-prog} and \ref{sub-realisability}). 

We highlight some feature that distinguish IFP from other approaches to program 
extraction.

\begin{description}

\item[{\it Classical logic}] Although IFP is based on intuitionistic logic a 
fair amount of classical logic is available. 
Soundness of realisability holds in the presence of any non-computational 
axioms that are classically true. This can be extended to Harrop axioms whose
realisability interpretations (see~\ref{sub-realisability}) are classically true. 

\item[{\it Sets}]
We add for every sort $s$ a powersort $\PPP(s)$ and a 
(non-computational) element-hood relation constant $\eps$ of arity 
$(s,\PPP(s))$. 
In addition, for every Harrop formula $A(x)$ the comprehension axiom  
\[
(\exists u)(\forall x)(x \eps  u \leftrightarrow A(x))
\]
is added. ($A(x)$ may contain free variables other than $x$.) 
The realisability interpretation of such a comprehension axiom is again a 
comprehension axiom and can hence be accepted as true.
We will use the notation $\set{x}{A(x)}$ for the element $u$ of sort 
$\PPP(s)$ whose existence is postulated in the comprehension axiom above. 
Hence, we can define the empty set $\emptyset \Def \set{x}{\bfalse}$, singletons 
$\{ x \} \Def  \set{y}{y = x}$, the classical union of two sets 
$u \cup v \Def \set{x}{\neg(x \noteps u \land x \noteps v)}$, the union 
of all members of a set of sets 
$\bigcup u \Def \set{x}{(\exists y \eps u)\, x \eps y}$,
and the intersection of a class of sets defined by a predicate $P$
of arity $(\PPP(s))$,
$\bigcap P \Def \set{x}{(\forall y \in P)\, x \eps y}$.

Note that the informal notion of `set' used in Section~\ref{sec-indef} is
represented in the formal system IFP in three different ways:
\begin{enumerate}
\item Sorts are names for abstract `ground' sets.
For example, $s$ is a name for the abstract set of real numbers.
\item Terms of sort $\PPP(s)$ denote subsets of the ground set denoted
by $s$. Elements of sort $\PPP(s)$ can be defined by comprehension,
$\{x \mid A(x)\}$, which is restricted to nc formulas $A(x)$.
\item Predicates are expressions denoting subsets of the ground sets.
Predicates can be constructed by $\lambda$-abstraction, $\lambda x\,.\,A(x)$
(also written $\set{x}{A(x)}$),
where $A(x)$ can be any formula.
\end{enumerate}
By `set' we will mean in the following always (2), that is
`element of sort $\PPP(s)$'. The three concepts form an increasing hierachy
since the sort $s$ corresponds to the set
$\{x \mid \btrue\}$ and every set $u$ corresponds to the predicate
$\lambda x\,.\,x \eps u$.
Note that $x\,\eps u$ is an nc formula while $x\in P$ (which is synonym for
$P(x)$) has computational content if the predicate $P$ has.

To clarify the distinction we formally recast the definition of `the set of
finite subsets of a set' (Example~\ref{ex-finset}), which should now rather be
called `the predicate of finite subsets of a predicate': Let $P$ be a predicate
of arity $(s)$ ($s$ and $P$ correspond to $X$ and $Y$ in \ref{ex-finset}).
We define the predicate $\bP_{\textbf{fin}}(P)$ of arity $(\PPP(s))$ as
$\mu\,\Phi_P$ where the operator $\Phi_P$ of arity $(\PPP(s))$ is defined as
$
\Phi_P = \lambda Z \,.\,
  \lambda u \,.\,
      (\exists x \in P)\, u = \{ x \} \vee
         (\exists v \in Z) (\exists y \in P)\, u = v \cup \{ y \}
$.

\item[{\it Abstract real numbers}] In formalising the theory of real numbers, e.g., the set $\bR$ of real numbers is regarded as a sort $\iota$. A predicate $\bN$ with $\bN(x)$ if the real number $x$ is a natural number, is introduced by induction as in  Example~\ref{ex-nat}. All arithmetic constants and functions we wish to talk about are admitted as constant or function symbols. The predicates $=$, $<$ and $\le$ are considered as non-computational. As axioms, any true disjunction-free formulas about real numbers can be chosen. As such, the axiom system $\AAA_{R}$ consists of a discunction-free formulation of the axioms of real-closed fields, equations for exponentiation, the defining axiom for $\max$, stability of $=, \le, <$, as well as the Archimedean property $\mathbf{AP}$ about the non-existence of real numbers greater than all natural numbers, and Brouwer's Thesis for nc predicates 
\[
(\textbf{BT}_{\mathbf{nc}}) \quad (\forall x)\, (\neg \path_{\prec}(x) \to \acc_{\prec}(x)).
\]

\item[{\it Compact sets}]
In order to be able to deal with the hyperspace of non-empty compact subsets 
of the compact real interval $[-1,1]$, we also add a predicate constant  
$\bK$ of arity $(\PPP(\iota))$ to denote the elements of that hyperspace.
We also add an axiom for 
the finite intersection property stating that the intersection of the members of 
a descending sequence in $\bK$ is not empty.

\item[{\it Partial computation}] Like the majority of programming languages, IFP's language of extracted programs admits general recursion and therefore partial, i.e., non-terminating computation. 

\item[{\it Infinite computation}] Infinite data, as they naturally occur in exact real number computation, can be represented by infinite computations. This \mytextcolor{red}{is} achieved by an operational semantics where computations may continue forever outputting arbitrarily close approximations to the complete (infinite) result at their finite stages.

\end{description}

The proof rules of IFP include the usual natural deduction rules for intuitionistic first-order logic with equality. In addition there are the following rules for strictly positive induction and co-induction: the closure and co-closure of the least and greatest fixed point, respectively, stated as assumption-free rules, and the induction as well as the co-induction principle.

\subsection{Programs and their semantics}
\label{sub-prog}
Extracted programs, i.e.\ realisers, are interpreted as elements of a Scott domain $D$ defined by the recursive domain equation
\[
D = (\bnil + \bleft(D) + \bright(D) + \bpair(D \times D) + \bfun(D \to D))_{\bot},
\]
where $D \to D$ is the domain of continuous functions from $D$ to $D$, $+$  denotes the disjoint sum of partial orders, and $(\cdot)_{\bot}$ adds a new bottom element. $\bnil$, $\bleft$, $\bright$, $\bpair$ and $\bfun$ denote the injections of the various components of the sum into $D$. $\bnil$, $\bleft$, $\bright$, $\bpair$ (but not $\bfun$) are called \emph{constructors}.

$D$ carries a natural partial order $\sqsubseteq$ with respect to which it is a countably based \emph{Scott domain} (\emph{domain} for short), that is a bounded-complete algebraic directed-complete partial order with least element $\bot$ and a basis of countably many compact elements~\cite{dom}.
An element of $D$ is called \emph{defined} if it is different from $\bot$.
Hence, each defined element is of one of the forms
$\bnil$, $\bleft(\_)$, $\bright(\_)$,$\bpair(\_,\_)$, $\bfun(\_)$.

Since domains are closed under suprema of increasing chains $D$ contains not only finite but also infinite combinations of the constructors. For example, writing $a : b$ for $Pair(a, b)$, an infinite sequence of domain elements $(d_i)_{i \in  \bN}$ is represented in $D$ as the stream
\[
d_{0} : d_{1} :  \ldots \Def \sup_{n \in \bN} \bpair(d_{0}, \bpair(d_{1}, \ldots, \bpair(d_{n}, \bot) \ldots)).
\]
Because Scott domains and continuous functions form a \mytextcolor{red}{Cartesian} closed category, $D$ can be equipped with the structure of a partial combinatory algebra (PCA, \cite{dom}) by defining a continuous application operation $a\, b$ such that $a\, b \Def  f(b)$, if $a = \bfun(f)$, and $a\, b \Def \bot$, otherwise, as well as combinators $K$ and $S$ satisfying $K\, a\, b = b$ and $S\, a\, b\, c = a\, c\, (b\, c)$ (where application associates to the left). In particular $D$ has a continuous least fixed point operator which can be defined by Curry's $Y$-combinator or as the mapping $(D \to D) \ni f \mapsto \sup_{n} f^{n}(\bot) \in D$.

Besides the PCA structure the algebraicity of $D$ will be used, that is, the fact that every element of $D$ is the directed supremum of compact elements. Compact elements have a strongly finite character. The finiteness of compact element is captured by their defining property, saying that $d \in D$ is compact if for every directed set $A \subseteq D$, if $d \sqsubseteq \bigsqcup A$, then $d \sqsubseteq a$ for some $a \in A$, and the existence of a function assigning to every compact element $a$ a \emph{rank}, $ \brk(a) \in \NN$, satisfying
\begin{description}

\item[rk1] If $a$ has the form $C(a_{1}, \ldots, a_{k})$ for a data constructor $C$, then $a_{1}, \ldots, a_{k}$
are compact and $\brk(a) > \brk(a_{i})$, for $1 \le i \le k$.

\item[rk2] If $a$ has the form $\bfun(f)$, then for every $b \in D$, $f(b)$ is compact with $\brk(a) > \brk(f(b))$ and there exists a compact $b_{0} \sqsubseteq b$ such that $\brk(a) > \brk(b_{0})$ and $f(b_{0}) = f(b)$. Moreover, there are finitely many compact elements $b_{1}, \ldots, b_{n}$ with $\brk(b_{i}) < \brk(a)$ such that $f(b) = \bigsqcup \set{f(b_{i})}{1 \le i \le n \wedge b_{i} \sqsubseteq b}$.

\end{description}

Elements of $D$ are denoted by programs which are defined as in \cite{btifp}
except that the case construct is more general since it allows overlapping patterns. 
For example, it is now possible
to define the function parallel-or~\cite{pl}. Setting $\btrue = \bleft(\bnil)$,
$\bfalse = \bright(\bnil)$ parallel-or can be defined as
\begin{align*}
\lambda c.\, \ccase{c}{\,
       & \bpair(\btrue,\_) \to \btrue;\\ 
       & \bpair(\_,\btrue) \to \btrue;\\
       & \bpair(\bfalse,\bfalse) \to \bfalse\,}
\end{align*}
which is not possible in the programming language defined in~\cite{btifp}.
We will need this greater expressivity in Section~\ref{sec-concon}.

Formally, \emph{Programs} are terms $M, N, \ldots$ of a new sort $\delta$
built up as follows:
\[
\begin{array}{rcl}
\mathit{Programs} \ni M, N, L, R & ::= & a, b 
           \quad\text{(program variables)}                        \\[.5ex]
 & | & \bnil\, |\, \bleft(M)\, |\, \bright(M)\, |\, \bpair(M, N)  \\[.5ex]
 & | & \ccase{M}{Cl_1;\ldots;Cl_n}               \\[.5ex]
 & |     & \lambda a.\, M                                         \\[.5ex]
 & |     & M N                                                    \\[.5ex]
 & |     & \brec M                                                \\[.5ex]
 & |     & \bot
\end{array}
\]
where in the case-construct the $Cl_i$ are pairwise compatible clauses 
(see \mytextcolor{red}{Definition}~\ref{dn-parcomp} below).
A \emph{clause} is an expression of the form $P \to N$ where $P$ is a pattern 
and $N$ is a program.
A \emph{pattern} is either a constructor pattern or a function pattern.
A \emph{constructor pattern} is a program built from constructors and variables
such that each variable occurs at most once. 
\emph{Function patterns} are of the form $\pfun(a)$ where $a$ is a program variable.

\begin{defi}\label{dn-parcomp}
Two clauses, $P_1 \to N_1$ and $P_2 \to N_2$, 
are \emph{compatible} if for any substitutions $\theta_1$, $\theta_2$, 
if $P_1\theta_1\alphaeq P_2\theta_2$, 
then $N_1\theta_1 \alphaeq N_2\theta_2$ where $\alphaeq$ means $\alpha$-equality,
that is, equality up to renaming of bound variables.
\end{defi}
Compatibility of clauses can be decided efficiently since it is enough to consider
most general unifiers $\theta_1$ and $\theta_2$.

The variables in $P$ are considered as binders. Hence, the free variables of a 
clause $P\to N$ are the free variables of $N$ that do not occur in $P$. 

In~\cite{btifp} only simple patterns containing one occurrence of one constructor
are considered and two clauses are required to have different constructors.
This is equivalent to allowing arbitrary pattern but requiring different clauses
to have non-unifiable pattern.

Programs that are $\alpha$-equal will be identified. 
Moreover, we will write $a \overset{\text{rec}}{=} M$ for 
$a \Def \mathbf{rec}(\lambda a.\, M)$, and $a\, b \overset{\brec}{=} M$ for 
$a \overset{\brec}{=} \lambda b.\, M$. 

\begin{defi}\label{dn-patmach} \hfill % moved first item to new line
\begin{enumerate}
\item\label{dn-patmach-1}
A program $M$ \emph{matches a constructor pattern} $P$ if there is a substitution $\theta$,
called the \emph{matching substitution}, such that $\dom(\theta) = \FV(P)$
and $P\theta = M$. 
\item\label{dn-patmach-2}
A program $M$ \emph{matches a function pattern} $\pfun(a)$ if $M$ is a $\lambda$-abstraction
and in this case the matching substitution is $[a \mapsto M]$.
\item\label{dn-patmach-3}
A program \emph{matches a clause} $P \to N$ if it matches $P$.
\end{enumerate}
\end{defi}

Except for the case-construct, the denotational semantics of programs in $D$
is defined as in \cite{btifp}.
To define the denotation $\ccase{M}{\vec{Cl}}$ we first define
when a domain element $d$ \emph{matches} a pattern $P$ and, if it
does, the \emph{matching environment} which has as domain the
variables of the pattern. 
\begin{itemize}
\item
In the case of a constructor pattern $P$
this is obvious and the matching environment $\eta$ (if it exists)
will satisfy $\val{P}\eta=d$.  

\item
The matches of a function pattern
$\pfun(a)$ are the domain elements of the form $\bfun(f)$ and the
matching environment is $[a \mapsto\bfun(f)]$.
\end{itemize}
The denotation of a case program in an environment $\eta$, 
$\val{\ccase{M}{\vec{Cl}}}\eta$, is defined as follows: 

\begin{defi}\label{dn-semcase} \hfill % moved first item to new line
\begin{enumerate}
\item\label{dn-semcase-1}
If $P\to N$ is a clause in $\vec{Cl}$ such
that $\val{M}\eta$ matches $P$ with matching environment $\eta'$, then 
$\val{\ccase{M}{\vec{Cl}}}\eta = \val{N}(\eta+\eta')$ 
where $\eta+\eta'$ is the environment obtained by
overriding $\eta$ with $\eta'$. 

\item\label{dn-semcase-2}
If no such matching is possible, 
then $\val{\ccase{M}{\vec{Cl}}}\eta=\bot$. 
\end{enumerate}
\end{defi}

Due to the compatibility condition the denotation is independent of the choice of
the matching clause. This follows from the fact that two patterns
$P_1$, $P_2$ are unifiable if and only if they have a common match and
the most general unifiers are in a one-to-one correspondence with the
matching environments of the common match.

\begin{defi}\label{dn-val}
A program is called a \emph{value} if it is an abstraction
or begins with a constructor. 
\end{defi}
Note that a closed program is a value exactly if it is a 
weak head normal form (whnf). 
Clearly, if $M$ is a value, then $\val{M}\eta\neq\bot$ for every 
environment $\eta$.

The following small-step operational semantics of closed programs is similar 
to the one in~\cite{btifp}. The difference is due to the more general 
case expressions. 
\renewcommand{\labelenumi}{\roman{enumi}.}
\begin{enumerate}
\item $\ccase{M}{\ldots;P\to N;\ldots}\ssp N\theta$\,
if $M$ matches $P$ with matching substitution $\theta$. 
\item $(\lambda x.\,M)\ N \ssp M[N/x]$.
\item $\brec\,M \ssp M\,(\brec\,M)$.
\item\AxiomC{$M \ssp M'$}
             \UnaryInfC{$\ccase{M}{\vec{Cl}}\ssp 
               \ccase{M'}{\vec{Cl}}$}
            \DisplayProof\,
 if $M$ doesn't match any clause in $\vec{Cl}$. 
\item \AxiomC{$M \ssp M'$}
            \UnaryInfC{$M\,N \ssp M'\,N$}                    
            \DisplayProof\,
if $M$ is not an abstraction. \\
\item
\AxiomC{$M_i \ssp M_i'$ $(i = 1,\ldots, k)$} 
\UnaryInfC{$C(M_1,\ldots,M_k) \ssp C(M_1',\ldots,,M_k')$}
\DisplayProof. \\
\item$\lambda x.\,M \ssp \lambda x.\,M$.
\end{enumerate}
\renewcommand{\labelenumi}{\arabic{enumi}.}

\begin{lemC}[\cite{btifp}]\label{lem:ssp}
Let $M$ be a closed program.
\begin{enumerate}
\item $M \ssp M'$ for exactly one $M'$. 
\item If $M \ssp M'$, then $\val{M} = \val{M'}$. 
\item 
  $\val{M} \ne \bot$ exactly if there is a hnf $V$ such that $M \ssp^* V$.
\end{enumerate}
\end{lemC}
\begin{proof}
  (1) holds by the compatibility condition for case-constructs.
(2) is easy.
  The proof of (3) is as the proof of \cite[Lemma~33]{btifp} 
  for the case that $M$ begins with a constructor, and an easy consequence of \cite[Lemma~32]{btifp}
    for the case that $M$ is a $\lambda$-abstraction. 
\end{proof}

\subsection{Types}
\label{sub-typ}
A type $\tau(E)$ is assigned to every IFP-formula and predicate $E$, 
where types are expressions defined by the grammar
\[
\textit{Types} \ni \rho, \sigma ::= \alpha\,\text{(type variables)}\, |\, \mathbf{1}\, |\, \rho \times \sigma\, |\,\rho + \sigma\, |\, \rho \Rightarrow \sigma\, |\, \bfix \alpha.\, \rho
\]
where in $\bfix \alpha.\, \rho$ the type $\rho$ must be strictly positive in 
$\alpha$. Types are interpreted by subdomains of $D$ in an obvious way.

The idea is that for a formula $A$, $\tau(A)$ is the type of potential realisers. 
Expressions without computational content will receive type {\bf 1}.

Intuitively, by saying that a program $a$ is a \emph{realiser} of a formula $A$, one means that $a$ is a computational content of formula $A$. In intuitionistic logic, a proof of $A \vee B$ gives us the evidence that A is true or B is true. The notion of realiser used in \mytextcolor{red}{the present} paper is designed by treating this as the primitive source of computational content. Therefore, we defined an expression \emph{non-computational (nc)} if it contains neither disjunctions nor free predicate variables. A more general notion of an expression with trivial computational content is provided by the Harrop property. A formula is \emph{Harrop} if it contains neither disjunctions nor free predicate variables at strictly positive positions. A predicate $P$ is \emph{$X$-Harrop}, if $P$ is strictly positive in $X$ and $P[\hat{X}/X]$ is Harrop for $\hat{X}$ a predicate constant associated with $X$. 
\begin{longtable}{lclr}
$\tau(P(\vec{t}))	$	&=&	$\tau(P)$		&										\\[.8ex]
$\tau(A \wedge B)$		&=& $\tau(A) \times \tau(B)$		&$\text{($A, B$ non-Harrop)}$			\\
				&=& $\tau(A)$				&$\text{($B$ Harrop)}$					\\
				&=& $\tau(B)	$			&$\text{(otherwise)}	$				\\[.8ex]
$\tau(A \vee B)$		&=& $\tau(A) + \tau(B)$			&								\\[.8ex]
$\tau(A \to B)$		&=& $\tau(A) \Rightarrow \tau(B)$ &$\text{($A, B$ non-Harrop)}$			\\
				&=& $\tau(B)	$			&$\text{($A$ Harrop)}$					\\[.8ex]
$\tau(\Diamond x\, A)$	&=& $\tau(A)$				&$\text{($\Diamond \in \{ \forall, \exists \}$)}$	\\[.8ex]		$\tau(X)$			&=& $\alpha_{X}$				&$\text{($X$ a predicate variable)}$		\\[.8ex]		
$\tau(P)$			&=& $\mathbf{1}$					&\text{($P$ a predicate constant)}		\\[.8ex]
$\tau(\lambda \vec{x}.\, A) $ &=& $\tau(A)	$		&								\\[.8ex]		$\tau(\Box (\lambda X.\, P))$	 &=& $\bfix \alpha_{X}.\, \tau(P)$	&$\text{($\Box \in \{ \mu, \nu \}$, $P$ not $X$-Harrop)}$ \\
					 &=& $\mathbf{1} $			&$\text{($\Box \in \{ \mu, \nu \}$, $P$ $X$-Harrop)}$	
\end{longtable}
For example, $\tau(\NN)=\bnat \Def \bfix\alpha.\,\mathbf{1}+\alpha$, 
the type of unary natural numbers.

\subsection{Realisability}
\label{sub-realisability}
Next, we define the notion that a program $a : \tau(A)$ is a \emph{realiser} of a formula $A$. In order to formalise this notion and to provide a formal proof of its soundness, Berger and Tsuiki~\cite{btifp} introduced an extension RIFP of IFP which in addition to the sorts of IFP contains the sort $\delta$, denoting the domain $D$. For each IFP formula $A$ they define an RIFP predicate $\bR(A)$ of arity $(\delta)$ that specifies the set of domain elements that realise $A$. Similarly,  for every non-Harrop predicate $P$ of arity $(\vec\sigma)$ a predicate $\bR(P)$ of arity $(\vec\sigma, \delta)$, and every non-Harrop operator $\Phi$ of arity $(\vec\sigma)$ an operator $\bR(\Phi)$ of arity $(\vec\sigma, \delta)$ is defined. Note that instead of $\bR(A)(a)$ we also write $a \br A$. Moreover, we write $\br A$ to mean $(\exists a)\, a \br A$.

Simultaneously, $\bH(B)$ is defined, for Harrop formulas $B$, which expresses that $B$ is realisable, however with trivial computational content $\bnil$. More precisely, we define a formula $\bH(A)$ for every Harrop formula $A$, a predicate $\bH(P)$ for every Harrop predicate $P$, and an operator $\bH(\Phi)$ for every Harrop operator $\Phi$. $\bH(P)$ and $\bH(\Phi)$, respectively, will be of the same arity as $P$ and $\Phi$.

\begin{longtable}{RCLR}
a \br A 			&= & (a = \bnil \wedge \bH(A)) 			&\text{($A$ Harrop)}  \\[.5ex]
a \br P(\vec t)		&= & \bR(P)(\vec t, a)	&\text{($P$ non-H.)}	\\[.5ex]
c \br (A \wedge B)	&= & (\exists a, b)\, (c = \bpair(a, b) \wedge a \br A \wedge b \br B) &\text{($A, B$ non-H.)} \\[.5ex]
a \br (A \wedge B)	&= & a \br A \wedge \bH(B) 			&\text{($B$ Harrop, $A$ non-H.)}	\\[.5ex]
b \br (A \wedge B)	&= & \bH(A) \wedge b \br B			&\text{($A$ Harrop, $B$ non-H.)}	\\[.5ex]
c \br (A \vee B)		&= & 
\multicolumn{2}{L}{
(\exists a)\, (c = \bleft(a) \wedge a \br A) \vee (\exists b)\, (c = \bright(b) \wedge b \br B)
}	\\[.5ex]
c \br (A \to B) 		&= & c : \tau(A) \Rightarrow \tau(B) \wedge (\forall a)\, (a \br A \to (c\, a) \br B)	&\text{($A, B$ non-H.)}	\\[.5ex]
b \br (A \to B)		&= & b : \tau(B) \wedge (\bH(A) \to b \br B)  &\text{($A$ Harrop, $B$ non-H.)}		\\[.5ex]
a \br \Diamond x\, A	&= & \Diamond x\, (a \br A)  &\text{($\Diamond \in \{ \forall, \exists \}$, $A$ non-H.)}	\\[.8ex]
\bR(X)			&= & \tilde{X}			&								\\[.5ex]
\bR(\lambda \vec x.\, A)	&= & \lambda (\vec x, a)\, (a \br A) 	&\text{($A$ non-H.)}		\\[.5ex]
\bR(\Box (\Phi))		&= & \Box (\bR(\Phi))		&\text{($\Box \in \{ \mu, \nu \}$, $\Phi$ non-H.)}	\\[.5ex]
\bR(\lambda X.\, P)	&= & \lambda \tilde{X}.\, \bR(P)		&\text{($P$ non-H.)}			\\[1.5ex]
\bH(P(\vec{t}))		&= & \bH(P)(\vec{t})					&\text{($P$ Harrop)}			\\[.5ex]
\bH(A \wedge B)	&= & \bH(A) \wedge \bH(B)		&\text{($A, B$ Harrop)}		\\[.5ex]
\bH(A \to B)		&= & \br A \to \bH(B)				&\text{($B$ Harrop)}			\\[.5ex]
\bH(\Diamond x\, A)	&= & \Diamond x\, \bH(A)	&\text{($\Diamond \in \{ \forall, \exists \}$, $A$ Harrop)}	\\[.8ex]
\bH(P)			&= & P						&\text{($P$ a predicate constant)} 	\\[.5ex]
\bH(\lambda \vec{x}.\, A)   &= & \lambda \vec{x}.\, \bH(A)	&\text{($A$ Harrop)}				\\[.5ex]
\bH(\Box(\Phi))		&= & \Box(\bH(\Phi))				&\text{($\Box \in \{ \mu, \nu \}$, $\Phi$ Harrop)}  \\[.5ex]
\bH(\lambda X.\, P)	&= & \lambda X.\, \bH_{X}(P)			&\text{($P$ $X$-Harrop)} 		
\end{longtable}
\noindent
For the last line recall that a predicate $P$ is \emph{$X$-Harrop}, if $P$ is strictly positive in $X$ and $P[\hat{X}/X]$ is Harrop for $\hat{X}$ a predicate constant associated with $X$. In this situation $\bH_{X}(P)$ stands for $\bH(P[\hat{X}/X])[X/\hat{X}]$. The idea is that $\bH_{X}(P)$ is the same as $\bH(P)$ but considering $X$ as a (non-computational) predicate constant.

\begin{lemC}[\cite{btifp}]\label{lem-harreal}
\mbox{}
\begin{enumerate}
\item\label{lem-harreal-1} If $A$ is Harrop, then $\bH(A) \leftrightarrow \br A$.

\item\label{lem-harreal-2} If $E$ is an nc expression, then $\bH(E) = E$, in particular $\bH(\bfalse) = \bfalse$.

\end{enumerate}
\end{lemC}

\begin{exa}[\textbf{Realiser of induction and co-induction}]\label{ex-ind-co}
Set 
\begin{gather*}
f \circ g \Def \lambda a\,.\,f(g\,a),\\
\hbox{$[f+g]$}
   \Def \lambda c\,.\,\ccase{c}{\bleft(a) \to f\,a; \bright(b) \to g\,b}.
\end{gather*}

Note that if $f:\rho\to\sigma$ and $g:\sigma\to\sigma'$, then
$g\circ f:\rho\to\sigma'$,
and if $f_1 : \rho_1\to\sigma$ and $f_2 : \rho_2\to\sigma$, then
$[f_1+f_2] : (\rho_1+\rho_2) \to \sigma$.

Note also that for a s.p.\ non-Harrop operator $\Phi$,
and a non-Harrop predicate $P$.
\[\tau(\mu(\Phi))=\tau(\nu(\Phi)) = 
\bfix\tau(\Phi) \Def \bfix\,\alpha\,.\,\tau(\Phi)(\alpha)\,\]

For every s.p.\ type operator $\varphi$ let 
$\mon_\varphi : (\alpha\to\beta)\to\varphi(\alpha)\to\varphi(\beta)$
be the canonical program such that 
for every s.p.\ operator $\Phi$ and all predicates $P,Q$ (of fitting arity),
$\mon_{\tau(\Phi)}$ realizes $(P\subseteq Q) \to \Phi(P)\subseteq\Phi(Q)$.
Then $\mon_{\varphi}$ is a polymorphic program whose type depends on the type variables
$\alpha,\beta$. These type variables may be substituted by any types $\rho,\sigma$.
We sometimes write $\mon^{\rho,\sigma}_\varphi$ to indicate that we are interested
in the typing obtained by this substitution, that is,
$\mon^{\rho,\sigma}_\varphi: (\rho\to\sigma)\to\varphi(\rho)\to\varphi(\sigma)$.
A similar convention applies to the polymorphic programs defined below, such as
$\iter_\varphi$, $\coiter_\varphi$, etc., as well as to the polymorphic constructors
$\bleft^{\alpha,\beta}:\alpha\to(\alpha+\beta)$,
$\bright^{\alpha,\beta}:\beta\to(\alpha+\beta)$, and the identity function 
$\id^\alpha:\alpha\to\alpha$.
Of course these programs do not depend on the superscripts but only 
on the subscript (if any).

To improve readability we will in the following omit the `$\tau$' from $\tau(A)$,
$\tau(P)$ and $\tau(\Phi)$ and we write $\nu\Phi$ instead of
$\nu(\Phi)$, etc. Hence for example, instead of writing
\[
\mon^{\tau(P),\tau(\nu(\Phi))}_{\tau(\Phi)}:
(\tau(P)\to\tau(\nu(\Phi)))\to\tau(\Phi)(\tau(P))\to\tau(\Phi)(\tau(\nu(\Phi)))
\]
we write
\[
\mon^{P,\nu\Phi}_{\Phi}:(P\to\nu\Phi)\to\Phi(P)\to\Phi(\nu\Phi)
\]
or even
\[
\mon^{P,\nu\Phi}_{\Phi}:(P\to\nu\Phi)\to\Phi(P)\to\nu\Phi
\]
since $\tau(\Phi)(\tau(\nu\Phi))\equiv\tau(\nu\Phi)$.

\paragraph{Induction.}

If $s : \Phi(P) \to P$ realises $\Phi(P) \subseteq P$, then 
$\iter^P_{\Phi}\,s$ realizes $\mu\Phi\subseteq P$ where 
\begin{gather*}
\iter_{\varphi} : (\varphi(\alpha) \to \alpha) \to \bfix(\varphi)\to\alpha,\\
\iter_{\varphi}\,s \Def 
     \brec\,\lambda f.\,s \circ \mon^{\bfix(\varphi),\alpha}_{\varphi}\, f.
\end{gather*}

\paragraph{Co-induction.}

If $s : P\to \Phi(P)$ realises $P \subseteq \Phi(P)$, then 
$\coiter^P_{\Phi}\,s$ realises $P\subseteq \nu\Phi$ where
\begin{gather*}
\coiter_{\varphi} :(\alpha\to \varphi(\alpha))\to\alpha\to\bfix\varphi,\\
\coiter_\varphi\,s \Def 
       \brec\,\lambda f.\,\mon^{\alpha,\bfix\varphi}_\varphi\, f \circ s.
\end{gather*}

\paragraph{Half-strong co-induction.}

If $s : P\to (\Phi(P) + \nu\Phi)$ 
realises $P \subseteq \Phi(P) \cup \nu\Phi$, then 
$\hscoiter^P_{\Phi}\,s$
realises $P\subseteq \nu\Phi$ where
\begin{gather*}
\hscoiter_{\varphi} :(\alpha\to (\varphi(\alpha)+\bfix\varphi))\to
                                                  \alpha\to\bfix\varphi\\
\hscoiter_\varphi\,s \Def \brec\,
     \lambda f.\,
    [\mon^{\alpha,\bfix\varphi}_\varphi\, f + \id^{\bfix\varphi}] \circ s
\end{gather*}

\paragraph{Strong co-induction.}

If $s : P\to (\Phi(P)+\nu\Phi)$  
realises $P \subseteq \Phi(P \cup \nu(\Phi))$, then 
$\scoiter^P_{\Phi}\,s$
realises $P\subseteq \nu\Phi$ where
\begin{gather*}
\scoiter_{\varphi} :(\alpha\to \varphi(\alpha+\bfix\varphi))\to
                                                  \alpha\to\bfix\varphi,\\
\scoiter_\varphi\,s \Def \brec\,
     \lambda f.\,
    \mon^{\alpha+\bfix\varphi,\bfix\varphi}_\varphi\, [f + \id^{\bfix\varphi}] \circ s.
\end{gather*}

\paragraph{Generalised half-strong co-induction.}

Assume $\Phi'$ is monotone.

Let $\absorb^{\alpha}_{\Phi',\Phi} : \Phi'(\Phi(\alpha)) \to \Phi(\alpha)$ 
realise $\Phi'(\Phi(Y))\subseteq \Phi(Y)$ for all $Y$. 

If $s : P\to \Phi'(\Phi(P) + \nu\Phi)$ 
realises $P \subseteq \Phi'(\Phi(P)\cup\nu(\Phi))$, then
\[\hscoiter_{\Phi}\,
   (\absorb_{\Phi',\Phi} \circ \mon_{\Phi'} 
           [\mon_{\Phi}\,\bleft +
            \mon_{\Phi}\,\bright] \circ s)\] 
realises $P \subseteq\nu\Phi$. This means that the realiser is a function 
$f:P\to\nu\Phi$ defined recursively by
\begin{align*}
  f \overset{\brec}{=} &\mon^{P+\nu\Phi,\nu\Phi}_{\Phi}\, [f + \id^{\nu\Phi}] \\
     &\circ\ \absorb^{P+\nu\Phi}_{\Phi',\Phi} \\ 
     &\circ\ \mon^{\Phi(P)+\nu\Phi,\Phi(P+\nu\Phi)}_{\Phi'} 
             [\mon^{P,P+\nu\Phi}_{\Phi}\,\bleft^{P,\nu\Phi} + 
               \mon^{\nu\Phi,P+\nu\Phi}_{\Phi}\,\bright^{P,\nu\Phi}] \\
     &\circ\ s. 
\end{align*}

\end{exa}

\begin{exa}[\textbf{Realiser of well-founded induction}]\label{ex-wfirec}
The schema of well-founded induction, $\wfi_{\prec, A}$(P), is realised as follows: If $s$ realises $\prog_{\prec, A}$ where $P$ is non-Harrop, then $\acc_{\prec} \cap A \subseteq P$ is realised by
\begin{itemize}
\item 
$\tilde{f} \overset{\text{rec}}{=} \lambda a.\, (s\, a\, (\lambda a'.\,\lambda b.\, \tilde{f}\, a'))$ if $\prec$ and $A$ are both non-Harrop,

\item 
$\tilde{f} \overset{\text{rec}}{=} \lambda a.\, (s\, a\, \tilde{f})$ if $\prec$ is Harrop and $A$ is non-Harrop,

\item
$\tilde{c} \overset{\text{rec}}{=} s\,(\lambda b.\, \tilde{c})$ if $\prec$ is non-Harrop and $A$ is Harrop,

\item
$\brec s$ if $\prec$ and $A$ are both Harrop.

\end{itemize}
See \cite[Lemma 21]{btifp} for a proof.
\end{exa}

 \subsection{Soundness}
\label{sub-soundness}

The Soundness Theorem~\cite{btifp} stating that provable formulas are realisable is the theoretical foundation for program extraction.
\begin{thm}[\textbf{Soundness}]\label{thm-sound} % added bold text for consistency
Let $\AAA$ be a set of nc axioms. From an \emph{IFP($\AAA$)} proof of formula $A$ one can extract a program $M : \tau(A)$ such that $M \br A$ is provable in \emph{RIFP($\AAA$)}.

More generally, let $\Gamma$ be a set of Harrop formulas and $\Delta$ a set of non-Harrop formulas. Then, from an \emph{IFP($\AAA$)} proof of a formula $A$ from the assumptions $\Gamma, \Delta$ one can extract a program $M$ with $\emph{FV($M$)} \subseteq \vec{u}$ such that $\vec{u} : \tau(\Delta) \vdash M : \tau(A)$ and $M \br A$ are provable in \emph{RIFP($\AAA$)} from the assumptions $\bH(\Gamma)$ and $\vec{u} \br \Delta$.
\end{thm}

If one wants to apply this theorem to obtain a program realising formula $A$ one must provide terms $K_1, \ldots, K_n$ realising the assumptions in $\Delta$. Then it follows that the term $M(K_1, \ldots, K_n)$ realises $A$, provably in RIFP. Because the program axioms of RIFP given in \cite{btifp} are correct with respect to the denotational semantics, a further consequence is that $M(K_{1}, \ldots, K_{n})$ is a correct realiser of $A$.

That realisers do actually \emph{compute} witnesses is shown in~\cite{btifp} by two \emph{Computational Adequacy Theorems} that relate the denotational definition of realisability with a lazy operational semantics.

%% file: sec-concon-final-final.tex
Non-termination is a natural and fundamental phenomenon in computation. 
It is denotationally modelled in domain theory~\cite{dom} and a logical account 
of it is Scott's logic with an existence predicate~\cite{elogic}. 
The Minlog system~\cite{minlog}
supports the extraction of programs that may or may not terminate and
keeps control of potential partiality through a logic with totality degrees.
A limitation of programs extracted from proofs in Minlog, or other systems 
such as Coq~\cite{coq}, is that they are sequential. Having the
possibility of running computations concurrently, on the other hand,
can be very useful to get around partiality. If, e.g., $M$ and $N$ are
two programs known to realise formula $A$ under the assumption that
condition $B$ or $\neg B$ holds, respectively, then at least one of
them is guaranteed to terminate. So running them concurrently and
picking the result obtained first, will lead to a result realising
$A$, provided that if $M$ or $N$ terminates then it realises $A$.

To capture realisability restricted to a condition as described above,
we follow the approach in~\cite{bt} and extend in Section~\ref{sub-rest} IFP 
by a propositional connective 

\[ 
\rt{B}{A}  \qquad\hbox{(``$A$ \emph{restricted} to $B$'')} 
\]
which, for nc-formulas $B$, has a similar meaning as the formula $B \to A$ but
behaves slightly differently (and better) with respect to~realisability: While a realiser of
$B \to A$ is a program that realises $A$ if $B$ holds 
but otherwise provides no guarantees,
a realiser of $\rt{B}{A}$ is a program $p$ that terminates and realises $A$ if $B$ holds,
but even if $B$ does not hold, $p$ will realise $A$ provided $p$ terminates.
In order to behave well, the formation of $\rt{B}{A}$ is restricted to formulas $A$
satisfying a syntactic condition called productivity (defined in Section~\ref{sub-rest})
that guarantees that only terminating programs can realise $A$.

In Section~\ref{sub-conc} we introduce the concurrency modality $\ddown(A)$ 
from~\cite{bt} with the crucial rule
\[\dfrac{\rt{B}{A} \quad \rt{\neg B}{A}}{\ddown(A)}\,\, \text{($\ddown$-lem)} \]
that makes precise the above intuition. We also prove the realisability of a couple
of further rules that say how $\ddown$ interacts with other logical connectives.

Finally, in Section~\ref{sub-monad-conc},
we introduce a new concurrency modality,
$\itdown(A)$, which inherits most of the properties of $\ddown(A)$
but in addition has realisable rules corresponding to a monad.
This new modality will be used 
in Section~\ref{sec-sdgc} to define concurrent versions of the signed digit 
representation and infinite Gray code which are constructively equivalent.

\subsection{Restriction $\rt{B}{A}$}
\label{sub-rest}
Following~\cite{bt}, we introduce restriction, $\rt{B}{A}$, where $A$ is required 
to be \emph{productive}\footnote{\mytextcolor{red}{Observe that in \cite{bt} the notion ``strict'' is used instead of ``productive''.}}, that is,
every implication and restriction in $A$ has to be part of a Harrop formula
or a disjunction. In particular, Harrop formulas and disjunctions are
always productive. 
The reason why $A$ is required to be productive is that this ensures that $A$
has only defined realisers, that is $\bot$ does not realise $A$.

The definition of the Harrop property is extended by demanding
that Harrop formulas must not contain a restriction at a strictly
positive position. In particular, restrictions are not
Harrop. Realisability for restrictions is defined as
\[
a \br \rt{B}{A} \Def a : \tau(A) \wedge
(\br B \to a \neq \bot) \wedge
(a \neq\bot \to a \br A).
\]
Note that if $A$ is a Harrop formula, then $a \br \rt{B}{A}$ is equivalent to the formula 
\[
(\br B \vee a \not= \bot) \to (a = \bnil \wedge \bH(A)).
\]
The type of restriction is $\tau(\rt{B}{A}) \Def \tau(A)$.

To gain some intuition suppose that a closed program $M$ realises
$\rt{B}{A}$.  Since closed programs denote a value different from $\bot$
exactly if they reduce to whnf, one has:
(i) If $B$ is realisable, then $M$ reduces to whnf.
(ii) If $M$ reduces to whnf, then $M$ realises $A$ (even if $B$ is not
realisable). In this sense, one has partial correctness of $M$ with
respect to the specification $A$. The distinction between
$\rt{B}{A}$ and $B \to A$ regarding realisability is carefully
discussed in \cite{bt}.

Sometimes, we need to get rid of the productivity requirement. This can be achieved by considering formulas of kind $A \lor \bfalse$ instead of just $A$. By definition, $A \lor \bfalse$ is always productive. Moreover, $a \br A$, exactly if $\bleft(a) \br A \lor \bfalse$. Note here that $\bfalse$ has no realiser. This leads us to the following unrestricted version of the restriction connective
\[
\rtu{B}{A} \Def \rt{B}{(A \lor \bfalse)}.
\]

Since the realisers of such formulas are more complicated than those in the productive case, we keep both versions of the restriction connective. It should be clear from the definition that all statements in this paper about the realisability of rules for the restriction connective $\rest$ also hold for the unrestricted version $\restu$.

The following derivation rules concerning restriction are added to IFP:

\begin{longtable}{c|c}
\multicolumn{2}{c}{
$\dfrac{B \to A_{0} \vee A_{1} \quad \neg B \to A_{0} \wedge A_{1}}
      {\rt{B}{(A_{0} \vee A_{1})}}\,\,
            \text{($\rest$-intro ($A_{0}, A_{1}, B$ Harrop))}$} \\[3ex] \hline \\
$\dfrac{A}{\rt{B}{A}}$\,\, \text{($\rest$-return)} &
       $\dfrac{\rt{B}{A} \quad A \to (\rt{B}{A'})}{\rt{B}{A'}}$\,\, \text{($\rest$-bind)} \\[3ex]
$\dfrac{\rt{B}{A} \quad B' \to B}{\rt{B'}{A}}$\,\, \text{($\rest$-antimon)} &
       $\dfrac{\rt{B}{A} \quad B}{A}$\,\, \text{($\rest$-mp)} \\[3ex]
$\dfrac{}{\rt{\bfalse}{A}}$\,\, \text{($\rest$-efq)} &
       $\dfrac{\rt{B}{A}}{\rt{\neg\neg B}{A}}$\,\, \text{($\rest$-stab)} \\[3ex] 
 $\dfrac{(\rtu{C}{\rt{B}{A})}}{\rt{B \land C}{A}}$\, \text{($\rest$-absorb)} &
$\dfrac{\rt{B}{A} \quad \rt{B}{C}}{\rt{B}{(A \land C)}}$\, \text{($\rest$-$\land$)}. 
\end{longtable}

\begin{lemC}[\cite{bt}]\label{lem-restrule}
The rules for restriction are realisable, provably in extended \emph{RIFP}. Hence, Soundness Theorem~\ref{thm-sound} remains valid for the extension of \emph{IFP} by restriction, however, classical logic is needed to derive the correctness of realisers.
\end{lemC}

Note that the last two rules have not been considered in \cite{bt}. As is easily verified, Rule~(\emph{$\rest$-absorb}) is realised by $\lambda a.\, \ccase{a}{\bleft(a') \to a'}$
and Rule~(\emph{$\rest$-$\land$}) by $(\bpair\low)\low$, where $\stc{f}{a}$ denotes \emph{strict application}:
\[
\stc{f}{a} \Def \ccase{a}{C(\_) \to f\, a \mid
                             C \in \{ \bnil, \bleft, \bright, \bpair, \pfun \}}.
\]
Observe that $\stc{f}{a} = f\, a$ if $a \not= \bot$ and  $\stc{f}{\bot} = \bot$. 

\begin{lemC}[\cite{bt}]\label{lem-restrul}
The following rule is derivable from the rules for restriction:
\[
\dfrac{\rt{B}{A} \quad A \to A'}{\rt{B}{A'}}\,\, \emph{($\rest$-mon)}.
\]
\end{lemC}

The rule is realised by $\lambda f.\, \lambda a.\, \stc{f}{a}$. To see this assume that $a \br \rt{B}{A}$ and $f \br (A \to A')$. We have to show that $\stc{f}{a} \br \rt{B}{A'}$. Suppose first that $\br B$. Then $a \neq \bot$ and hence by definition of $f\low$, $\stc{f}{a} = f\,a$. We need that $f\,a \neq \bot$. Here, the productivity requirement for the restriction connective comes into play: $A'$ needs to be productive and since $f\,a \br A'$, we have that $f\,a \neq \bot$, as was to be shown. Next, suppose that $\stc{f}{a} \neq \bot$. Then $a \neq \bot$ as well, by definition of $f\low$. Therefore, $a \br A$. It follows that $f\,a \br A'$. Since $a \neq \bot$, we moreover have that $\stc{f}{a} = f\,a$.  Thus, $\stc{f}{a} \br A'$.

\subsection{McCarthy's Amb and the concurrency modality $\ddown(A)$}
\label{sub-conc}

To deal with concurrency, \cite{bt} introduced a
further constructor $\amb(a, b)$ indicating that its arguments $a, b$
need be evaluated concurrently in order to obtain one of the results
even if the other one is not terminating. The domain $D$ now has to
satisfy the domain equation
\[
D =
(\bnil + \bleft(D) + \bright(D) + \bpair(D \times D) + \bfun(D \to D) + \amb(D\times D))_{\bot}.
\]
The programming language is extended by a constructor $\amb$ which denotes the constructor $\amb$
in the domain $D$.
Hence, denotationally, the constructor $\amb$ is an exact copy of $\bpair$, that is,
it acts like a lazy pairing operator.
Only the operational semantics interprets a program $\amb(M, N)$ as a concurrent 
computation of $M$ and $N$ until one of them is reduced to whnf.
This is formalised by the following (non-deterministic) relation $\sspc$ (`c' for `choice'):
\renewcommand{\labelenumi}{c\roman{enumi}.}
\begin{enumerate}
\item \AxiomC{$M \ssp M'$}
             \UnaryInfC{$M \sspc M'$}
            \DisplayProof,
\item $\amb(M_1,M_2) \sspc M_i$\, if $M_i$ is a whnf  ($i=1,2$),
\item
\AxiomC{$M_i \sspc M_i'$ $(i = 1,\ldots, k)$} 
\UnaryInfC{$C(M_1,\ldots,M_k) \sspc C(M_1',\ldots,,M_k')$}
\DisplayProof\,
 if $C\neq \amb$.
\end{enumerate}
\renewcommand{\labelenumi}{\arabic{enumi}.}
The deterministic relation $\ssp$ which $\sspc$ extends is now to be understood with respect
to all constructors, including $\amb$.
The intuition of $\sspc$ is that a program is first deterministically evaluated using $\ssp$. 
If a program of the form $\amb(M_1,M_2)$ is obtained, deterministic computation continuous in 
parallel with $M_1$ and $M_2$. 
As soon as one of the two programs reach a whnf, the other may be discarded using Rule~(cii). 
Rule (ciii) says that the computation can be carried out inside (nested) 
data constructors.
Note that rule (cii) cannot be applied to proper subterms of a term $\amb(M_1,M_2)$
since the only way of reducing $\amb(M_1,M_2)$ with a rule other than (cii) is by Rule~(ci) and 
Rule~(vi) of $\ssp$.
This ensures that at any point only two concurrent threads are needed to carry out the computation.
For the reductions to yield the desired result (see~\cite{bt}, Theorems~1 and~2), fairness conditions 
must be imposed. For example, if in $\amb(M_1,M_2)$ at least one of the $M_i$ is hnf, then
Rule~(cii) will eventually be applied. 
In~\cite{bt} slighty more general (and more complicated) but essentially equivalent rules are 
given which facilitate the formalisation of the fairness condition and allow parallel threads 
to be evaluated at differend `speeds'.

To indicate at the logical level that a realiser of a formula $A$ may be computed concurrently, 
a new modality $\ddown(A)$ was introduced in \cite{bt} where, again, the formula
$A$ is required to be productive. 
For a non-Harrop formula $A$ realisability is defined as
\begin{align*}
c \br \ddown(A) \Def\, &c = \amb(a,b) \wedge a, b : \tau(A) \wedge \mbox{} \\
&(a \neq\bot \vee b \neq \bot) \wedge \mbox{} \\
&(a \neq\bot \to a \br A) \wedge (b \neq\bot \to b \br A).
\end{align*}
Thus, a realiser of $\ddown(A)$ is a pair of candidate realisers $a$
and $b$ at least one of which denotes a defined value
and all of $a$ and $b$ which are defined values are
correct realisers. In particular, if $a$ and $b$ are both defined,
then they are \emph{both} correct realisers. Therefore, by running the
two programs for the candidates $a$ and $b$ concurrently and taking
the one which becomes defined (i.e.~a whnf) first, it is guaranteed that
we obtain a correct result. Hence, we can safely stop the 
other process. 

The occurrence of $A$ in $\ddown(A)$ is regarded strictly
positive. The definition of the Harrop property is further extended by demanding
that Harrop formulas must not contain the concurrency operator at a strictly
positive position. In particular, $\ddown(A)$ is always non-Harrop 
(even if $A$ is Harrop). 
The type of the concurrency modality is $\tau(\ddown(A)) = \bfA(\tau(A))$
where $\bfA$ is a new type operator.

IFP is once more extended by adding the following derivation rules. The logical system thus obtained is called \emph{Concurrent Fixed Point Logic (CFP)}.

\begin{longtable}{c|c}
$\dfrac{\rt{C}{A} \quad \rt{\neg C}{A}}{\ddown(A)}\,\, \text{($\ddown$-lem)} $
&$ \dfrac{A}{\ddown(A)}\,\, \text{($\ddown$-return)}$ \\[3ex]
$\dfrac{\ddown(A) \quad A \to B}{\ddown(B)}\,\, \text{($\ddown$-mon)}$
&$\dfrac{\ddown(A)}{A}\,\, \text{($\ddown$-H ($A$ Harrop))}$ \\[3ex] 
$ \rule{0mm}{8mm} \dfrac{\ddown(\rt{B}{A})}{\ddown(A)}\,\, \text{($\ddown$-$\rest$-absorb)}$
& $ \dfrac{\rt{B}{A} \quad \rt{D}{C}}{\rt{B \lor D}{\ddown(A \lor C)}}\,\, \text{($\ddown$-$\rest$-$\lor$)}$
\\[5mm] \hline
\multicolumn{2}{c}{$\rule{0mm}{8mm}
\dfrac{B \to \ddown(A_{0} \vee A_{1}) \quad \neg B \to A_{0} \wedge A_{1}}
      {\rt{B}{\ddown(A_{0} \vee A_{1})}}\,\, 
\text{($\ddown$-$\rest$-intro ($A_{0}, A_{1}, B$ Harrop))}$}
\end{longtable}

Note that the last three rules have not been considered in~\cite{bt}.

\begin{lem}\label{lem-concrule}
The rules for the concurrency modality are realisable.
\end{lem}
\begin{proof} 
The realisability of the first four rules has been shown in~\cite{bt}. It remains to consider the last three rules. 

It is easy to see that (\emph{$\ddown$-$\rest$-absorb}) is realised by the identity function:
Assume $c \br \ddown(\rt{B}{A})$. 
Then $c = \amb(a, b)$. Furthermore, $a \neq \bot$ or $b \neq\bot$, and 
in the first case $a \br \rt{B}{A}$ while in the second case $b \br \rt{B}{A}$.
We show that $c \br \ddown(A)$. 
By the facts we know about $c$, it suffices to show that if $a \neq \bot$ then $a \br A$ 
(and similarly for $b$).  But if $a \neq \bot$, then $a \br \rt{B}{A}$ and hence $a \br A$.

Next, we show that Rule~($\ddown$-$\rest$-$\lor$) is realised by
\[
g \Def \lambda a.\, \lambda b.\, \amblr(\bleft\low a, \bright\low b).
\]
where $\amblr(u,v) \Def \amb(u,v)$ if $u=\bleft(\_)$ or $v = \bright(\_)$ and $\Def \bot$ otherwise.
$\amblr$ can be easily defined using the case construct.

Suppose that $a \br \rt{B}{A}$ and $b \br \rt{D}{C}$.
We have to verify that
\begin{enumerate}
\item\label{gr1} $\br (B \lor D) \to g\,a\,b \ne \bot$.

\item\label{gr2} $g\, a\, b \ne \bot \to (g\, a\, b) \br \ddown(A \lor C)$.

\end{enumerate}

(\ref{gr1}) Assume $B \lor D$ is realisable. Then $B$ or $D$ is realisable. 
Without restriction assume $\br B$. Then $a \ne \bot$ and hence 
$\bleft\low a = \bleft(a)$ and $g\,a\,b =\amb(\bleft(a), \bright\low b) \ne \bot$.  

(\ref{gr2}) If $g\, a\, b \ne \bot$, then
$g\, a\, b = \amb(u,v)$ where $u=\bleft\low a$ and $v = \bright\low b$,
and $u\ne\bot$ or $v\ne\bot$. 
Furthermore, if $u\ne\bot$, then $u = \bleft(a)$ where $a\ne\bot$.
Hence $a\br A$ and therefore $u \br (A\lor C)$.
With a similar argument one sees that if $v \ne \bot$, then 
$v \br (A\lor C)$.

Rule~(\emph{$\ddown$-$\rest$-intro}), finally, is realised by
\[
h \Def 
 \lambda c.\, 
   \ccase{c}{\amb(C(\_),\_) \to c; \amb(\_,C(\_)) \to c \mid C \in \{\bleft,\bright\}}.
\]
Observe the overlapping clauses. First we note that for
$c : \bfA(\tau(A_0 \lor A_1))$:
\begin{description}
\item[\rm (*)] If $c$ is of the form $\amb(a,b)$, where
$a\neq\bot$ or $b\neq\bot$, then $h\,c = c$. This is the case, in particular, if $c$ realises
$\ddown(A_{0} \vee A_{1})$.
\item[\rm (**)] If $h\,c\neq\bot$, then $h\,c = c$ and $c$ is of the form $\amb(a,b)$
where $a\neq\bot$ or $b\neq\bot$.
\end{description}
Now, assume $c$ realises 
$B \to \ddown(A_{0} \vee A_{1})$, that is, $c : \bfA(\tau(A_0 \lor A_1))$ and 
$\bH(B) \to c \br (\ddown(A_{0} \vee A_{1}))$, 
and that $\bH(\neg B \to A_{0} \wedge A_{1})$ holds, that is, 
$\neg\bH(B) \to \bH(A_{0}) \wedge \bH(A_{1})$.
We show that $h\,c$ realises $\rt{B}{\ddown(A_{0} \vee A_{1})}$:

First, assume $\bH(B)$. Then $c$ realises $\ddown(A_{0} \vee A_{1})$. 
Hence, by (*), $h\,c = c \neq\bot$.

Next, suppose $h\,c\neq\bot$. Then, by (**), $h\,c = c$ and 
$c$ is of the form $\amb(a,b)$ where $a\neq\bot$ or $b\neq\bot$.
Therefore, it suffices to show that $c$ realises $\ddown(A_{0} \vee A_{1})$. 
We do a classical case analysis on $\bH(B)$. 
If $\bH(B)$ holds, then $c$ realises $\ddown(A_{0} \vee A_{1})$.
If $\bH(B)$ does not hold, then $\bH(A_0)$ and $\bH(A_1)$ hold. 
To prove that $c$ realises $\ddown(A_{0} \vee A_{1})$ it suffices to show
that whenever $a$ or $b$ are defined, then they realise $A_0\lor A_1$.
Since $c : \bfA(\tau(A_0 \lor A_1))$, we have $a,b\in\tau(A_0\lor A_1)$.
But, since $\bH(A_0)$ and $\bH(A_1)$ hold, every defined element
in $\tau(A_0\lor A_1)$ realises $A_0\lor A_1$. 
\end{proof}

\extra{
We summarise the realisers obtained by displaying the rules above with their realisers.
We restrict Rule~($\ddown$-mon) to the most interesting cases where $A$ and $B$ are
both non-Harrop. In this case we need for the rule the program
\[
\mamb \Def \lambda f.\, \lambda c.\, \ccase{c}{\amb(a, b) \to \amb(\stc{f}{a}, \stc{f}{b})}.
\]

\begin{longtable}{c|c}
$\dfrac{a\br \rt{C}{A} \quad b \br \rt{\neg C}{A}}{\amb(a,b) \br \ddown(A)}\,\,
                             \text{($\ddown$-lem)} $
&$ \dfrac{a\br A}{\amb(a,\bot) \br \ddown(A)}\,\, \text{($\ddown$-return)}$ \\[3ex]
$\dfrac{c \br \ddown(A) \quad f \br (A \to B)}{(\mamb f\,c) \br \ddown(B)}\,\,
                             \text{($\ddown$-mon, $A, B$ non-Harrop)}$
&$\dfrac{\bH(\ddown(A))}{\bH(A)}\,\, \text{($\ddown$-H ($A$ Harrop))}$ \\[3ex]
$\dfrac{a \br \rt{B}{A} \quad b \br \rt{D}{C}}{(g\, a\, b) \br \rt{B \lor D}{\ddown(A \lor C)}}\,\, \text{($\ddown$-$\rest$-$\lor$)} $
& $\dfrac{c \br \ddown(\rt{B}{A})}{c \br \ddown(A)}\,\,
              \text{($\ddown$-$\rest$-absorb)}$   \\[3ex] \hline
\multicolumn{2}{c}{$\rule{0mm}{8mm}
\dfrac{\bH(B) \to c \br \ddown(A_{0} \vee A_{1})\quad\neg \bH(B)\to\bH(A_{0})\wedge\bH(A_{1})}
      {(h\,c) \br \rt{B}{\ddown(A_{0} \vee A_{1})}}\,\, 
\text{($\ddown$-$\rest$-intro ($A_{0}, A_{1}, B$ Harrop))}$
}
\end{longtable}
where $g, h$ are as defined in the proof.
}

\begin{lem}\label{lem-concder}
The following rules are derivable in \emph{CFP}:
\[
\begin{array}{c|c}
\dfrac{\neg\neg (A \vee B) \quad \rt{A}{C} \quad \rt{B}{C}}{\ddown(C)}\,\, \emph{($\ddown$-$\vee$-elim)}
& \dfrac{\neg\neg (A \vee B) \quad \rt{A}{C} \quad \rt{B}{D}}{\ddown(C \vee D)}\,\, \emph{($\ddown$-$\vee$-elim-or)} \\[3ex]
\dfrac{\ddown(A \vee B)}{\ddown(\ddown(A) \vee \ddown(B))}\,\, \emph{($\ddown$-$\vee$-dist)}
&\dfrac{\ddown(A \wedge B)}{\ddown(A) \wedge \ddown(B)}\,\, \emph{($\ddown$-$\wedge$-dist).} 
\end{array}
\]
\end{lem}
\begin{proof}
Rules~(\emph{$\ddown$-$\vee$-elim}) and (\emph{$\ddown$-$\vee$-elim-or}) have already been considered in \cite{bt}.
The Rules~(\emph{$\ddown$-$\vee$-dist}) and (\emph{$\ddown$-$\wedge$-dist}) are
easy consequences of the Rules~(\emph{$\ddown$-return}) and (\emph{$\ddown$-mon}).
\end{proof}

Let us again display the rules with their realisers.
\begin{longtable}{c}
$\dfrac{\neg\bH(\neg (A \vee B)) \quad a \br \rt{A}{C} \quad b \br \rt{B}{C}}
      {\amb(a,b)\br \ddown(C)}\,\,  \emph{($\ddown$-$\vee$-elim)}$ \\[3ex]
$\dfrac{\neg\bH(\neg (A \vee B)) \quad a \br \rt{A}{C} \quad b \br \rt{B}{D}}
      {\amb(\stc{\bleft}{a},\stc{\bright}{b})\br \ddown(C \vee D)}\,\,
                                    \emph{($\ddown$-$\vee$-elim-or)}$ \\[3ex]
$\dfrac{c \br \ddown(A \vee B)}
      {(\mamb g\,c) \br \ddown(\ddown(A) \vee \ddown(B))}\,\,
                                   \emph{($\ddown$-$\vee$-dist)}$\\[2ex]
\multicolumn{1}{l}{
\text{where $g\, d \Def \ccase{d}{\bleft(a) \to \bleft(\amb(a,\bot)); 
                         \bright(b) \to \bright(\amb(b,\bot))}$}
                         } \\[1ex]
$\dfrac{c \br \ddown(A \wedge B)}
      {\bpair(\mamb g_1\,c ,\mamb g_2\,c) \br (\ddown(A) \wedge \ddown(B))}\,\, 
                                   \emph{($\ddown$-$\wedge$-dist)} $\\[2ex]
\multicolumn{1}{l}{
\text{where $g_i\,d \Def \ccase{d}{\bpair(a_1,a_2)\to\amb(a_i,\bot)}$.}
				}	
\end{longtable}

\subsection{The monadic concurrency modality $\itdown(A)$}
\label{sub-monad-conc}

It is easy to see that the concurrency modality $\ddown$ is not a monad.
The monadic lifting law $(A \to \ddown(B)) \to (\ddown(A) \to \ddown(B))$
is in general not realisable.
In order to turn it into a monad we use its finite iterative closure
\[
\itdown(A) \overset{\mu}{=} \ddown(A \vee \itdown(A)).
\]
Note that $\itdown(A)$ is defined for arbitrary formulas $A$ (not only productive ones)
since in its definition $\ddown$ is applied to a disjunction.
As follows from the definition, we have for $c : \delta$,
\begin{align*}
c \br \itdown(A)         \overset{\mu}{=} \,\,
 & c = \amb(a, b) \wedge a, b : \tau(A \vee \itdown(A)) \wedge 
 (a \neq \bot \lor b \neq \bot) \land \mbox{}\\      
 & (a \neq \bot \to (a =\bleft(a') \wedge a' \br A) \vee 
                (a = \bright(a'') \wedge a'' \br \itdown(A))) \land \mbox{} \\
&  (b \neq \bot \to
 (b =\bleft(b') \wedge b' \br A) \vee 
               (b = \bright(b'') \wedge b'' \br \itdown(A))). 
\end{align*}

As we will see next, in case of the iterated concurrency modality 
the following analogue of Rule~(\emph{$\ddown$--$\rest$-intro}) modality is realisable. 
Again $A_{0}, A_{1}, B$ are required to be Harrop:
\[
\dfrac{B \to \itdown(A_{0} \vee A_{1}) \quad \neg B \to A_{0} \wedge A_{1}}
      {\rt{B}{\itdown(A_{0} \vee A_{1})}}\,\, 
\text{($\itdown$-$\rest$-intro)}.
\]
We need the following functions $\sqcup_{\bnil}, \triangleright : D \times D \to D$  and $\fun{f^{*}, h^{+}, g}{D}{D}$:
\begin{gather*}
a \sqcup_{\bnil} b \Def \ccase{\bpair(a, b)}{\bpair(\bnil, \_)\to\bnil;
                                            \bpair(\_, \bnil) \to\bnil},\\[.5ex]
d \triangleright a \Def \ccase{d}{\bnil \to a},
\end{gather*}
\begin{align*}
f^*d \overset{\brec}{=} \ccase{d}{&\bleft(\bleft(\bnil)) \to \bnil; \\
&\bleft(\bright(\bnil)) \to \bnil; \\
&\bright(\amb(u,v)) \to f^* u \sqcup_{\bnil} f^* v},
\end{align*}
\begin{align*}
h^+ c \overset{\brec}{=} \ccase{c}{&\amb(a,b) \to \\
&\hspace{.5cm} \ccase{f^* a \sqcup_{\bnil} f^* b}{\bnil \to \\
&\hspace{2.5cm} \amb(f^* a\triangleright g\, a, f^* b \triangleright g\, b)}\\
          &}, 
 \end{align*}
 \begin{equation*}
g\, d \Def \ccase{d}{\bleft(\_) \to d;
              \bright(e) \to \bright(h^+ e)}.
\end{equation*}

\begin{lem}\label{lem-auxitintro}
Assume $A = A_0 \lor A_1$ where $A_0$ and $A_1$ are Harrop formulas.
\begin{enumerate}

\item\label{lem-auxitintro-1}
 $(\forall a,b)\, (\amb(a,b) \br \itdown(A) \to f^* a=\bnil \lor f^* b = \bnil)$.
 
\item\label{lem-auxitintro-2} $(\forall c)\,(c \br \itdown(A) \to (h^+ c) \br \itdown(A))$.
\item\label{lem-auxitintro-3}
If $\bH(A_0)$ and $\bH(A_1)$, 
then $(h^+ c) \br \itdown(A_0 \lor A_1)$, for all $c$ such that $h^+ c$
is of the form $\amb(\_,\_)$.

\end{enumerate}
\end{lem}
\begin{proof}
(\ref{lem-auxitintro-1})
We use s.p.~induction.
If $\amb(a,b) \br \itdown(A)$, then $a \br (A \lor \itdown(A))$,
or $b \br (A \lor \itdown(A))$.
Without restriction assume the former.
If $a = \bleft(d)$ where $d \br A$, then $f^* a =\bnil$. 
If $a = \bright(\amb(u,v))$ where $\amb(u,v) \br \itdown(A)$, then, 
by the induction hypothesis, $f^* u =\bnil$ or $f^* v =\bnil$. 
Hence $f^* a=\bnil$.

(\ref{lem-auxitintro-2}) Again, we use s.p.~induction.
Assume $c \br \itdown(A)$.
Then, by (\ref{lem-auxitintro-1}),  $c = \amb(a,b)$ with
$f^* a \sqcup_{\bnil} f^* b = \bnil$ and 
$h^+ c = \amb(a', b')$ with $a' = f^* a\triangleright g\, a$ and
$b' = f^* b \triangleright g\, b$. 
Since $f^* a \sqcup_{\bnil} f^* b = \bnil$, 
$a' = g\, a \neq\bot$ or $b' = g\, b \neq\bot$, as required.
It remains to show that every defined element of $\{a',b'\}$ realises
$A\lor\itdown(A)$. 
Without restriction let $a'$ be defined. Hence $a'=g\, a$
and either (i) $g\, a = a =\bleft(p)$ and thus $a'=\bleft(p)$
with $p\in\{\bleft(\bnil),\bright(\bnil)\}$; 
or (ii) $a = \bright(q)$ and $g\, a = \bright(h^+ q)$, which means that
$a'=\bright(h^+ q)$.
Since $c \br \itdown(A)$ it follows that 
$a \br (A\lor\itdown(A))$. 
In Case~(i), $a=a'$ and we are done since $a$ is defined and hence realises
$A\lor\itdown(A)$. 
In Case~(ii), $q \br \itdown(A)$. Therefore,
by the induction hypothesis, $(h^+ q )\br \itdown(A)$
and hence $a' \br (A\lor\itdown(A))$.

(\ref{lem-auxitintro-3})  Assume $\bH(A_0)$ and $\bH(A_1)$.
We prove the required formula first for compact $c$ only and will later show 
that this is enough.
Hence we show first
\[(\forall\hbox{ compact }c)\,(h^+ c = \amb(\_,\_) \to (h^+ c) \br \itdown(A)).\]
We prove this by induction on the rank of $c$. 
The proof is in large parts similar to the proof of (\ref{lem-auxitintro-2}).
Let $c$ be compact and assume  $h^+ c = \amb(a',b')$.
Then $c = \amb(a,b)$ with
$f^* a \sqcup_{\bnil} f^* b = \bnil$, 
 $a' = f^* a \triangleright g\, a$ and
$b' = f^* b \triangleright g\, b$. 
Since $f^* a \sqcup_{\bnil} f^* b = \bnil$, it follows 
as in \mytextcolor{red}{the proof of Statement} (\ref{lem-auxitintro-2}) that
$a' = g\, a \not= \bot$, or 
$b' = g\, b \not= \bot$.
In either case, one of $a', b'$ is defined, as required.
It remains to show that every defined element of $\{a',b'\}$ realises
$A\lor\itdown(A)$. 
Without restriction let $a'$ be defined. Hence $a'=g\, a$
and either (i) $g\, a = a =\bleft(p)$, \mytextcolor{red}{and} thus $a'=\bleft(p)$
with $p\in\{\bleft(\bnil),\bright(\bnil)\}$; 
or (ii) $a = \bright(q)$ and $g\, a = \bright(h^+ q)$, which means that 
 $a'=\bright(h^+ q)$.
Since $\bH(A_{0})$ and $\bH(A_{1})$, 
$\bleft(\bnil)$ and $\bright(\bnil)$ both realise $A$. 
Hence, in Case~(i), $a' \br (A\lor\itdown(A))$.
In Case~(ii), since $f^* a'=\bnil$, $q$ must be of the 
form $\amb(\_,\_)$. Therefore, since $q$ has smaller rank than $c$, 
by the induction hypothesis, $(h^+ q) \br \itdown(A)$
and hence $a' \br (A\lor\itdown(A))$.

To remove the restriction to compact $c$, it
suffices to show that for every $c$ there is a compact $c_0 \sqsubseteq c$
such that $h^+ c_0 = h^+ c$. 
If $h^+ c=\bot$, then we can choose $c_0 = \bot$.
Otherwise, $c=\amb(a,b)$. 
Since $f^*$ is continuous and its range contains only compact elements 
(namely $\bot$ and $\bnil$), there are compact $a_0\sqsubseteq a$
and $b_0\sqsubseteq b$ such that $f^* a_0=f^* a$ and $f^* b_0=f^* b$.
Therefore, it suffices to show
\begin{multline*}
(\forall \hbox{ compact } a_0,b_0)\ (\forall a,b)\,
    (a_0\sqsubseteq a \land b_0\sqsubseteq b 
      \land f^*(a_0)=f^* a\land f^* b_0 = f^* b \\
\to h^+ \amb(a_0,b_0) = h^+ \amb(a,b).
\end{multline*}
This can be easily shown by induction on the maximum of the ranks of 
$a_0$ and $b_0$.
\end{proof}

Now, we are able to derive the result on Rule~(\emph{$\itdown$-$\rest$-intro})
we are aiming for.
\begin{lem}\label{lem-itdownintro}
Rule~\emph{($\itdown$-$\rest$-intro)}
is realised by $h^+$.
\end{lem}
\begin{proof}
Set $A \Def A_0\lor A_1$, and assume that $c$ realises $B \to \itdown(A)$
and $\bH(\neg B \to A)$ holds. 
The former means that $\bH(B) \to c \br \itdown(A)$ and
the latter that $\neg \bH(B) \to \bH(A_{0}) \wedge \bH(A_{1})$. 
We have to show that $h^+ c$ realises $\rt{B}{\itdown(A)}$.

First,  assume $\bH(B)$. Then $c \br \itdown(A)$ and, by 
Lemma~\ref{lem-auxitintro}(\ref{lem-auxitintro-1}),
$h^+ c$ realises $\itdown(A)$ and is therefore defined.

 Next, suppose $h^* c$ is defined. 
Hence $c = \amb(a,b)$
and $h^+ c = \amb(a',b')$ with $a' = f^* a\triangleright g\, a$
and $b'=f^* b \triangleright g(b)$.
We have to show $(h^+ c) \br \itdown(A)$.
Since $\itdown(A)$ has the same realisers as $\ddown(A \lor \itdown(A))$,
it is sufficient to derive that $(h^+ c) \br \ddown(A \lor \itdown(A))$.
That is, it suffices to verify that every defined element of $\{a',b'\}$
realises $A \lor \itdown(A)$. Without restriction assume that $a'$ is defined.
We do a classical case analysis on $\bH(B)$.
If $\bH(B)$, then $c \br \itdown(A)$. Therefore,
by Lemma \ref{lem-auxitintro}(\ref{lem-auxitintro-2}),
$(h^+ c) \br \itdown(A)$ and hence $a' \br  A \lor\itdown(A)$.
If $\neg\bH(B)$, then $\bH(A_0)$ and $\bH(A_1)$.
By Lemma \ref{lem-auxitintro}(\ref{lem-auxitintro-3}) we thus have that,
$(h^+ c )\br \itdown(A)$ and with the same argument as above we obtain
$a' \br  A \lor\itdown(A)$.
\end{proof}

Similarly, an analogue of the well known introduction rule for the connective $\land$ is realisable:
\[
\dfrac{\itdown(A) \quad \itdown(B)}{\itdown(A \land B)}\,\, \text{($\itdown$-$\land$-intro)}.
\]
Define
\begin{equation*}
f^{*} d \overset{\brec}{=} \ccase{d}{\bleft(\_) \to \bnil; \bright(\amb(u,v)) \to f^{*} u \sqcup_{\bnil} f^{*} v},
\end{equation*}
\begin{equation*}
\begin{aligned}
g\, d\, e \Def \ccase{d}{&\bleft(d') \to \ccase{e}{\bleft(e') \to \bleft(\bpair(d',e'));  \\
                                     & \hspace{3.74cm}                   \bright(e'') \to \bright(h_{2}\, d\, e'') };  \\
                                     &\bright(d'') \to \bright(h_{1}\, d'' e)}, 
\end{aligned}
\end{equation*}
\begin{multline*}                                  
h_{1}\, u\, v \overset{\brec}{=} \ccase{u}{\amb(a,b) \to\\
 \ccase{f^{*} a \sqcup_{\bnil} f^{*} b}{\bnil \to \amb(f^{*} a \triangleright g\, a\, v, f^{*} b \triangleright g\, b\, v)}}, 
\end{multline*}
\begin{multline*}
h_{2}\, u\, v \overset{\brec}{=} \ccase{v}{\amb(\bar{a}, \bar{b}) \to\\
 \ccase{f^{*} \bar{a} 
\sqcup_{\bnil} f^{*} \bar{b}}{\bnil \to \amb(f^{*} \bar{a} \triangleright g\, u\, \bar{a}, f^{*} \triangleright g\, u\, \bar{b})}}.
\end{multline*}

\begin{lem}\label{lem-andintro}\hfill
\begin{enumerate}
\item\label{lem-andintro-1} $d \br A \land c \br \itdown(B) \to (h_{2}\, \bleft(d)\, c) \br \itdown(A \land B)$.

\item\label{lem-andintro-2} $c_{1} \br \itdown(A) \land c_{2} \br \itdown(B) \to (h_{1}\, c_{1}\, c_{2}) \br \itdown(A \land B)$.
\end{enumerate} 
\end{lem}
\begin{proof} Both statements are shown by s.p.\ induction.

(\ref{lem-andintro-1}) Assume that $d \br A$ and $c \br \itdown(B)$. Then $c = \amb(\bar{a}, \bar{b})$ with $f^{*} \bar{a} \sqcup_{\bnil} f^{*} \bar{b} = \bnil$. Hence, $h_{2}\, \bleft(d)\, c = \amb(a', b')$ with $a' = f^{*} \bar{a} \triangleright g\, d\, \bar{a}$ and $b' = f^{*} \bar{b} \triangleright g\, d\, \bar{b}$. Since $f^{*} \bar{a} \sqcup_{\bnil} f^{*} \bar{b} = \bnil$, we have that $a' \ne \bot$ or $b' \ne \bot$ as required. It remains to show that all defined elements of $\{ a', b' \}$ realise $(A \land B) \lor \itdown(A \land B)$. Without restriction suppose that $a'$ is defined. Recall that $d \br A$. So, \mytextcolor{red}{we} have that (i) $\bar{a} = \bleft(p)$ with $p \br B$ and $g\, \bleft(d)\, \bleft(p) = \bleft(\bpair(d, p))$, or (ii) $\bar{a} = \bright(q)$ and $g\, \bleft(d)\, \bar{a} = \bright(h_{2}\, \bleft(d)\, q)$. Thus, $a' = \bright(h_{2}\, \bleft(d)\, q)$. Because $p \br B$, it follows in Case (i) that $\bar{a} \br (B \lor \itdown(B))$. Therefore, $(g\, \bleft(d)\, \bar{a}) \br ((A \land B) \lor \itdown(A \land B))$, i.e., $(h_{2}\, \bleft(d)\, c) \br \itdown(A \land B)$. 

In Case (ii) we have $q \br \itdown(B)$. By the induction hypothesis we therefore obtain that $(h_{2}\, \bleft(d)\, q) \br \itdown(A \land B)$. Thus, $a' \br ((A \land B) \lor \itdown(A \land B))$ which implies that 
\[
(h_{2}\, \bleft(d)\, c) \br \itdown(A \land B).
\]

(\ref{lem-andintro-2}) Now, \mytextcolor{red}{suppose} that $c_{1} \br A$ and $c_{2} \br B$. Then $c_{1} = \amb(u, v)$ with $f^{*} u \sqcup_{\bnil} f^{*} v = \bnil$ and $h_{1}\, c_{1}\, c_{2} = \amb(a', b')$ with $a' = f^{*} u \triangleright g\ u\, c_{2}$
and $b' = f^{*} v \triangleright g\, v\, c_{2}$. Since $f^{*} u \sqcup_{\bnil} f^{*} v = \bnil$, we \mytextcolor{red}{have} that $a' = g\, u\, c_{2} \ne \bot$ or $b' = g\, v\, c_{2} \ne \bot$, as required. Again, it remains to show that each defined element of $\{ a', b' \}$ realises $(A \land B) \lor \itdown(A \land B)$. Without restriction assume that $a'$ is defined. Then $a' = g\, u\, c_{2}$ and (i) $u = \bleft(d)$ with 
\[
g\, u\, c_{2} = \ccase{c_{2}}{\bleft(e) \to \bleft(\bpair(d, e)); \bright(e') \to \bright(h_{2}\, u\, e')},
\]
or (ii) $u = \bright(d')$ and $g\, u\, c_{2} = \bright(h_{1}\, d'\, c_{2})$, i.e., $a' = \bright(h_{1}\, d'\, c_{2})$. 

In Case (i) it follows that either $a' = \bleft(\bpair(d, e))$ and thus $a' \br ((A \land B) \lor \itdown((A \land B))$, or $a' = \bright(h_{2}\, d\, e')$, from which we obtain with the first statement that again $a' \br ((A \land B) \lor \itdown(A \land B))$.

In Case (ii) we have that $d' \br \itdown(A)$ and hence by the induction hypothesis that $(h_{1}\, d'\, c_{2}) \br \itdown(A \land B)$. Thus, $a' \br ((A \land B) \lor \itdown(A \land B))$, that is $(h_{1}\, c_{1}\, c_{2}) \br \itdown(A \land B)$.
\end{proof}

\begin{cor}\label{cor-andintro}
Rule~\emph{($\itdown$-$\land$-intro)} is realised by $h_{1}$.
\end{cor}

We extend CFP by the Rules~(\emph{$\itdown$-$\rest$-intro}) and (\emph{$\itdown$-$\land$-intro}).

\begin{lem}\label{lem-itconcrule}
The following rules for the iterated concurrency modality are derivable:

\begin{longtable}{c|c} 
\multicolumn{2}{c}{$\dfrac{\itdown(A)}{A}$\,\, \emph{($\itdown$-H ($A$ Harrop))}} \\[2ex] \hline \\[-1.5ex]
$\dfrac{\ddown(A)}{\itdown(A)}$\,\, \emph{($\itdown$-emb)}
& $\dfrac{\ddown(\itdown(A))}{\itdown(A)}$\,\, \emph{($\ddown$-$\itdown$-absorb)} \\[3ex]
$\dfrac{A}{\itdown(A)}$\,\, \emph{($\itdown$-return)}
& $\dfrac{\itdown(A) \quad A\to \itdown(A')}{\itdown(A')}$\,\, \emph{($\itdown$-bind)}  \\[3ex]
$\dfrac{\itdown(\rt{B}{A})}{\rt{B}{\itdown(A)}}$\,\, \emph{($\itdown$-$\rest$-dist)}
& $\dfrac{\itdown(A \vee B)}{\itdown(\itdown(A) \vee \itdown(B))}$\,\,
                                 \emph{($\itdown$-$\vee$-dist)}  \\[3ex]
$\dfrac{\itdown(A) \quad A \to B}{\itdown(B)}$\,\, \emph{($\itdown$-mon)}
& $\dfrac{\itdown(A \wedge B)}{\itdown(A) \wedge \itdown(B)}$\,\, \emph{($\itdown$-$\wedge$-dist)}  \\[3ex] 
 $\dfrac{\itdown(A \to B)}{\itdown(A) \to \itdown(B)}$\,\, \emph{($\itdown$-$\to$-dist)}
& $\dfrac{\itdown(\itdown(A))}{\itdown(A)}$\,\, \emph{($\itdown$-idem).}
\\[4ex] \hline 
\multicolumn{2}{c}
 {$\rule{0mm}{8mm}
 \dfrac{\rt{B_1}{A}\ \ldots\ \rt{B_n}{A}}
 {\rt{B_1\lor\cdots\lor B_n}{\itdown(A)}}$\,\, 
\emph{($\itdown$-$\rest$-$\lor$)}} 
\end{longtable}
\end{lem}
\begin{proof}
Rule~(\emph{$\itdown$-H}) follows by induction. From $A \vee A \to A$ we obtain with Rules~(\emph{$\ddown$-mon}) and \mytextcolor{red}{(\emph{$\ddown$-H})} that $\ddown(A \vee A) \to \ddown(A)$ and $\ddown(A) \to A$, hence $\ddown(A \vee A) \to A$.

The Rules~(\emph{$\itdown$-emb})  and (\emph{$\ddown$-$\itdown$-absorb})
follow directly from the definition of $\itdown$ and Rule~(\emph{$\ddown$-mon}).

Rule~(\emph{$\itdown$-return}) follows directly from the definition of $\itdown$
and the Rule~(\emph{$\ddown$-return}).

For Rule~(\emph{$\itdown$-bind})
assume that $A\to \itdown(A')$. We prove by induction that also $\itdown(A) \to \itdown(A')$, that is, we have to show that $\ddown(A \vee \itdown(A')) \to \itdown(A')$, which means that we must demonstrate that $\ddown(A \vee \itdown(A')) \to \ddown(A' \vee \itdown(A'))$. Because of the monotonicity of $\ddown$ it suffices to prove that $A \vee \itdown(A') \to A' \vee \itdown(A')$, which is an immediate consequence of our assumption.

In case of Rule~(\emph{$\itdown$-$\rest$-dist}) we apply the induction
principle again. It suffices to show
$\ddown(\rt{B}{A} \vee \rt{B}{\itdown(A)}) \to \rt{B}{\itdown(A)}$.
With (\emph{$\itdown$-return}) we have $A \to \itdown(A)$ and
hence, because of (\emph{$\rest$-mon}),
$\rt{B}{A} \to \rt{B}{\itdown(A)}$.
Therefore,
$(\rt{B}{A} \vee \rt{B}{\itdown(A)}) \to \rt{B}{\itdown(A)}$,
from which it follows with
Rule~(\emph{$\ddown$--mon}) that
$\ddown(\rt{B}{A} \vee \rt{B}{\itdown(A)} )\to \ddown(\rt{B}{\itdown(A)})$.
With Rules~(\emph{$\ddown$-$\rest$-absorb}) and~(\emph{$\ddown$-$\itdown$-absorb}) we get
$\ddown(\rt{B}{A} \vee \rt{B}{\itdown(A)}) \to \rt{B}{\itdown(A)}$.

The Rules~(\emph{$\itdown$-$\lor$-dist}), (\emph{$\itdown$-mon}), 
(\emph{$\itdown$-$\land$-dist}), (\emph{$\itdown$-$\to$-dist}), and (\emph{$\itdown$-idem})
follow from the monadic laws (\emph{$\itdown$-return})
and (\emph{$\itdown$-bind}) in the usual way.

Rule~(\emph{$\itdown$-$\rest$-$\lor$}) is obtained, roughly speaking, by iterating Rule~(\emph{$\ddown$-$\rest$-$\lor$}). The proof is by induction on $n$. 
For $n=1$, the rule follows immediately \mytextcolor{red}{with Rules}~(\emph{$\itdown$-return}) 
and (\emph{$\rest$-mon}).
For the step assume $\rt{B_1}{A}$, $\rt{B_2}{A}$, \ldots, $\rt{B_{n+1}}{A}$.
By the induction hypothesis, $\rt{B_2\lor\cdots\lor B_{n+1}}{\itdown(A)}$.
With Rule~(\emph{$\ddown$-$\rest$-$\lor$}) it therefore follows 
$\rt{B_1\lor B_2\lor\cdots\lor B_{n+1}}{\ddown(A \lor \itdown(A))}$.
Since $\ddown(A \lor \itdown(A))$ is equivalent to $\itdown(A)$, we obtain
$\rt{B_1\lor B_2\lor\cdots\lor B_{n+1}}{\itdown(A)}$ by applying Rule~(\emph{$\rest$-mon}).
\end{proof}

The subsequent list contains realisers for the rules in the above lemma extracted from their proofs.

\begin{longtable}{c}
$\dfrac{c \br \ddown(A)}
       {(\mamb\,\bleft\,c)\, \br \itdown(A)}$\,\, 
               \text{($\itdown$-emb)}  \\[4ex]
$\dfrac{c \br\ddown(\itdown(A))}
       {(\mamb\,\bright\,c)\, \br\itdown(A)}$\,\, 
               \text{($\ddown$-$\itdown$-absorb)} \\[4ex]
$\dfrac{a \br A}
       {(f_{\mathrm{ret}}\,a) \br \itdown(A)}$\,\, 
               \text{($\itdown$-return, $f_{\mathrm{ret}}$ below)} \\[4ex]
$\dfrac{c \br \itdown(A) \quad g \br (A\to \itdown(A'))}
       {(f_{\mathrm{bind}}\,g\,c)\,\br\itdown(A')}$\,\,
               \text{($\itdown$-bind, $f_{\mathrm{bind}}$ below)}  \\[4ex]
$\dfrac{c \br \itdown(\rt{B}{A})}
       {(f_{\rest{-}\mathrm{dist}}\,c) \br \rt{B}{\itdown(A)}}$\,\, 
               \text{($\itdown$-$\rest$-dist, $f_{\rest{-}\mathrm{dist}}$ below)}  \\[4ex]
$\dfrac{c \br \itdown(A) \quad g \br (A\to B)}
       {(f_{\mathrm{mon}}\,g\,c)\,\br\itdown(B)}$\,\, 
               \text{($\itdown$-mon, $f_{\mathrm{mon}}$ below)} \\[4ex]
$\dfrac{c \br \itdown(A \vee B)}
       {(f_{\mathrm{mon}}\,(\mlr\,f_{\mathrm{ret}})\,c)\,\br \itdown(\itdown(A) \vee \itdown(B))}$\,\,
               \text{($\itdown$-$\vee$-dist)}  \\[4ex]
$\dfrac{c \br \itdown(A \wedge B)}
  {\bpair(f_{\mathrm{mon}}\, (\pi_L\, c), f_{\mathrm{mon}}\, (\pi_R\, c))\, \br (\itdown(A) \wedge \itdown(B))}$\,\, 
               \text{($\itdown$-$\wedge$-dist)}  \\[4ex]
$\dfrac{c \br \itdown(A \to B)}
  {(f_{\mathrm{bind}}\, (\lambda a.\, f_{\mathrm{bind}}\, (\lambda f.\, f\,a)\,c))\, \br(\itdown(A) \to \itdown(B))}$\,\, 
               \text{($\itdown$-$\to$-dist)} \\[4ex]
$ \dfrac{c\br \itdown(\itdown(A))}{(f_{\mathrm{bind}}\, (\lambda a.\, a)\,c)\, \br\itdown(A)}$\,\, 
     \text{($\itdown$-idem)} \\[4ex]
$ \dfrac{b_{1} \br \rt{B_{1}}{A}\, \ldots\, b_{n} \br \rt{B_{n}}{A}}{(f_{n}\, b_{1}\, \ldots\, b_{n})\, \br \rt{B_{1} \lor \cdots \lor B_{n}}{\itdown(A)}}$\,\, \text{($\itdown$-$\rest$-$\lor$)}     
\end{longtable}     
where
\begin{gather*}
f_{\mathrm{ret}}\,a \Def \amb(\bleft(a),\bot)\\
f_{\mathrm{bind}}\,g \overset{\brec}{=} \mamb\,(\lambda d.\,
     \ccase{d}{\bleft(a)\to \bright(g\,a);  \bright(c')\to f_{\mathrm{bind}}\,g\,c'}),\\
f_{\rest{-}\mathrm{dist}} \overset{\brec}{=} \mamb\,(\lambda d.\,
     \ccase{d}{\bleft(a)\to d; \bright(c')\to f_{\rest{-}\mathrm{dist}}\,c'}),\\
f_{\mathrm{mon}}\,g \overset{\brec}{=} \mamb\,(\lambda d.\,
     \ccase{d}{\bleft(a)\to \bleft(g\,a); \bright(c')\to f_{\mathrm{mon}}\,g\,c'}),\\
\mlr\,g \Def \lambda c.\, \ccase{c}{\bleft(a)\to\bleft(g\,a); \bright(b)\to\bright(g\,b)},\\
\pi_L \Def \lambda p.\,  \ccase{p}{\bpair(a,\_)\to a},\\
\pi_R \Def \lambda p.\,  \ccase{p}{\bpair(\_,b)\to b},\\
 f_{1}\, b_{1} \Def f_{\mathrm{ret}}\low b_{1},\\ 
 f_{n+1}\, b_{1} \ldots b_{n+1} \Def \bar{g}\, b_{1}\, (f_{n}\, b_{2}\, \ldots\, b_{n+1}),\\
 \bar{g}\, a\, c \Def \amb_{\mathbf{LR}}(\bleft\low a, \bright\low c).
\end{gather*}
\vspace{.3ex}

A further useful rule that we will use in the sequel is a concurrent version of half-strong co-induction.

\begin{lem}[\bf Concurrent Half-strong Co-induction Principle]\label{lem-chscoi}
Let $\fun{\Phi_{0}}{\PPP(X)}{\PPP(X)}$ be a monotone operator and $\Phi(Y) \Def \itdown(\Phi_{0}(Y))$. Then:
\[\text{
If $Y \subseteq \itdown(\Phi(Y) \cup \nu \Phi)$ then $Y \subseteq \nu \Phi$.
}\]
\end{lem}
The principle is an immediate consequence of  the generalised half-strong co-induction principle (Lemma~\ref{lem-halfstrong}): Because of Rule~(\emph{$\itdown$-mon}) $\Phi$ is monotone and with Rule~(\emph{$\itdown$-idem}) we have that $\Phi$ absorbes $\itdown$.
Note that  $\tau(\Phi)(\alpha)$ is of the form $\bfA^*(\rho(\alpha))$. So, 
if $s : P \to \bfA^*(\Phi(P)+\nu\Phi)$ 
realises $P \subseteq \itdown(\Phi(P)\cup\nu\Phi)$,
then $P \subseteq\nu\Phi$ is realised by
$\chscoiter_{\Phi}\,s : P\to\nu\Phi$
with
\[ \chscoiter_{\Phi}\,s \Def  f, \]
where $f$ is defined as in the case of the realisability of generalised half-strong co-induction (Example~\ref{ex-ind-co}).
Moreover,
\begin{gather*}
\absorb^{\alpha}_{\itdown,\Phi} \Def \fun{f_{\mathrm{bind}}\, (\lambda a.\, a)}{\bfA^{*}(\Phi(\alpha))}{\Phi(\alpha)}, \\
\mon_{\itdown} \Def f_{\mathrm{mon}}.
\end{gather*}

\ignore{
Below we list the extracted programs for these rules. 
We first define the program
\[
\mamb \Def \lambda f.\, \lambda c.\, \ccase{c}{\amb(a, b) \to \amb(\stc{f}{a}, \stc{f}{b}},
\]
where $\stc{f}{a}$ denotes \emph{strict application}:
\[
\stc{f}{a} \Def \ccase{a}{C(\_) \to f\, a \mid C \in \{ \bnil, \bleft, \bright, \bpair, \pfun \}}.
\]
Note that $\stc{f}{a} = f\, a$ if $a \not= \bot$ and  $\stc{f}{\bot} = \bot$. 

Now, the following programs realising these rules are extracted 
from the rules in Lemma~\ref{lem-itconcrule} (in the cases where
the conclusion is a non-Harrop formula).
\begin{center}
\begin{tabular}{lcl}
(\emph{$\ddown$-$\itdown$-absorb})&&$\mamb\,\bright$\\
(\emph{$\itdown$-emb})&&$\mamb\,\bleft$\\
(\emph{$\itdown$-return})&&$\lambda\,a\,.\, \amb(\bleft(a),\bot)$\\
(\emph{$\itdown$-bind})&& 
  $f\,c\,g \overset{\brec}{=} \mamb$\\
  &&\hbox{}\hspace{4em}
        $(\lambda d\,.\,
               \ccase{d}{\bleft(a)\to \bright(g\,a);\,
                         \bright(c')\to f\,c'\,g})$\\
  &&\hbox{}\hspace{4em} 
       $c$
\\
\end{tabular}
\end{center}
}

\ignore{
\begin{rem}
 By the Rules (\emph{$\itdown$-idem}) and (\emph{$\ddown$-$\itdown$-absorb}), every s.p.\ operator of the form
$\Phi(X) = \itdown(\ldots)$ satisfies the assumption of Lemma~\ref{lem-conhscoind} that $\itdown(\Phi(X))\subseteq\Phi(X)$.
\end{rem}
}

%% file: sec-conarch-final-final.tex
A powerful tool in the investigation in~\cite{btifp,bt} of the relationship between 
the signed digit representation and infinite Gray code 
is Archimedean induction, which is the Archimedean principle formulated as an induction rule:
\[
\dfrac{(\forall x \ne 0)\ (|x| \leq 1/2 \to P(2x)) \to P(x)}
{(\forall x \ne 0)\ P(x)}\, (\mathrm{AI})
\]
In IFP Rule(AI) is deduced as a special case of well-founded induction (cf.~\cite{btifp}). A useful variant is:

\[
\dfrac{(\forall x \in B \setminus \{ 0 \})\, P(x) \vee (|x| \le 1/2 
\wedge B(2x) \wedge (P(2x) \to P(x)))}
{(\forall x \in B \setminus \{ 0 \})\, P(x)} \, (\mathrm{AIB})
\]

In what follows a concurrent version of the Rule~(AIB) is needed.

\begin{defi}\label{def-caibstar}
\emph{Iterated concurrent Archimedean induction} is the rule
\[
\dfrac{(\forall x \in B \setminus \{ 0 \})\, \itdown(P(x) \vee (|x| \le 1/2 
\wedge B(2x) \wedge (P(2x) \to P(x))))}
{(\forall x \in B\setminus \{ 0 \})\, \itdown(P(x))} \, (\mathrm{CAIB^*})
\]
where $B$ and $P$ are non-Harrop predicates.
\end{defi}

\begin{lem}\label{lem-ciabstarreal}
Rule {\rm (CAIB$^*$)} is realisable. 
Let $s$ realise the premise of {\rm (CAIB$^*$)}. 
Then $\caibs \Def \lambda b.\, a\,(s\, b)$ realises the 
conclusion of {\rm (CAIB$^*$)},
where $a$ is defined by simultaneous recursion together with $s'$ as
\begin{align*}
a\,w &\overset{\brec}{=} \mamb\, s'\,w && \\
s'\,u &\overset{\brec}{=} 
  \ccase{u}{\ &&\hspace{-1.6cm} \bleft(\bleft(c)) \to \bleft(c);  
\\
 & &&\hspace{-1.6cm} \bleft(\bright(\bpair(b, d))) \to \bright(f_{\mathrm{bind}}\, (f_{\mathrm{ret}} \circ d)\, (a\, (s\,b)));
\\ 
  &&&\hspace{-1.6cm} \bright(w') \to \bright(a\,w')\ }.
\end{align*}
\end{lem}
\begin{proof}
$\lambda b.\, a\,(s\, b)$ has the right type, 
$\tau(B) \to\bfA^*(\tau(P))$
where $\bfA^*(\alpha) \Def \bfix \beta.\,\bfA(\alpha+\beta)$.
This holds, since we have that
$s : \tau(B) \to \bfA^*(\rho)$, 
with $\rho \Def \tau(P) + \tau(B)\times (\tau(P) \to\tau(P))$,
from which one can infer by a simple type inference that 
$a : \bfA^*(\rho) \to \bfA^*(\tau(P))$, and
$s' : (\rho + \bfA^*(\rho))  \to (\tau(P) + \bfA^*(\tau(P)))$.

Set $C(x) \Def P(x) \vee (|x| \le 1/2 \wedge B(2x) \wedge (P(2x) \to P(x)))$.
To complete the proof, it clearly suffices to show that if $x\ne 0$ and
$w$ realises $\itdown(C(x))$, then $a\,w$ realises $\itdown(P(x))$.

Let 
\[
Q(x) \Def (\forall w)\, (w \br \itdown(C(x)) \to (a\, w) \br \itdown(P(x))).
\]
We use Rule (AI) to prove that $(\forall x \not= 0)\, Q(x)$.

Let $x \not= 0$. The Archimedean induction hypothesis is
\begin{description}
\item[\hspace{1em}(AIH)] $|x|\le 1/2 \to Q(2x)$.
\end{description}  
We show $Q(x)$ by a side induction on the definition of $w \br \itdown(C(x))$.
Hence we assume $w \br \itdown(C(x))$, that is,
$w \br \ddown(C(x) \vee \itdown(C(x)))$, and have to derive 
that $a\, w$ realises $\itdown(P(x))$. 

Thus, $w = \amb(u, v)$ and
$a\, w = \amb(\stc{s'}{u},\stc{s'}\,v)$.
Furthermore, $u \not= \bot$ or $v \not= \bot$, hence 
$\stc{s'}{u} = s'\,u$ or $\stc{s'}{v} = s'\,v$.
Moreover,  for $k \in \{ u, v \}$, if $k \not= \bot$ then 
\[
k \br(C(x) \vee \itdown(C(x)))\,.
\]
Before showing that $\amb(\stc{s'}{u},\stc{s'}{v})$ realises
$\itdown(P(x))$, we prove:
\begin{equation}\label{eq-ciabrealstar}
\text{If $k \in \{ u, v \}$ such that $k \not= \bot$, then 
$(s'\,k) \br (P(x) \vee \itdown(P(x)))$.
}
\end{equation}
Without restriction assume $k = u \not= \bot$.
Hence $u \br (C(x) \vee \itdown(C(x)))$
and $\stc{s'}{u}=s'\,u$. 

If $u = \bleft(\bleft(c))$ with $c\br P(x)$, then $s'\,u = \bleft(c)$, which means that $(s'\,u) \br (P(x) \lor \itdown(P(x)))$.

If $u = \bleft(\bright(\bpair(b,d)))$ with $|x|\le1/2$, $b\br B(2x)$ and
$d\br(P(2x) \to P(x))$, then $s'\,u = \bright(f_{\mathrm{bind}}\, (f_{\mathrm{ret}} \circ d)\, (a\,(s\,b)))$. 
Moreover, by (AIH), $a\,(s\,b)$ realises $P(2 x)$. Therefore, since $f_{\mathrm{ret}} \circ d$ realises $(P(2x) \to \itdown(P(x)))$, we have that 
$f_{\mathrm{bind}}\, (f_{\mathrm{ret}} \circ d)\, (a\,(s\,b))$ realises $\itdown(P(x))$, that is, $(s'\,u) \br (P(x) \lor \itdown(P(x)))$.

If $u = \bright(w')$ with $w' \br \itdown(C(x))$, then $s'\,u = \bright(a\,w')$. 
By the side induction hypothesis, $(a\,w') \br \itdown(P(x))$. Thus, $(s'\,u) \br (P(x) \lor \itdown(P(x)))$.

From the proof of (\ref{eq-ciabrealstar}) it follows easily that $\stc{s'}{u}, \stc{s'}{v} \in (\tau(P) + \bfA^*(\tau(P)))$.
Therefore, it remains to show:
\begin{enumerate}
\item\label{lem-ciabrealstar-1} For some $k\in\{u,v\}$, $\stc{s'}{k}\not=\bot$.
\item\label{lem-ciabrealstar-2} If $k\in\{u,v\}$ such that $\stc{s'}{k} \not= \bot$, 
then $(\stc{s'}{k}) \br (P(x)\vee\itdown(P(x)))$.
\end{enumerate}

For (\ref{lem-ciabrealstar-1}) assume without restriction that $k = u \not= \bot$.
Then $\stc{s'}{u} = s'\,u$. By (\ref{eq-ciabrealstar}), $(s'\,u) \br (P(x) \vee \itdown(P(x)))$.
Hence, $s'\,u \not= \bot$.

To prove (\ref{lem-ciabrealstar-2}), suppose $\stc{s'}{k} \not= \bot$.
Then $k\not= \bot$. It follows that $(\stc{s'}{k}) \br (P(x)\vee\itdown(P(x)))$, 
by (\ref{eq-ciabrealstar}).
\end{proof}

For what follows, we extend CFP again by adding Rule~(CAIB$^*$).

%% file: sec-sdgc-final-final.tex
We first briefly  formalise the definitions given in Section~\ref{sec-indef} in IFP and CFP, respectively. For details the reader is referred to \cite{btifp,bt}.

Define the IFP predicates $\bSD(x)$ and $\bS(x)$ as follows 
\begin{gather*}
\bSD(z) \Def (z = -1 \vee z = 1) \vee z = 0,  \\
\bII(z, x) \Def |2x - z| \le 1, \\
 \mytextcolor{red}{ \bS(x) \overset{\nu}{=} (\exists z)\, \bSD(z) \wedge  \bII(z, x) \wedge \bS(2x -z). }
\end{gather*}
For $\bthree \Def (\bone + \bone) + \bone$ and $\delta^{\omega} \Def \bfix \alpha.\, \delta \times \alpha$, their types are
\[
\tau(\bSD) = \bthree \quad \text{and} \quad
\tau(\bS)   = \bthree^{\omega}.
\]
The predicate $\bSD(z)$ is realised as follows
\begin{align*}
d \br \bSD(z) = \mbox{} &(d = \bleft(\bleft(\bnil)) \wedge z = -1) \vee \mbox{} \\
                        &(d = \bleft(\bright(\bnil)) \wedge z = 1) \vee \mbox{} \\
                        &(d = \bright(\nil) \wedge z = 0).
\end{align*} 
\mytextcolor{red}{
In the sequel the three digits $-1, 1, 0$ will be identified with their realisers, which are programs of type $\bthree$.
For the predicate $\bS(x)$ we obtain
\begin{align*}
p \br \bS(x) 
&\overset{\nu}{=} (\exists d, p')\, p = \bpair(d, p') \wedge (\exists z)\, d \br \bSD(z) \wedge \bII(z, x) \wedge p' \br \bS(2x - z) \\
&\overset{\nu}{=} (\exists d, p')\, p = \bpair(d, p') \wedge \bII(d,x) \le 1 \wedge p' \br \bS(2x - d).
\end{align*} 
}
The realisers of $\bS(x)$ are hence streams of digits -1, 0, 1.

Let us next consider the Gray code case. Define
\begin{gather*}
\bB(x) \Def x \le 0 \vee x \ge 0 , \\
\bD(x) \Def x \not= 0 \to \bB(x), \\
\bt(x) \Def 1 - 2|x|, \\
\bG(x) \overset{\nu}{=} (-1 \le x \le 1) \wedge \bD(x) \wedge \bG(\bt(x)),
 \end{gather*}
where types are $\tau(\bB) = \tau(\bD) = \btwo$ with $\btwo \Def \bone + \bone$, and $\tau(\bG) = \btwo^{\omega}$. Then the predicates $\bD(x)$ and $\bG(x)$ are realised as follows
\begin{gather*}
a \br \bD(x) = a : \btwo \wedge (x \not= 0 \to (a = \bleft(\bnil) \wedge x \le 0) \vee (a = \bright(\bnil) \wedge x \ge 0),\\
q \br \bG(x) \overset{\nu}{=} (-1 \le x \le 1) \wedge (\exists a, q')\, q = \bpair(a, q') \wedge a \br \bD(x) \wedge q' \br \bG(\bt(x)).
\end{gather*}

Since an infinite Gray code may contain a $\bot$, a sequential access of the sequence from left to right will diverge when it accesses a $\bot$. However, because at most one $\bot$ is contained in each sequence, if one evaluates the first two cells concurrently, then at least one of the two processes is guaranteed to terminate. On the basis of this idea Berger and Tsuiki~\cite{bt} showed that a concurrent algorithm converting infinite Gray code into signed digit representation can be extracted from a CFP proof. To this end a concurrent variant of the predicate $\bS$ is introduced
\[
\mytextcolor{red}{\bS_{2}(x) \overset{\nu}{=} \ddown((\exists z)\, \bSD(z) \wedge \bII(z, x) \wedge \bS_{2}(2x - z)). }
\]
$\bS_{2}(x)$ means that a signed digit representation of $x$ is obtained through the concurrent computation of two threads. Note that $\tau(\bS_{2}) = \bfix \alpha.\, \bfA(\bthree \times \alpha)$.
\begin{thmC}[\cite{bt}]\label{thm-StoGtoS2}
$\bS \subseteq \bG \subseteq \bS_{2}$.
\end{thmC}

This result expresses the fact that, as explained above, when computably translating from Gray code to signed digit representation one needs to allow for computations to be carried out concurrently, which, however is not the case for the converse translation from signed digit representation to Gray code. On the other hand the result is not completely satisfying, as one would like to see under which conditions both representations are computably equivalent. 

To achieve a result of this kind we have to introduce concurrent Gray code. In addition we have to allow for iterated concurrent computations. Define
\begin{gather*}
\mytextcolor{red}{\bS^{*}(x) \overset{\nu}{=} \itdown((\exists z)\, \bSD(z) \wedge \bII(z, x) \wedge \bS^{*}(2x - z)),} \\
\bB^{*}(x) \Def \itdown(x \le 0 \vee x \ge 0), \\
\bD^{*}(x) \Def x \not= 0 \to \bB^{*}(x), \\
\bG^{*}(x) \overset{\nu}{=} (-1 \le x \le 1) \wedge \bD^{*}(x) \wedge \bG^{*}(\bt(x)).
\end{gather*}

For the concurrent signed digit representation we have 
$\bS^* = \nu(\Phi_{\bS^*})$ where
\begin{gather*}
\mytextcolor{red}{\Phi_{\bS^*}(Y)(x) = \itdown((\exists z)\, \bSD(z) \wedge \bII(z, x)\land Y(2x-z)),} \\
\tau(\Phi_{\bS^*})(\alpha) = \bfA^*(\bthree\times\alpha), \\
\tau(\bS^*) = \bfix\tau(\Phi_{\bS^*}) = 
                \bfix\,\alpha\,.\,\bfA^*(\bthree\times\alpha), \\
\mon_{\tau(\Phi_{\bS^*})} : (\alpha\to\beta) \to 
                      \bfA^*(\bthree\times\alpha)\to\bfA^*(\bthree\times\beta), \\
\mon_{\tau(\Phi_{\bS^*})}\,f = f_{\mathrm{mon}}\,(\lambda \bpair(d,a)\,.\,\bpair(d,f\,a)),
\end{gather*}
and
\begin{gather*}
\bfA^*(\alpha) = \bfix \beta\,.\,\bfA(\alpha+\beta),\\
f_{\mathrm{mon}} : (\alpha\to\beta)\to \bfA^*(\alpha)\to\bfA^*(\beta),\\ 
\mytextcolor{red}{f_{\mathrm{mon}}
           \overset{\brec}{=} \lambda h.\,\mamb\,(\lambda d\,.\,
 \ccase{d}{\bleft(a)\to \bleft(h\,a);\,\bright(c)\to f_{\mathrm{mon}}\,h\,c}). }
\end{gather*}

On the other hand, for the concurrent infinite Gray code we have 
$\bG^* = \nu(\Phi_{\bG^*})$ where
\begin{gather*}
\Phi_{\bG^*}(Y)(x) = \bD^*(x)\land Y(\bt(x)),\\
\tau(\Phi_{\bG^*})(\alpha) = \bfA^*(\btwo)\times\alpha, \\
\tau(\bG^*) = \bfix\tau(\Phi_{\bG^*}) = 
               \bfix\,\alpha\,.\,\bfA^*(\btwo)\times\alpha = 
                (\bfA^*(\btwo))^{\omega},\\
\mon_{\tau(\Phi_{\bG^*})} : (\alpha\to\beta) \to 
                   (\bfA^*(\btwo)\times\alpha)\to(\bfA^*(\btwo)\times\beta), \\
\mon_{\tau(\Phi_{\bS^*})}\,f\,\bpair(m,a) = \bpair(m,f\,a).
\end{gather*}

We see that a realiser of $\bG^*(x)$ is simply an ordinary infinite stream
(where the cons-operation is the deterministic constructor $\bpair$) of non-deterministic
partial binary digits, whereas a realiser of $\bS^*(x)$ is something that could be called
a non-deterministic stream (where the cons-operation is non-deterministic)
given by a pair of concurrent computations at least one
of which will yield a head, which is signed digit, and a tail, which is again a
non-deterministic stream. Consequently, the function $\mon_{\tau(\Phi_{\bG^*})}$
is much simpler than $\mon_{\tau(\Phi_{\bS^*})}$.

Our next goal is to show that $\bS^{*} = \bG^{*}$. Note that iterated concurrent computations also occur in the case of $\bS_{2}$, which can be seen by unfolding the co-inductive definition.

\begin{lem}\label{lem-neg}
If $\bS^{*}(x)$, then also
\begin{enumerate}
\item\label{lem-neg-1} $\bS^{*}(-x)$ and
\item\label{lem-neg-2} $\bS^{*}(\bt(x))$.
\end{enumerate}
\end{lem}
\begin{proof} %\mycolor{red} % no use for this command
(\ref{lem-neg-1}) 
Let $P \Def \set{x}{\bS^{*}(-x)}$. We use co-indution to prove that $P \subseteq \bS^{*}$. So, we have to show that
\[
\bS^{*}(-x) \rightarrow \itdown((\exists z')\, \bSD(z') \wedge \bII(z', x) \wedge \bS^{*}(-(2x -z'))),
\]
i.e.,
\[
\itdown((\exists z)\, \bSD(z) \wedge \bII(z, -x) \wedge \bS^{*}(-2x-z)) 
\rightarrow \itdown((\exists z')\, \bSD(z') \wedge \bII(z', x) \wedge \bS^{*}(-(2x -z'))).
\]
Because of Rule (\emph{$\itdown$-mon}) it suffices to prove that 
\[
(\exists z)\, \bSD(z) \wedge \bII(z, -x) \wedge \bS^{*}(-2x-z)
\rightarrow (\exists z')\, \bSD(z') \wedge \bII(z', x) \wedge \bS^{*}(-(2x -z')).
\]

Let $z \in \bSD$. We show that 
\[
\bII(z, -x) \wedge \bS^{*}(-2x-z)
\rightarrow (\exists z')\, \bSD(z') \wedge \bII(z', x) \wedge \bS^{*}(-(2x -z')).
\]

If $\bII(z, -x)$ with $\bS^{*}(-2x - z)$, then $\bII(-z, x)$. Moreover, $\bS^{*}(-(2x-z'))$ with $z' = -z$.

(\ref{lem-neg-2})
Let $\bQ \Def \set{y}{(\exists x)\, \bS^{*}(x) \wedge y = \bt(x)}$. We use concurrent half-strong co-induction (Lemma~\ref{lem-chscoi})  to show that $\bQ \subseteq \bS^{*}$. This means that we have to prove that
\[
\bQ(y) \rightarrow \itdown(\itdown((\exists z')\, \bSD(z') \wedge \bII(z', y) \wedge \bQ(2y -z')) \vee \bS^{*}(y)).
\]
By the definition of $\bQ$ we therefore have to show for $x, y$ with $y = \bt(x)$ that
\[
\itdown((\exists z)\, \bSD(z) \wedge \bII(z, x) \wedge \bS^{*}(2x-z))
\rightarrow \itdown(\itdown((\exists z')\, \bSD(z') \wedge \bII(z', y) \wedge \bQ(2y -z')) \vee \bS^{*}(y)).
\]
Because of the monotonicity and return rules for $\itdown$, it suffices to prove that for $x, y,$ with $y = \bt(x)$ and $ z\in \bSD$ that
\begin{equation*}\label{eq-neg-2}
\bII(z, x) \wedge \bS^{*}(2x-z)
\rightarrow ((\exists z')\, \bSD(z') \wedge \bII(z', y) \wedge \bQ(2y -z')) \vee \bS^{*}(y), 
\end{equation*}
which follows by case distinction:

\begin{ncase}
$z=-1$
\end{ncase}
We have that $x \in [-1, 0]$ and $\bS^{*}(2x+1)$. Since $2x+1 = \bt(x) = y$ in this case, it follows that $\bS^{*}(y)$. 

\begin{ncase}
$z=1$.
\end{ncase}
Now, $x \in [0, 1]$. Therefore, $2x - 1 = - \bt(x) = -y$. Since moreover, $\bS^{*}(2x-1)$, we have that $\bS^{*}(-y)$, from which we obtain that $\bS^{*}(y)$, by Part (\ref{lem-neg-1}). 

\begin{ncase}
$z=0$.
\end{ncase}
It follows that $x \in [-1/2, 1/2]$, which implies that $\bII(1, \bt(x))$. Therefore, it suffices to show that $\bQ(2y -1)$. Note that $2y-1 = -\bt(y) = -\bt(\bt(x)) = \bt(2x)$. Since $\bS^{*}(2x)$, it follows that $\bQ(2y-1)$.
\end{proof}

The first statement is realised by
\[
f_{\ref{lem-neg}.\ref{lem-neg-1}} \Def f_{\mathrm{mon}}\,
\lambda\, \bpair(d,a).\, \bpair({-}d,a),
\]
and the second by
\begin{multline*}
  f_{\ref{lem-neg}.\ref{lem-neg-2}} \Def
    \chscoiter_{\tau(\Phi_{\bS^*})}\, (f_{\mathrm{mon}}\,\lambda\,\bpair(d,a).\, \\
   \bcase d\, \bof \{ {-}1 \to \bright(a); \\
                       1 \to \bright(f_{\ref{lem-neg}.\ref{lem-neg-1}}\,a); \\
      0 \to \bleft(\amb(\bleft(\bpair(1,a)),\bot))\, \}).
\end{multline*}

\begin{prop}\label{prop-ctog}
$\bS^{*} \subseteq \bG^{*}$.
\end{prop}
\begin{proof}
The proof is by co-induction. Because of Lemma~\ref{lem-neg}(\ref{lem-neg-2})
it remains to show that 
\[
\bS^{*}(x) \rightarrow \bD^{*}(x)\,,
 \]
that is, $(\forall x\in\bS^{*}\setminus\{0\})\,\itdown(\bB(x))$. 
We prove this formula by iterated concurrent Archimedean induction (CAIB$^*$).
Therefore, we have to show
\begin{equation*}
\label{prop-ctog-caib-prem}
(\forall x\in\bS^{*}\setminus\{0\})\,\itdown(\bB(x) 
\vee (|x| \le 1/2 \wedge \bS^{*}(2x) \wedge (\bB(2x) \to \bB(x))))\,.
\end{equation*}
By Rule~($\itdown$-mon) and since $\bB(2x) \to \bB(x)$ holds,
it suffices to show
\begin{equation*}
\label{prop-ctog-claim}\mytextcolor{red}{
((\exists z)\,  \bSD(z) \wedge \bII(z,x) \wedge \bS^{*}(2x-d)) \to 
                   (\bB(x) \vee (|x| \le 1/2 \wedge \bS^{*}(2x)))\,. }
\end{equation*}
\mytextcolor{red}{
Let $z \in\bSD$ with $\bII(z,x)$ and $\bS^{*}(2x-z)$.
If $z = -1$ or $z= 1$, then $x\le0$ or $x\ge 0$, hence $\bB(x)$.
If $z=0$, then $|2x| \le 1$ and $\bS^{*}(2x)$. 
}
\end{proof}

Statement $\bS^*\subseteq\bD^*$ is realised by

\begin{align*}
f_d \Def \caibs\,(f_{\mathrm{mon}}\,\lambda\,\bpair(d,a).\,
  \ccase{d}{{-}&1\to\bleft({-}1);\\
  		      & 1\to\bleft(1); \\
		      & 0\to\bright(\bpair(a,\id))}),
\end{align*}  

and the inclusion $\bS^*\subseteq\bG^*$ by
$f_{\ref{prop-ctog}}\Def
  \coiter_{\tau(\Phi_{\bG^*})}\,
      (\lambda\,c.\,\bpair(f_d\,c,f_{\ref{lem-neg}.\ref{lem-neg-2}}\, c))$,
that is,
\[ f_{\ref{prop-ctog}}\,c \overset{\brec}{=} 
     \bpair(f_d\,c, f_{\ref{prop-ctog}}\,(f_{\ref{lem-neg}.\ref{lem-neg-2}}\, c)). 
\]

Let us now consider the converse inclusion. 

\begin{lem}\label{lem-gneg}
$\bG^{*}(x) \rightarrow \bG^{*}(-x)$.
\end{lem}
\begin{proof}
Let $\bP \Def \set{x}{\bG^{*}(-x)}$. We will use strong co-induction to show that $\bP \subseteq \bG^{*}$. To this end it needs to be shown that 
\[
\bP(y) \to \bD^{*}(y) \wedge (\bP(\bt(y)) \vee \bG^{*}(\bt(y))).
\]
Assume that $\bP(y)$. Then $y = -x$ for some $x \in \bG^{*}$. It follows that $\bD^{*}(x)$ and $\bG^{*}(\bt(x))$. The first property implies that $\bD^{*}(-x)$, that is $\bD^{*}(y)$. Moreover, since $\bt(x) = \bt(-x)$, we also have $\bG^{*}(\bt(y))$.
\end{proof}

The statement $(\forall x)\,(\bG^*(x)\to\bG^*(-x))$ is realised by
\[ 
f_{\ref{lem-gneg}}\,\bpair(m,a) \Def
      \bpair(f_{\mathrm{mon}}\,(\lambda d.\,{-}d)\,m, a).
\]

\begin{lem}\label{lem-glikec}
For $d \in \{ -1, 1 \}$,
\[
\bG^{*}(x) \rightarrow \bII(d, x) \rightarrow \bG^{*}(2x - d).
\]
\end{lem}
\begin{proof}
By case distinction on $d$ we show that 
\[
\bG^{*}(x) \rightarrow \bG^{*}(2x - d).
\]
Therefore, assume that $\bG^{*}(x)$. Then $\bD^{*}(x)$ and $\bG^{*}(\bt(x))$. 

\begin{ncase}
$d=-1$.
\end{ncase}
This case is obvious as $\bt(x)=2x+1$.

\begin{ncase}
$d=1$.
\end{ncase}
Now $\bt(x) = 1-2x = - (2x-1)$. Therefore, the statement follows with Lemma~\ref{lem-gneg}.
\end{proof}

$(\forall x)\,(\bG^*(x)\to
(\forall d\in\{{-}1,1\})\,(\bII(d, x) \rightarrow \bG^{*}(2x - d)))$
is realised by
\[ 
f_{\ref{lem-glikec}}\,\bpair(\_,a)\,d \Def
\ccase{d}{{-}1\to a; 1 \to f_{\ref{lem-gneg}}\,a}.
\]

\begin{lem}\label{lem-min1}
$\bII(1, x) \wedge \bG^{*}(x) \to \bG^{*}(1-x)$.
\end{lem}
\begin{proof}
Assume $\bII(1, x)$ and $\bG^{*}(x)$. $\bII(1, x)$ implies $\bII(1, 1 - x)$
and hence $\bD^{*}(1-x)$. Therefore, it suffices to show $\bG^{*}(\bt(1-x))$.
Note that in our case, $\bt(1-x) = 2x -1$. Thus, $\bG^{*}(x)$ implies
$\bG^{*}(\bt(1-x))$, by Lemma~\ref{lem-glikec}.
\end{proof}

$(\forall x)\,(\bII(1, x) \wedge \bG^{*}(x) \to \bG^{*}(1-x))$
is realised by
\[ 
f_{\ref{lem-min1}}\,a \Def \bpair(\amb(\bleft(1),\bot),f_{\ref{lem-glikec}}\,a\,1).
\]
Hence, 
\[ 
f_{\ref{lem-min1}}\,\bpair(\_,a) =
   \bpair(\amb(\bleft(1),\bot), f_{\ref{lem-gneg}}\,a).
\]

\begin{lem}\label{lem-ii0}
$\bII(0, x) \wedge \bG^{*}(x) \to \bG^{*}(2x)$.
\end{lem}
\begin{proof}
Assume $\bII(0, x)$ and $\bG^{*}(x)$. Then $\bD^*(x)$ and $\bG^{*}(\bt(x))$.
Hence also $\bD^*(2x)$ and $\bt(x)\ge 0$. With Lemma~\ref{lem-min1} it
follows $\bG^*(2|x|)$ (since $1-\bt(x) = 2|x|$).
Therefore, $\bG^*(\bt(2x))$ (since $\bt(|x|) = \bt(x)$). It follows
$\bG^*(2x)$.
\end{proof}   

$(\forall x)\,(\bII(0, x) \wedge \bG^{*}(x) \to \bG^{*}(2x))$
is realised by
\[ 
f_{\ref{lem-ii0}}\,\bpair(m,a) \Def \bpair(m, \pi_{R}\, (f_{\ref{lem-min1}}\,a)).
\] 

\begin{lem}\label{lem-d2}
$\bD^{*}(x) \leftrightarrow \itdown\rt{x \not= 0}{(x \le 0 \vee x \ge 0)}$.
\end{lem}
\begin{proof}
Note that $\neg (x \not= 0)$, i.e, $x = 0$, implies that $x \le 0 \wedge x \ge 0$. Now, assume that $\bD^{*}(x)$. Then we obtain with Rule~(\emph{$\itdown$-$\rest$-intro}) that $\rt{x \not= 0}{\itdown(x \le 0 \vee x \ge 0)}$.

For the converse implication assume that $x \not= 0$. Then if follows with Rule~(\emph{$\rest$-mp}) that $\itdown(x \le 0 \vee x \ge 0)$. So, $\bD^{*}(x)$ holds.
\end{proof}

$(\forall x)\,(\bD^*(x) \to \itdown\rt{x \neq 0}{(x\le 0 \lor x\ge 0)})$
is realised by the function $h^+$ as defined in Lemma~\ref{lem-itdownintro}.
The converse implication is realised by the identity.

\begin{lem}\label{lem-halfC}
\mytextcolor{red}{$\bG^{*}(x) \to \itdown((\exists z)\, \bSD(z) \wedge \bII(z, x))$.}
\end{lem}
\begin{proof}
Assume $\bG^*(x)$.
We have to show \mytextcolor{red}{$A(x) \Def \itdown((\exists z)\, \bSD(z) \wedge \bII(z, x))$}.
By Rule~($\ddown$-$\itdown$-absorb), it suffices to show $\ddown(A(x))$.
Therefore, by Rule~(\emph{$\ddown$-$\vee$-elim}) it is sufficient to derive
\begin{enumerate}
\item \label{lem-halfC-1} $\rt{x \not= 0}{A(x)}$
\item \label{lem-halfC-2} $\rt{\bt(x) \not= 0}{A(x)}$.
\end{enumerate}
Set $\bD'(y) \Def \rt{\bt(x) \not= 0}{\itdown(y \le 0 \vee y \ge 0)}$.
The assumption $\bG^{*}(x)$ entails $\bD^{*}(x)$ and $\bD^{*}(\bt(x))$ and
therefore also $\bD'(x)$ as well as $\bD'(\bt(x))$, by Lemma~\ref{lem-d2}.
Now, (\ref{lem-halfC-1}) follows immediately from $\bD'(x)$,
by Rules~(\emph{$\itdown$-mon}) and (\emph{$\rest$-mon}).

For (\ref{lem-halfC-2}) we use that $\bD'(\bt(x))$. Because of Rule~(\emph{$\rest$-mon}) it suffices to derive 
$\itdown(\bt(x) \le 0 \vee \bt(x) \ge 0) \to A(x)$.
With Rule~($\itdown$-bind) this can be further reduced to showing
$(\bt(x) \le 0 \vee \bt(x) \ge 0) \to A(x)$.
If $\bt(x) \le 0$, then $x \not= 0$,
and hence $\itdown(x \le 0 \vee x \ge 0)$ because of $\bD^{*}(x)$.
As \mytextcolor{red}{$x \le 0 \vee x \ge 0 \to (\exists z)\, \bSD(z) \wedge \bII(z, x)$},
an application of Rule~(\emph{$\itdown$-mon}) leads to $A(x)$.
If $\bt(x) \ge 0$, then $\bII(0, x)$. Hence, $A(x)$,
by Rule~(\emph{$\itdown$-return}).
\end{proof}

\mytextcolor{red}{$(\forall x)\,(\bG^*(x) \to \itdown((\exists z)\, \bSD(z) \wedge \bII(z, x)))$}
is realised by the function 
\begin{multline*}
f_{\ref{lem-halfC}}\,\bpair(m,\bpair(n,\_)) \Def 
\mamb\,\bright\,
    \amb(h^+
         \,m, \\
         f_{\mathrm{bind}}\,(\lambda\,c.\,
                             \ccase{c}{\bleft(a)\to m;
                                       \bright(b)\to f_{\mathrm{ret}}\,0})\,n).
\end{multline*}

\begin{prop}\label{prop-gtoc}
$\bG^{*} \subseteq \bS^{*}$.
\end{prop}
\begin{proof}
Again the assertion follows by co-induction. We have to show that 
\[\mytextcolor{red}{
\bG^{*}(x) \to \itdown((\exists z)\, \bSD(z) \wedge \bII(z, x) \wedge \bG^{*}(2x-z)).  }
\]
From Lemmas~\ref{lem-glikec} and \ref{lem-ii0} it follows 
\begin{equation*}\label{eq-gtoc-1} \mytextcolor{red}{
\bG^{*}(x) \to (\exists z)\, \bSD(z) \wedge \bII(z, x) \to (\exists z)\, (\bSD(z) \wedge \bII(z, x) \wedge \bG^{*}(2x-z)), }
\end{equation*}
from which we obtain by Rule~(\emph{$\itdown$-mon}) that
\begin{equation*}\label{eq-gtoc-2} \mytextcolor{red}{
\bG^{*}(x) \to \itdown((\exists z)\, \bSD(z) \wedge \bII(z, x)) \to \itdown((\exists z)\, (\bSD(z) \wedge \bII(z, x) \wedge \bG^{*}(2x-z))). }
\end{equation*}
Note that because of Lemma~\ref{lem-halfC} the assumption \mytextcolor{red}{$\itdown((\exists z)\, \bSD(z) \wedge \bII(z, x))$} can be discharged.
\end{proof}

The inclusion $\bG^*\subseteq \bS^*$ is realised by 
\begin{multline*}
f_{\ref{prop-gtoc}} \Def \coiter_{\tau(\Phi_{\bS^*})}\,(\lambda a.\,
   f_{\mathrm{mon}}\,(\lambda d.\,
                     \ccase{d}{{-}1\to f_{\ref{lem-glikec}}\,a\,d; \\
                               1   \to f_{\ref{lem-glikec}}\,a\,d;
                               0   \to f_{\ref{lem-ii0}}\,a})\,
                  (f_{\ref{lem-halfC}}\,a)). 
\end{multline*}

\begin{thm}\label{thm-maingc}
$\bS^{*} = \bG^{*}$
\end{thm}

%% file: sec-comset-final-final.tex
As is well known, the collection $\mathcal{K}(X)$ of non-empty compact subsets of a non-empty compact metric space is a compact space again with respect to the Hausdorff metric $\hdm$. 

\begin{defi}
Let $(X,E)$ be a digit space. A \emph{digital tree} is a nonempty set  $T \subseteq E^{< \omega}$
of finite sequences of digits that is downwards closed under the prefix ordering
and has no maximal element, that is, $[] \in T$ and whenever $[e_0, \ldots, e_n] \in T$ , then
$[e_0, \ldots, e_{n-1}] \in T$ and $[e_{0}, \ldots, e_n, e] \in T$ for some $e \in E$.
\end{defi}

Let $\TTT_E$ denote the set of digital trees with digits in $E$.
Note that each such tree is finitely branching as $E$ is finite.
Moreover, every element $[e_0,\ldots,e_{n-1}]\in T$ can be continued to an 
infinite path $\alpha$ in $T$, that is, $\alpha \in E^\omega$ is such that $\alpha_i = e_i$, for $i < n$, and 
$[\alpha_0,\ldots,\alpha_{k-1}]\in T$ for all $k\in\mytextcolor{red}{\NN}$.
In the following we write  $\alpha\in [T]$ to mean that $\alpha$ is a 
path in $T$, and by a path we always mean an infinite path.
$[T]$ is a non-empty compact subset of $E^\omega$, for every tree $T \in \TTT_E$, and conversely, for every non-empty compact subset $C$ of $E^\omega$,  $C = [T^C]$, where $T^C \Def \{\, \alpha^{< n} \mid \alpha \in C \wedge n \in \mytextcolor{red}{\NN} \,\}$ (cf.~\cite{bs}).

For $T \in \TTT_E$ and $n \ge 0$, let $T^{\le n}$ be the finite initial subtree of $T$ of height $n$. Then 
\[
T^{\le n} = \{\, \alpha^{<m} \mid \alpha \in [T] \wedge m \le n \,\}.
\]
Every such initial subtree defines a map $\fun{f_{T,n}}{X}{\PPP(X)}$ from $X$ into the powerset of $X$ in the obvious way:
\[
f_{T,n}(x) \Def \set{\vec{e}(x)}{\vec{e} \in E^n \cap T}.
\]

\begin{defi}
\label{def-val}
For every $T \in \TTT_E$ we define its \emph{value} by
\[
\treeval{T} \Def \bigcap\nolimits_{n\in\bN} f_{T,n}[X].
\]
\end{defi}

\begin{lemC}[\cite{bs}]
\label{lem-semT}
$\treeval{T} =  \set{\val{\alpha}}{\alpha\in [T]}$.
\end{lemC}

The metric defined on $E^\omega$ in Section~\ref{sec-dig} can be transferred to $\TTT_E$. As we will see next, it coincides with the Hausdorff metric.

\begin{lemC}[\cite{bs}]\label{lem-treehausd}
For $S, T \in \TTT_E$,
\[
\treehdm(S,T) =
\begin{cases}
0 & \text{ if $S = T$,}\\
2^{- \min \set{n}{S^{\le n}\not= T^{\le n}}}  & \text{otherwise.}
\end{cases}
\]
\end{lemC}

\begin{propC}[\cite{bs}]\label{prop-conttreeval} \hfill % moved first item to new line
\begin{enumerate}
\item\label{conttreeval1}
$\fun{\treeval{\cdot}}{\TTT_E}{\KKK(X)}$ is onto and uniformly 
continuous.
\item\label{conttreeval2}
The topology on $\KKK(X)$ induced by the Hausdorff metric is equivalent 
to the quotient topology induced by $\treeval{\cdot}$.
\end{enumerate}
\end{propC}

As a consequence of Lemma~\ref{lem-semT} we have for trees $T_1, T_2 \in \TTT_E$ that
\[
\treeval{T_1} = \treeval{T_2} \Longleftrightarrow
(\forall \alpha \in [T_1]) (\exists \beta \in [T_2])\, \alpha \sim \beta \wedge
(\forall \beta \in [T_2]) (\exists \alpha \in [T_1])\, \alpha \sim \beta.
\]

\begin{defi}
A digital tree $T \in \TTT_{E}$ is \emph{full}, if $[T]$ is closed under $\sim$.
\end{defi}

\begin{lem}\label{lem-fulltree}
Let $T_1, T_2 \in\TTT_E$ be full. Then
\[
\treeval{T_1} = \treeval{T_2} \Longleftrightarrow T_1 = T_2.
\]
\end{lem}
\begin{proof}
We have that
\begin{align*}
\treeval{T_1} = \treeval{T_2} 
&\Rightarrow (\forall \alpha \in [T_1])(\exists \beta \in [T_2])\, \alpha \sim \beta \wedge (\forall \beta \in [T_2])(\exists \alpha \in [T_1])\, \alpha \sim \beta \\
&\Rightarrow [T_1] \subseteq [T_2] \wedge [T_2] \subseteq [T_1] \qquad\text{(as $T_1, T_2$ are full)}\\
&\Rightarrow T_1 = T_2.
\end{align*}

The converse implication holds trivially, as $\treeval{\cdot}$ is a map.
\end{proof}

\begin{lem}
Let $T \in \TTT_E$ and $C$ be a non-empty compact subset of $E^\omega$. Then the following two statements hold:
\begin{enumerate}
\item  If $T$ is full, then $[T]$ is a non-empty compact subset of $E^\omega$ that is closed under $\sim$.

\item If $C$ is closed under $\sim$, then $T^C$ is full.

\end{enumerate}
\end{lem}

By Proposition~\ref{prop-valhom}, $\widehat{\val{\cdot}}$ is a bijection between $E^\omega / \mathord{\sim}$ and space $X$. So, if $C$ is a non-empty compact, and hence closed, subset of $X$, $\widehat{\val{\cdot}}^{-1}[C]$ is a non-empty closed subset of $E^\omega / \mathord{\sim}$. Consequently, $\overline{C} \Def q_\sim^{-1}[\widehat{\val{\cdot}}^{-1}[C]]$ is a non-empty closed, and thus compact, subset of  $E^\omega$, which in addition is closed under $\sim$. It follows that $T^{\overline{C}}$ is a full tree in $\TTT_D$ with $\val{T^{\overline{C}}} = C$.

Let $\TTT^f_E$ be the subspace of full trees in $\TTT_E$. 

\begin{prop}
$\fun{\treeval{\cdot}}{\TTT^f_E}{\KKK(X)}$ is one-to-one and onto.
\end{prop}

This shows that $\KKK(X)$ can be represented in straightforward one-to-one way without requiring that $X$ is represented in this way. For the special case of the real interval $\II$ and Gray code we have seen in Section~\ref{sec-dig} that $\fun{\val{\cdot}}{\widehat{G}}{\II}$ is one-to-one. Hence, every digital tree $T \in \TTT_{\overline{\GF}}$ with $[T] \subseteq \widehat{G}$ is full.

%% file: sec-arcp-final-final.tex
Archimedean induction is a formulation of the Archimedean property as an induction principle introduced in \cite{btifp}. It turned out quite a powerful proof tool. We will now lift this induction principle to the case of non-empty compact sets. Let $\bZ(x)$ be the predicate stating that $x$ is an integer. Moreover, for $K : \PPP(\iota)$ and $n \in \bZ$ define
\begin{gather*}
K \le 0 \Def (\forall x \eps K)\, x \le 0, \\
K \ge 0 \Def (\forall x \eps K)\, x \ge 0, \\
|K| \Def \set{y}{(\exists x \eps K)\, y = |x|}, \\
n K \Def \set{y}{(\exists x \eps K)\, y = n x}, \\
\bK_{0}(K) \Def \bK(K) \wedge 0 \noteps K. 
\end{gather*}
Here, $\bK$ is a predicate constant denoting the set of non-empty compact subsets of 
the compact interval
$\bII \Def [-1,1]$ (see Section~\ref{sub-ifp}). 

\begin{defi}\label{def-arcp}
\emph{Archimedean induction for compact sets} is the following rule
\[
\frac{(\forall K \in \bK_0)\, 
       ((\forall K' \in \bK) (K'\subseteq K \wedge |K'| \le 1/2 \rightarrow P(2 K'))
                   \rightarrow P(K))}
     {(\forall K \in \bK_0)\, P(K)}\,\, \mathrm{(AIC)}.
\]
\end{defi}

Also Archimedean induction for compact sets is a special case of well-founded induction. Set
\[
K'' \prec K \Def K \in \bK \wedge (\exists K' \in \bK)\, 
(K'\subseteq K \wedge |K'| \le 1/2 \wedge K'' = 2 K').
\]
Then the premise of Rule~(AIC) is equivalent to $\prog_{\prec, \bK_{0}}(P)$.

\begin{lem}\label{lem-acck0}
$\acc_{\prec}(K)$ if and only if $K \in \bK_{0}$.
\end{lem}
\begin{proof}
The `only if' part follows by induction on $\acc_{\prec}(K)$. 
Since $\acc_{\prec} = \mu \Phi$ with 
\begin{align*}
\Phi(X)  &\Def
  \{\, K \in \bK \mid (\forall K' \in \bK)\, 
(K'\subseteq K \land |K'| \le 1/2 \to X(2 K')) \,\}
\end{align*}
we have to show that $\Phi(\bK_{0}) \subseteq \bK_{0}$. Let $K \in \Phi(\bK_{0})$ and 
suppose that $0 \eps K$. Then the compact set $\{ 0 \}$ is a subset of $K$ 
and $|\{ 0 \}| \le 1/2$. Since $K \in \Phi(\bK_{0})$, it follows that 
$2 \{ 0 \} \in \bK_{0}$, 
which is a contradiction.

The `if' part reduces, by \textbf{BT}$_{\mathbf{nc}}$, to the implication $K \in \bK_{0} \to \neg \path_{\prec}(K)$. Therefore, we assume $K \in \bK_{0}$ and $\path_{\prec}(K)$ with the aim to arrive at a contradiction. Recall that 
\begin{align*}
\path_{\prec}(K) &\overset{\nu}{=}  
K \in \bK \wedge (\exists K' \in \bK)\,
(K' \subseteq K \wedge |K'| \le 1/2 \wedge \path_{\prec}(2K')).
\end{align*}
Hence by unfolding $\path_{\prec}(K)$ we can construct a decreasing
sequence $(K_{n})_{n \in \bN} \subseteq \bK$ such that $K_{0} = K$ and
for all $n \in \bN$, $|K_{n}| \le 2^{-n}$.

The sequence $(K_{n})_{n \in \bN}$ is constructed such that $K_0=K$
and for all $n$, $\path_{\prec}(2^nK_n)$, $|K_n| \le 2^{-n}$, and $K_{n+1}\subseteq K_n$.
For $K_0$ the properties hold by assumption. For the step, we use that
$\path_{\prec}(2^nK_n)$ holds and therefore exists $K' \in \bK$ such that
$K' \subseteq 2^nK_n$, $|K'| \le 1/2 $ and $\path_{\prec}(2K')$.
We set $K_{n+1} \Def 2^{-n}K'$. Since $2^{n+1}K_{n+1} = 2K'$ it follows that
$\path_{\prec}(2^{n+1}K_{n+1})$ holds.
Furthermore, $|K_{n+1}|= 2^{-n}|K'| \le 2^{-(n+1)}$.
Finally, $K_{n+1}=2^{-n}K'\subseteq 2^{-n}(2^nK_n)=K_n$.

Since $K$ is compact, there exists
$x \eps \bigcap_{n \in \bN} K_{n}$. Then $|x| \le 2^{-n}$, for all
$n \in \bN$. By the Archimedean axiom, $x=0$, hence $0\eps K$, 
contradicting our assumption.
\end{proof}

\begin{prop}\label{prop-aic}
Archimedean induction for compact sets \emph{(AIC)} is derivable in \emph{IFP($\AAA_{R}$)} and realised by $\brec$.
\end{prop}
\begin{proof}
It remains to show the second statement. Note that both $\prec$ and the predicate $\bK_{0}(K)$ are Harrop. Moreover, let $s$ realise the premise of Rule~(AIC). Then $s$ also realises $\prog_{\prec, \bK_{0}}$. Therefore, it follows with the result in Example~\ref{ex-wfirec} that $\brec\, s$ realises $\acc_{\prec} \cap \bK_{0} \subseteq P$ which is equivalent to the conclusion of the rule.
\end{proof}

In applications, Archimedean induction is mostly used for compact sets that are generated in a particular way and therefore come with a special kind of realisers. Here, we are interested in the case that non-empty  compact sets are represented by signed digit code. 

\begin{defi}\label{def-sk}
We define the analogue of the signed digit representation for compact sets as
\[
\bS_{\bK}(K) \overset{\nu}{=} \bK(K)\land(\exists E \in \bP_{\fin}(\bSD))\,
K \subseteq \bII_E \wedge (\forall d\eps E) (K_d \not= \emptyset \wedge 
\bS_{\bK}(\bav{d}^{-1}[K_d]))
\]
with $\bII_{d} \Def \set{x}{\bII(d, x)}$,
$\bII_{E} \Def \set{x}{(\exists d\eps E)\,  \bII(d, x)}$,
$K_d \Def K \cap \bII_d$, and $\bav{d}(x) \Def (x + d) / 2$.
\end{defi}

As follows from the definition of realisability, the type $\tau(\bS_{\bK}(K))$
of realisers of the formula $\bS_{\bK}(K)$ is given by
\begin{align*}
\tau(\bS_{\bK}(K)) &= \bfix \alpha.\, \sum\nolimits_{E \in \bP_{\fin}(\bSD)} \alpha^{\| E \|} \\
		       &= \bfix \alpha.\, \{ -1 \} \times \alpha + \{ 0 \} \times \alpha + \{ 1 \} \times \alpha +  
		       \{ -1, 0 \} \times \alpha^{2} + \mbox{} \\
		       &\hspace{4cm} \{-1, 1 \} \times \alpha^{2} + \{ 0, 1 \} \times \alpha^{2} +  \{ -1, 0, 1 \} \times \alpha^{3},
\end{align*}		       
which is essentially the set $\TTT_{\bSD}$ of all digital trees. 

In the case of non-empty compact sets with property $\bS_{\bK}$ the Archimedean induction rule 
can be much simplified. Let $\bS^0_{\bK}$ denote the set of all $K \in \bS_{\bK}$ with $0 \noteps K$.

\begin{defi}\label{def-aicsd}
\emph{Archimedean induction for signed-digit represented compact sets} is the rule
\[
\frac{(\forall K \in \bS^0_{\bK})\, (P(K) \vee (\bS_{\bK}(2(K_0)) \wedge (P(2(K_0)) \rightarrow P(K))))}{(\forall K \in \bS^0_{\bK})\, P(K)}\,\, \mathrm{(AICSD)}
\]
where $P$ is a non-Harrop predicate.
\end{defi}

\begin{prop}\label{prop-aicsd}
Archimedean induction for signed-digit represented compact sets \emph{(AICSD)} is derivable in \emph{IFP($\AAA_{R}$)}, and if $s$ realises the premise, then 
\[
f\, a \overset{\rm rec}{=} \mathbf{case}\, s\, a\, \mathbf{of} \{\mathbf{Left}(b) \rightarrow b; \mathbf{Right}(\bpair(a', h)) \rightarrow h (f a') \}
\]
realises the conclusion.
\end{prop}
\begin{proof}
We will show that Rule~(AICSD) is a consequence of Rule~(AIC). Set $A(X) \Def \bS_{\bK}(X) \rightarrow P(X)$. It suffices to show that the premise of (AICSD) implies the premise of (AIC). Therefore, let $K\in\bK_0$ and assume that
\begin{equation}\label{eq-aicsd}
(\forall K' \in \bK) (K'\subseteq K \wedge |K'| \le 1/2 \rightarrow A(2 K'))
\end{equation}
We have to prove that $A(K)$. So, let $K \in \bS_{\bK}$. Then we need to derive $P(K)$.

By the premise of (AICSD) we have that 
\begin{enumerate}
\item $P(K)$ or
\item $\bS_{\bK}(2(K_0)) \wedge (P(2(K_0)) \rightarrow P(K)).$
\end{enumerate}
In the first case we are done. Let us therefore consider the second case.

Since $| K_0 | \le 1/2$, by (\ref{eq-aicsd}), $A(2(K_0))$ holds, i.e.,
\[
\bS_{\bK}(2(K_0)) \rightarrow P(2(K_0)).
\]
Since we know that $\bS_{\bK}(2(K_0))$, we obtain that $P(2(K_0))$ and hence, 
as we are considering the second case, that $P(K)$.

As we have just seen, the premise of (AICSD) implies the premise of (AIC). If the first premise is realised by $s$ the latter is realised by
\[
s' = \lambda f.\,\lambda a.\, \mathbf{case}\, s\, a\, \mathbf{of} \{ \mathbf{Left}(b) \rightarrow b; \mathbf{Right}(\bpair(a', h)) \rightarrow h\, (f\, a') \}.
\]
Thus, $f \overset{\rm rec}{=} s'\, f$, i.e., 
\[
f\, a \overset{\rm rec}{=} \mathbf{case}\, s\, a\, \mathbf{of} \{\mathbf{Left}(b) \rightarrow b; \mathbf{Right}(\bpair(a', h)) \rightarrow h\, (f\, a')\},
\]
realises the conclusion $(\forall K \in \bS^0_{\bK})\, P(K)$.
\end{proof}

If one strengthens the premise of Rule~(AICSD) to all $K \in \bS_{\bK}$ instead of only those not containing $0$, one can strengthen the conclusion to a restriction. 

\begin{defi}\label{def-aicr}
\emph{Archimedean induction with restriction for signed-digit represented compact sets} is the rule
\[
\frac{(\forall K \in \bS_{\bK})\, (P(K) \vee (\bS_{\bK}(2(K_0)) \wedge (P(2(K_0)) \rightarrow P(K))))}{(\forall K \in \bS_{\bK})\, \rt{0\, \mathrel{\not\varepsilon}\, K}{P(K)}}\,\, \text{(AICR)}
\]
where $P$ is a productive non-Harrop predicate.
\end{defi}

\begin{prop}\label{prop-airc}
Archimedean induction with restriction for signed-digit represented compact sets \emph{(AICR)} is realisable. More precisely, if $s$ realises the premise, then the conclusion is realised by
\[
\chi\, a \overset{\rm rec}{=} \ccase{s\, a}{\bleft(b) \to b; \bright(\bpair(a', f)) \to \stc{f}{(\chi\, a')}}.
\]
\end{prop}
\begin{proof}
Assuming $a\, \mathbf{r}\, \bS_{\bK}(K)$ we have to show
\begin{enumerate}
\item\label{prop-airc-1}
$0 \,\noteps K \to \chi\, a \not= \bot$

\item\label{prop-airc-2}
$\chi\, a \not= \bot \to (\chi\, a) \br P(K)$.

\end{enumerate}

(\ref{prop-airc-1}) It suffices to show
\begin{equation*}\label{eq-airc-1'}
(\forall K \in \bK_0) (\forall a)\, (a\, \mathbf{r}\, \bS_{\bK}(K) \rightarrow \chi\, a \not= \bot).
\end{equation*}

We prove the statement by Archimedean induction for compact sets. Let $K \in \bK_0$ and assume, as induction hypothesis, 
\[
(\forall K' \in \bK) (K'\subseteq K \wedge |K'| \le 1/2 \rightarrow 
    (\forall a') (a'\, \mathbf{r}\, \bS_{\bK}(2K') \rightarrow \chi\, a' \not= \bot)) \mytextcolor{red}{.}
\]
We need to show that $(\forall a) (a\, \mathbf{r}\, \bS_{\bK}(K) \rightarrow \chi\,a \not= \bot)$. Assume $a\, \mathbf{r}\, \bS_{\bK}(K)$. Then 
\[
(s\,a)\, \mathbf{r}\, (P(K) \vee (\bS_{\bK}(2(K_0)) \wedge (P(2(K_0)) \rightarrow P(K)))).
\]
If $s\, a = \mathbf{Left}(b)$ where $b\, \mathbf{r}\, P(K)$, then $\chi\, a = b$. Since $P(K)$ is productive, by asumption, $b \not= \bot$. Hence, $\chi\, a \not= \bot$.
If, however, $s\, a = \mathbf{Right}(\bpair(a', f))$, then 
$a'\, \mathbf{r}\, \bS_{\bK}(2(K_0))$ (with $K_0\in\bK$) and
$f\, \mathbf{r}\,(P(2(K_0)) \rightarrow P(K))$.
Since $| K_0 | \le 1/2$, we have $\chi\, a' \not= \bot$, by the induction
hypothesis. It follows that $\chi\, a = \stc{f}{(\chi\,
a')} \not= \bot$. Thus, we are done.

(\ref{prop-airc-2}) We use Scott induction, that is, we consider the approximations $\chi_i$ of $\chi$,
\begin{align*}
& \chi_0\, a = \bot, \\
& \chi_{i+1}\, a = \ccase{s\, a}{\bleft(b) \to b; \bright(\bpair(a', f)) \to \stc{f}{(\chi_i\, a)}}
\end{align*}
Observe that a restricted form of Scott induction (as is used here) is included in the axiom set for the extension RIFP of IFP that allows to deal with realisability in a formal way (cf.\ \cite{btifp}).

By the continuity of function application, if $\chi\, a \not= \bot$, then $\chi_i\, a \not= \bot$, for some $i \in \bN$. Therefore, it suffices to show
\[
(\forall i \in \bN) (a\, \mathbf{r}\, \bS_{\bK}(K) \wedge \chi_i\, a \not= \bot \rightarrow (\chi\, a) \br P(K)).
\]
We induce on $i$. The induction base is trivial as $\chi_0\, a = \bot$. For the induction step assume $a\, \mathbf{r}\, \bS_{\bK}(K)$ and $\chi_{i+1}\, a \not= \bot$. Then
\[
(s\, a)\, \mathbf{r}\, (P(K) \vee (\bS_{\bK}(2(K_0)) \wedge (P(2(K_0)) \rightarrow P(K)))).
\]
If $s\, a = \mathbf{Left}(b)$ where $b\, \mathbf{r}\, P(K)$, then $\chi\, a = b$ and we are done. In the other case $s\,a = \mathbf{Right}(\bpair(a', f))$ where $a'\, \mathbf{r}\, \bS_{\bK}(2(K_0))$ and $f\, \mathbf{r}\, (P(2(K_0)) \rightarrow P(K))$. Then $\chi_{i+1}\, a = \stc{f}{(\chi_i\, a')}$. Since $\chi_{i+1}\, a \not= \bot$ and the application of $f$ is strict, it follows that $\chi_i\, a' \not= \bot$ as well. By the induction hypothesis we therefore have that $(\chi\, a')  \br P(2(K_0))$. Consequently, $\chi\, a = (\stc{f}{(\chi\, a')}) \br P(K)$. 
\end{proof}

As in the real number case, in what follows also a concurrent version of the predicate $\bS_{\bK}$ for the signed digit representation of non-empty compact subsets of the interval $\bII$ will be considered.
\begin{defi}\label{def-skstar}
\[
\bS^{*}_{\bK}(K) \overset{\nu}{=} \bK(K)\land\itdown((\exists E \in \bP_{\fin}(\bSD))\,
K \subseteq \bII_E \wedge (\forall d \eps E) (K_d \not= \emptyset \wedge \bS^{*}_{\bK}(\bav{d}^{-1}[K_d]))).
\]
\end{defi}
In this case the above induction rule is still valid, if we allow the `or' in the premise being decided concurrently.

\begin{defi}\label{def-parch}
\emph{Concurrent Archimedean induction with restriction for signed-digit represented compact sets} is the following rule
\[
\frac{(\forall K \in \bS^{*}_{\bK})\, \ddown(P(K) \vee (\bS^{*}_{\bK}(2(K_0)) \wedge (P(2(K_0)) \rightarrow P(K))))}{(\forall K \in \bS^{*}_{\bK})\, \rt{0\, \mathrel{\not\varepsilon}\, K}{P(K)}}\,\, \mathrm{(CAICR)}
\]
where $P = \itdown(P')$ for some non-Harrop predicate $P'$.
\end{defi}

\begin{prop}\label{prop-pairc}
Concurrent Archimedean induction with restriction for
signed-digit represented compact sets \emph{(CAICR)} is realisable.
More precisely, let $g$ be the canonical realiser of
Rule~{\rm (\emph{$\ddown$-$\itdown$-absorb})}, namely $g = \mamb\,\bright$,
and let $s$ realise the premise of (CAICR).
Set
\[
s' \Def \lambda f.\, \lambda u.\, \ccase{u}{\bleft(u') \to u'; \bright(\bpair(u'', d)) \to d\low (f\, u'')}.
\]
Then the conclusion of (CAICR) is realised by
\[
f\, b \overset{rec}{=} g\low \mamb (s' f)\, (s\, b).
\]  
\end{prop}
\begin{proof}
Let $b \br \bS^{*}_{\bK}(K)$. We have to show
\begin{enumerate}
\item\label{prop-pairc-1}
$0 \,\noteps K \to f\, b \not= \bot$,

\item\label{prop-pairc-2}
$f\, b \not= \bot \to (f\, b) \br P(K)$.

\end{enumerate}

(\ref{prop-pairc-1})
Since $b \br \bS^{*}_{\bK}(K)$ it follows that $s\, b = \amb(u, v)$ and
\begin{align*}
f\,b &= g \low \mamb (s' f)\, (s\, b)\\
&= g \low \amb((s' f)\low u, (s' f)\low v)\\
&= \amb(\bright \low ((s' f)\low u), \bright \low ((s' f)\low v))
\not= \bot
\end{align*}

(\ref{prop-pairc-2}) Again we use Scott induction. For $i \in \NN$ let
\begin{align*}
&f_{0}\, b \Def \bot, \\
&f_{i+1}\, b \Def g\low \mamb (s' f_{i})\, (s\, b).
\end{align*}

By the continuity of function application, if $f\, b \not= \bot$ then $f_{i}\, b \not= \bot$, for some $i \in \mytextcolor{red}{\bN}$. Therefore, it suffices to show
\[
(\forall i \in \mytextcolor{red}{\bN})\,(b \br \bS^{*}_{\bK}(K) \wedge f_{i}\, b \not= \bot \to (f\, b) \br P(K)).
\]
We induce on $i$. The induction base is trivial as $f_{0}\, b = \bot$. For the induction step assume that $b \br \bS^{*}_{\bK}(K)$ and $f_{i+1}\, b \not= \bot$. As we have seen above, $s\, b = \amb(u, v)$. Hence, $f_{i+1}\, b = \amb((s' f_{i})\low u, (s' f_{i})\low v)$. Moreover, $u \not= \bot$ or $v \not= \bot$, and for $k \in \{ u, v \}$ with $k \not= \bot$, $(s' f_{i})\low k = s' f_{i}\, k$ as well as
\[
k \br (P(K) \vee (\bS^{*}_{\bK}(2(K_{0})) \wedge (P(2(K_{0})) \to P(K)))).
\]
We show that $(s' f_{i}\, k) \br P(K)$.

If $k = \bleft(k')$ with $k' \br P(K)$, then $s' f_{i}\, k = k'$. Hence we are done. In the other case $k = \bright(\bpair(k'', d))$ where $k'' \br \bS^{*}_{\bK}(2(K_{0}))$ and $d \br (P(2(K_{0})) \to P(K))$. Then $s' f_{i}\, k = d\low (f_{i}\, k'')$. Since $f_{i+1}\, b \not= \bot$, we have that also $(s' f_{i})\low k \not= \bot$, as otherwise $s' f_{i} = \bot$ and hence $\amb((s' f_{i})\low u, (s' f_{i})\low v) = \bot$ as well as $g\low \amb((s' f_{i})\low u, (s' f_{i})\low v) = \bot$. Thus, $d\low (f_{i}\, k'') \not= \bot$. Because application is strict, it follows that $f_{i}\, k'' \not= \bot$. By the induction hypothesis we therefore have that $(f\, k'') \br P(2(K_{0}))$. Hence, $d\low (f\, k'') \br P(K)$. It follows that $(\mamb (s' f)\, (s\,b)) \br \ddown(P(K)$ and consequently $(f\, b) \br P(K)$.
\end{proof}

We extend the rules of CFP by the new Rules~(AICR) and (CAICR).

In the following we will use that the elements of $\bP_{\fin}(\bSD)$
are decidable classical subsets of $\bSD$:
\begin{lem}\label{lem-findec}
If $E \in \bP_{\fin}(\bSD)$, then
\begin{enumerate}
\item\label{lem-findec-1} $(\forall d \in \bSD)\,(d \eps E \lor d \noteps E)$
\item\label{lem-findec-2} $(\forall d\eps E)\,\neg\neg (d \in\bSD)$
\end{enumerate}
\end{lem}
\begin{proof}
Easy induction on $\bP_{\fin}(\bSD))$. In part (\ref{lem-findec-1}), 
$\bSD$ could be replaced
by any discrete predicate, that is, predicate $P$ such that
$(\forall x,y\in P)\,(x=y \lor x\not= y)$. In part (\ref{lem-findec-2}), $\bSD$ could be
replaced by any predicate.
\end{proof}

Let 
\[
\bB_{\bK}(K) \Def (K \le 0 \vee K \ge 0 )\vee (K_{-1} \not= \emptyset \land K_{1} \ne \emptyset).
\]
\begin{prop}\label{pn-dc}
If $\bS_{\bK}(K)$, then $\rt{0\, \mathrel{\not\varepsilon}\, K}{\bB_{\bK}(K)}$.
\end{prop}
\begin{proof}
It suffices to verify the premise of Rule~(AICR) with $P(K) \Def \bB_{\bK}(K)$.  That is we must show that
\[
\bB_{{\bK}}(K) \vee (\bS_{\bK}(2(K_0)) \wedge (\bB_{\bK}(2(K_0)) \rightarrow \bB_{\bK}(K))).
\]
Since $\bS_{\bK}(K)$, there is some $E \in \bP_{\fin}(\bSD)$ with 
$K \subseteq \bII_{E}$ so that $K_d$ is non-empty, for all $d \eps E$. 
Thanks to Lemma~\ref{lem-findec}~(\ref{lem-findec-1}), we can do a case analysis on 
the elements of $E$.

\begin{ncase}
$0\noteps E$, that is, $E \subseteq \bGC$.
\end{ncase}
In this case, we have that both $K_{-1}$ and $K_1$ are not empty, if $E = \bGC$; $K \le 0$, if $1 \noteps E$; and $K \ge 0$, 
if $-1 \noteps E$. Hence $\bB_{\bK}(K)$ holds.

\begin{ncase}
$0 \eps E$.
\end{ncase}
Then $\bS_{\bK}(2(K_0))$, by the definition of $\bS_{\bK}$. It remains to show that $\bB_{\bK}(2(K_0)) \rightarrow \bB_{\bK}(K)$. Assume that $\bB_{\bK}(2(K_0))$. Then also $\bB_{\bK}(K_0)$. 

If both $K_0 \cap \bII_{-1}$ and $K_0 \cap \bII_1$ are not empty, $K_{-1}$ and 
$K_1$ are not empty as well. In case $K_0 \le 0$, then $K \le 0$, if, in addition, 
$1 \noteps E$. Otherwise, $K_{-1} \not= \emptyset$ and $K_{1} \not= \emptyset$; 
similarly, if $K_0 \ge 0$. Thus, $\bB_{\bK}(K)$.
\end{proof}

%% file: sec-cpcode-final-final.tex
In this section the Gray code representation of non-empty compact sets is introduced 
and its connection with the signed digit representation of these sets is studied. 
\begin{defi}\label{def-gk}
\[
\bG_{\bK}(K) \overset{\nu}{=} 
\bK(K) \land  \bG(\bmin K) \land \bG(\bmax K) \land 
 (\forall d \in  \bGC)\, (K_{d} \ne \emptyset \to  \bG_{\bK}(\bt[K_{d}])).
\]
\end{defi}
Our first goal is to show that $\bS_{\bK} \subseteq \bG_{\bK}$. To this end we need the following results.

\begin{lem}\label{lem-cpneg}  
If $\bS_{\bK}(K)$ then also 
\begin{enumerate}
\item\label{lem-cpneg-1} $\bS_{\bK}(-K)$.
\item\label{lem-cpneg-1.4} $\bS(\bmin K)$.
\item\label{lem-cpneg-1.5} $\bS(\bmax K)$.
\item\label{lem-cpneg-2} $(\forall d \in \bGC)\,(K_d \ne \emptyset\to  \bS_{\bK}(\bt[K_d]))$.
\end{enumerate}
\end{lem}
\begin{proof}
(\ref{lem-cpneg-1}) Let $P \Def \set{K}{\bS_{\bK}(-K)}$. We use co-induction to prove that $P \subseteq \bS_{\bK}$. That is, we show that
\[
P(K) \rightarrow (\exists E \in \bP_{\fin}(\bSD))\, (K \subseteq \bII_{E} \wedge 
(\forall d \eps E)\, (K_d \not= \emptyset  \wedge P(\bav{d}^{-1}[K_d]))).
\]
Since $\bS_{\bK}(-K)$, there is some $F\in\bP_{\fin}(\bSD)$ so that $-K = \bigcup\set{(-K)_d}{d \in F}$ and for all $d \eps F$, $(-K)_d \not= \emptyset$ as well as $\bS_{\bK}(\bav{d}^{-1}[(-K)_d])$. Note that $(-K)_d = - (K_{(-d)})$ and $\bav{d}^{-1}[(-K )_d] = - \bav{-d}^{-1}[K_{(-d)}]$. Therefore, we can choose $E \Def \set{-d}{d \eps F}$.

(\ref{lem-cpneg-1.4})  
The proof is by co-induction.
Let $R \Def \set{x \in \bII}{(\exists K \in \bS_{\bK})\, x = \bmin K}$. 
We show 
\[ 
R(x) \rightarrow (\exists d \in \bSD)\, ( x \in \bII_{d} \land R(\bav{d}^{-1}(x))).
\]

If $x \in R$ then $x = \bmin K$, for some $K \in \bS_{\bK}$. 
Hence, there exists $E \in \bP_{\fin}(\bSD)$ so that 
$K = \bigcup \set{K_{e}}{e \eps E}$. 
Moreover, $K_{e} \ne \emptyset$ and $\bS_{\bK}(\bav{e}^{-1}[K_{e}])$, 
for all $e \eps E$. 
Order $\bSD$ by $-1 < 0 < 1$ and let $d$ be the least element of $E$ with 
respect to this order (which can be determined, thanks to 
Lemma~\ref{lem-findec}(\ref{lem-findec-1})). 
Then $\bmin K \in K_{d} \subseteq \bII_{d}$ and 
$\bav{d}^{-1}(\bmin K) \in \mytextcolor{red}{\bav{d}^{-1}[K_{d}]}$. Note that $\bav{d}^{-1}$ is monotone. 
Therefore, $\bav{d}^{-1}(\bmin K) = \bmin \bav{d}^{-1}[K_{d}]$. 
Since $\bav{d}^{-1}[K_{d}] \in \bS_{\bK}$, it follows that 
$\bav{d}^{-1}(\bmin K) \in R$. 

(\ref{lem-cpneg-1.5}) 
The statement follows easily with the first two statements and \cite[Lemma~23]{btifp}, stating that $\bS$ is closed under $\lambda x.\, {-}x$,

(\ref{lem-cpneg-2}) 
Let $d\in \bGC$ and set
\[
Q^d \Def \set{L}{(\exists K \in \bS_{\bK})\, (K_d \not= \emptyset \wedge
   L = \bt[K_d])}.
\]
We use half-strong co-induction to show that $Q^d \subseteq \bS_{\bK}$.
That is, we prove
\begin{equation*}
  Q^d(L) \rightarrow ((\exists E \in \bP_{\fin}(\bSD))\,
        (L \subseteq \bII_E
\wedge 
     (\forall e \eps E)\, (L_e \not= \emptyset \wedge
 Q^d(\bav{e}^{-1}[L_e])))) \vee \bS_{\bK}(L).
\end{equation*}

Assume that $Q^d(L)$. Then there is some $K \in \bS_{\bK}$ such that $K_d \not= \emptyset$ and $L = \bt[K_d]$. Since $\bS_{\bK}(K)$, 
there is some $F \in \bP_{\fin}(\bSD)$ so that
\begin{itemize}
\item $K \subseteq \bII_F$ and 
\item $(\forall f \eps F)\, (K_f \not= \emptyset \wedge \bS_{\bK}(\bav{f}^{-1}[K_f]))$.
\end{itemize}

We perform a case analysis on whether $d \eps F$ using 
Lemma~\ref{lem-findec}(\ref{lem-findec-1}).

If $d \eps F$, we have that $\bS_{\bK}(\bav{d}^{-1}[K_d])$. If $d = -1$, then $\bav{d}^{-1}(x) = 2x + 1 = \bt(x)$. Thus, $\bav{d}^{-1}[K_d] = \bt[K_d] = L$. That is, we have that $\bS_{\bK}(L)$. 
On the other hand, if $d = 1$, then $\bav{d}^{-1}(x) = 2x -1 = -\bt(x)$. Hence, $\bav{d}^{-1}[K_d] = -L$. It follows that $\bS_{\bK}(-L)$, whence we obtain that $\bS_{\bK}(L)$.

If $d \noteps F$, then $F \subseteq \{0\}\cup \{-d\}$, by Lemma~\ref{lem-findec}(\ref{lem-findec-2}). 

\begin{ncase}
$0 \eps F$.
\end{ncase}
Then $K_d \subseteq \bII_0$ and $\bS_{\bK}(2(K_0))$. Furthermore, 
\[
2(K_d) = 2(K \cap \bII_0 \cap \bII_d) = 2(K_0) \cap \bII_d
\]
and 
\[
\bav{1}^{-1}[L] = -\bt[L] = -\bt[\bt[K_d]] = \bt[2(K_d)] = \bt[2(K_0) \cap \bII_d],
\]
from which it follows that $Q^d(\bav{1}^{-1}[L])$.
Moreover, $L = \bt[K_d] \subseteq \bt[\bII_0] = \bII_1$ 
and hence $L_1 = L \not= \emptyset$.
Therefore, we have proven the left part of the disjunction with $E \Def \{1\}$.

\begin{ncase}
$0 \noteps F$.
\end{ncase}
Now, $F = \{ -d \}$. Hence $K_d = \{ 0 \}$ and $L = \{ 1 \}$. 
As it follows by co-induction that $\{\{-1\},\{ 1 \}\} \subseteq \bS_{\bK}$, 
we have $\bS_{\bK}(L)$.
\end{proof}

With Theorem~\ref{thm-StoGtoS2} and Lemma~\ref{lem-cpneg} we now obtain by co-induction what we were looking for. 

\begin{prop}\label{prop-cpsdgc}
$\bS_{\bK} \subseteq \bG_{\bK}$.
\end{prop}

\begin{rem}
\mytextcolor{green}{Inspecting the proof of Part (4) of Lemma~\ref{lem-cpneg}, 
one sees that the extracted realiser}
\mytextcolor{red}{yields a defined result for \emph{every} $d\in \SD$, 
even if $K_d=\emptyset$. In that case the computed realiser of the implication
$K_d \ne \emptyset\to  \bS_{\bK}(\bt[K_d])$ is defined but does not 
\mytextcolor{blue}{necessarily} realise $\bS_{\bK}(\bt[K_d])$. Therefore,
this implication cannot be strengthened to a restriction.
The computation contained in the proof of (4) takes as input only $d$ and $F$ 
(more precisely, realisers of $d\in\SD$ and $F \in \bP_{\fin}(\bSD)$)
and a realiser of $\bS_{\bK}(K)$.
It does not use the information that $K_d$ is non-empty.
This information is only needed to prove the correctness of the result.
Thus the transformation extracted from the proof of Proposition~\ref{prop-cpsdgc}
outputs for every realiser of $\bS_{\bK}(K)$ a \emph{total} full binary tree, 
that is, the addresses of nodes are \emph{all} finite sequences of elements in $\bGC$.
Each node $\vec d = d_0,\ldots,d_{n-1}$ is labelled by a pair of infinite Gray codes 
such that with $g_{\vec d} = [g_{d_0},\ldots,g_{d_{n-1}}]$
(using the notation of Section~\ref{sec-dig} where $g_{-1}$ and $g_1$ are the 
inverses of the legs of $\bt$),
if $g_{\vec d}^{-1}[K]$ is non-empty, then the label 
consists of realisers of
$\bG(\min g_{\vec d}^{-1}[K])$ and $\bG(\max g_{\vec d}^{-1}[K])$.
If
$g_{\vec d}^{-1}[K]$ 
is empty, the label is meaningless.
It is not possible to computationally distinguish meaningful from meaningless labels
since in general a realiser of $K$ does not allow us to recognise the non-emptiness 
of 
$g_{\vec d}^{-1}[K]$.
An extreme example is $K=\{0\}$ where we may only ever know that the label at
the root contains reliable information, namely realisers of $\bG(\min K)$ 
and $\bG(\max K)$. The labels at all other nodes may never be known to carry 
correct information. This shows, in particular, that the conjuncts
$\bG(\min K)$  and $\bG(\max K)$, in the co-inductive definition of $\bG_{\bK}(K)$ cannot
be replaced by the weaker formulas $\bD(\min K)$ and $\bD(\max K)$ 
which would only provide the first digits of the Gray codes of $\min K$ and $\max K$.
\mytextcolor{blue}{If, however, for some node $\vec d$ the first digit of the Gray code of $\min g_{\vec d}^{-1}[K]$  is defined,  it will tell us whether $\min g_{\vec d}^{-1}[K] \le 0$ and hence $(g_{\vec d}^{-1}[K])_{-1} \neq \emptyset$, or  $\min g_{\vec d}^{-1}[K] \ge 0$ and thus $(g_{\vec d}^{-1}[K])_{1}\neq \emptyset$; similarly for $\max g_{\vec d}^{-1}[K]$.}
}
\end{rem}

Our next aim is to show that $\bG_{\bK} \subseteq \bS_{\bK}^{*}$. We start with a technical lemma.

\begin{lem}\label{lem-gkmin}
\[
\bG_{\bK}(K) \rightarrow \bG_{\bK}(-K).
\]
\end{lem}
\begin{proof}
Let $R \Def \set{K}{\bG_{\bK}(-K)}$. We use strong co-induction to show that $R \subseteq \bG_{\bK}$, That is, we have to show that
\[
R(K) \rightarrow (\bG(\bmin K) \land \bG(\bmax K) \land (\forall d \in \bGC)\, (K_{d} \ne \emptyset \to (R(\bt[K_d]) \vee \bG_{\bK}(\bt[K_d])))).
\]
Assume that $R(K)$. Then $\bG_{\bK}(-K)$ and hence 
\[
\bG(\bmin (-K)) \land \bG(\bmax (-K)) \land (\forall d \in \bGC)\, ((-K)_{d} \ne \emptyset \to \bG_{\bK}(\bt[(-K)_{d}])).
\]
Note that $-(K_{d}) = (-K)_{(-d)}$ and hence $\bt[K_{d}] = \bt[-(K_{d})] = \bt[(-K)_{(-d)}]$. Moreover, $\bmin (-K) = - \bmax K$ and $\bmax (-K) = - \bmin K$. Since $\bG$ is closed under $\lambda x. -x$, by \cite[Lemma~7]{becsl}, it follows that
\[
\bG(\bmax K) \land \bG(\bmin K) \land (\forall d \in \bGC)\, (K_{d} \ne \emptyset \to \bG_{\bK}(\bt[K_{d}])),
\]
as was to be shown.
\end{proof}

\begin{lem}\label{lem-clavd}
For $d \in \bGC$,
\[
\bG_{\bK}(K) \rightarrow K_d \not= \emptyset \rightarrow \bG_{\bK}(\bav{d}^{-1}[K_d]).
\]
\end{lem}
\begin{proof}
The statement follows by case distinction on $d$. Assume that  $\bG_{\bK}(K)$. 
Then $\bG_{\bK}(\bt[K_d])$.

\begin{ncase}
$d = -1$.
\end{ncase}
This case is obvious, as for $x\in\bII_{-1}$, $\bt(x) = 2x +1 = \bav{-1}^{-1}(x)$.

\begin{ncase}
$d=1$.
\end{ncase}
For $x\in\bII_{1}$,  $\bt(x) = 1 - 2x = -\bav{1}^{-1}(x)$. 
Therefore the statement follows with Lemma~\ref{lem-gkmin}. 
\end{proof}

\begin{lem}\label{lem-mirr}
Let $K \subseteq \bII_{1}$. Then
\[
\bG_{\bK}(K) \to \bG_{\bK}((\lambda x.\, 1-x)[K]).
\]
\end{lem}
\begin{proof}
The statement follows by co-induction. Note to this end that $(\lambda x.\, 1-x)[K] \subseteq \bII_{1}$ as well. Moreover, $\bmin ((\lambda x.\, 1-x)[K]) = 1- \bmax K$ and $\bmax ((\lambda x.\, 1-x)[K]) = 1-\bmin K$. Now assume that $\bG_{\bK}(K)$. Then $\bG(\bmin K)$ and $\bG(\bmax K)$. By \cite[Lemma~10]{becsl} we have for $x \in \bII_{1}$ with $\bG(x)$ that also $\bG(1-x)$. Thus, we obtain $\bG(\bmin (\lambda x.\, 1-x)[K])$ and $\bG(\bmax (\lambda x.\, 1-x)[K])$. Since for $x \in \bII_{1}$, $\bt(1-x) = 2x -1 = \bav{1}^{-1}(x)$, it follows with Lemma~\ref{lem-clavd} that \mytextcolor{red}{$\bG_{\bK}(\bt[(\lambda x.\, 1-x)[K]])$}.
\end{proof}

\begin{lem}\label{lem-II0}
Let $K \subseteq \bII$. Then
\[
\bG_{\bK}(1/2 K) \to \bG_{\bK}(K).
\]
\end{lem}
\begin{proof}
The statement follows again by co-induction. Assume that $\bG_{\bK}(1/2 K)$. Then $\bG(\bmin K / 2)$ and $\bG(\bmax K / 2)$. Since by \cite[Lemma~11]{becsl} $\bG$ is closed under $\lambda x.\, 2x$ for $|x|\le 1/2$, it follows that $\bG(\bmin K)$ and $\bG(\bmax K)$. 

As a further consequence of our assumption we have for $d \in \bGC$ with $K_{d}\ne \emptyset$ that $\bG_{\bK}(\bt[1/2 K_{d}])$. Because $\bt(x/2) = 1-|x|$, we obtain that $\bG_{\bK}((\lambda x.\, 1-|x|)[K_{d}])$. Hence, $\bG_{\bK}(|K_{d}|)$, by Lemma~\ref{lem-mirr}, and therefore $\bG_{\bK}(\bt[K_{d}])$. 
\end{proof}

Set
\begin{align*}
B^{\mathrm{min}}_{0} &\Def \bmin K \ne 0, 
&B^{\mathrm{min}}_{1} &\Def \bt(\bmin K) \ne 0, \\
B^{\mathrm{max}}_{0} &\Def \bmax K \ne 0, 
&B^{\mathrm{max}}_{1} &\Def \bt(\bmax K) \ne 0.
\end{align*}
Then
\[
\neg \neg(B^{\mathrm{min}}_{0}  \lor B^{\mathrm{min}}_{1})
\]
and 
\[\neg\neg(B^{\mathrm{max}}_{0}  \lor B^{\mathrm{max}}_{1}).
\]
It follows for 
\[
C_{i,j} \Def B^{\mathrm{min}}_{i} \land B^{\mathrm{max}}_{j}
\]
with $i, j \in \{ 0, 1 \}$ that
\begin{equation}\label{eq-negnegC}
\neg\neg(\bigvee_{0 \le i,j \le 1} C_{i,j}).
\end{equation}
Moreover, all $C_{i,j}$ are Harrop.

Now, let
\[
A(K) \Def (\exists E \in \bP_{\fin}(\bSD))\, 
(K \subseteq \bII_E \land
(\forall d \in E)\, K_{d} \ne \emptyset).
\]

\begin{lem}\label{lem-GK-to-ddownAKn}
$\bG_{\bK}(K) \to \itdown(A(K))$.
\end{lem}
\begin{proof}
Assume $\bG_{\bK}(K)$. Because of (\ref{eq-negnegC}) and the Rules~(\emph{$\itdown$-$\rest$-$\lor$}), (\emph{$\rest$-stab}), and (\emph{$\rest$-mp}), 
it suffices to show
\begin{enumerate}
\item\label{lem-GK-to-ddownAK-1}  $\rt{C_{0,0}}{A(K)}$,

\item\label{lem-GK-to-ddownAK-2}  $\rt{C_{0,1}}{A(K)}$,

\item\label{lem-GK-to-ddownAK-3}  $\rt{C_{1,0}}{A(K)}$,

\item\label{lem-GK-to-ddownAK-4}  $\rt{C_{1,1}}{A(K)}$.
\end{enumerate}

The assumption $\bG_{\bK}(K)$ entails that $\bG(\bmin K)$ and $\bG(\bmax K)$. 
Hence, we have for $x \in \{ \bmin K, \bmax K \}$ that 
\begin{equation}\label{eq-star1}
\qquad\rt{x \ne 0}{(x \ge 0 \lor x \le 0)}.
\end{equation}
Since from $\bG(x)$ we obtain that also $\bG(\bt(x))$, it follows in the same way that 
\begin{equation}\label{eq-star2}
\qquad\rt{\bt(x) \ne 0}{(\bt(x) \ge 0 \lor \bt(x) \le 0)}.
\end{equation}
Note that
\begin{gather*}
\bmin K \ge 0 \leftrightarrow K \ge 0, \\
\bmin K \le 0 \leftrightarrow K_{-1} \ne \emptyset, \\
\bmax K \le 0 \leftrightarrow K \le 0, \\
\bmax K \ge 0 \leftrightarrow K_1 \ne \emptyset.
\end{gather*}

(\ref{lem-GK-to-ddownAK-1}) 
Observe that $C_{0,0}$ is the formula $(\bmin K \ne 0\land\bmax K \ne 0)$.
We use (\ref{eq-star1}) for $x=\bmin K$ and $x=\bmax K$. 
With Rules~($\rest$-$\land$), ($\rest$-mon), and ($\rest$-antimon) we then obtain
\[\rt{\bmin K \ne 0 \land \bmax K \ne 0}{((\bmin K \ge 0 \lor \bmin K \le 0)\land (\bmax K \ge 0 \lor \bmax K \le 0))}\]
which is equivalent to 
\[\rt{C_{0,0}}{(\bmin K \ge 0 \lor (\bmin K \le 0 \land \bmax K \ge 0) \lor \bmax K \le 0)}\]
and, by the above equivalences, to
\[\rt{C_{0,0}}{(K\ge 0 \lor (K_{-1}\ne\emptyset \land K_{1}\ne\emptyset) \lor K \le 0)}.\]
Since the formula 
$(K\ge 0 \lor (K_{-1}\ne\emptyset \land K_{1}\ne\emptyset) \lor K \le 0)$
clearly implies $A(K)$, we are done by Rule~($\rest$-mon).

(\ref{lem-GK-to-ddownAK-2}) 
$C_{0,1}$ is the formula $(\bmin K \ne 0\land\bt(\bmax K) \ne 0)$.
We use (\ref{eq-star1}) for $x=\bmin K$ and (\ref{eq-star2}) for $x=\bmax K$. 
With a similar argument as in the previous case we receive
\[\rt{C_{0,1}}{((\bmin K \ge 0 \lor \bmin K \le 0)\land (\bt(\bmax K)\ge 0 \lor \bt(\bmax K)\le 0))}.
\]
Therefore, it suffices to show that $A(K)$ is implied by the formula
\[(\bmin K \ge 0 \lor \bmin K \le 0)\land (\bt(\bmax K)\ge 0 \lor \bt(\bmax K)\le 0)).\]
The latter is equivalent to
\[(K \ge 0 \lor K_{-1}\ne\emptyset) \land (|\bmax K| \le 1/2 \lor |\bmax K|\ge 1/2). \]

If $K \ge 0$, we choose $E\Def\{1\}$. 

If $K_{-1}\ne\emptyset$ and $|\bmax K| \le 1/2$ we chose $E\Def\{-1,0\}$.

If $K_{-1}\ne\emptyset$ and $|\bmax K|\ge 1/2$ we
have $\bmax K \ne 0$ and can therefore use (\ref{eq-star1}) and Rule~(\emph{$\rest$-mp})
to get $\bmax K \ge 0 \lor \bmax K \le 0$, that is, $K_1 \ne \emptyset \lor K \le 0$.

If $K_1 \ne \emptyset$, we choose $E\Def\{-1,1\}$. If $K\le 0$, we choose $E\Def\{-1\}$.

(\ref{lem-GK-to-ddownAK-3})  is dual to (\ref{lem-GK-to-ddownAK-2}).  

(\ref{lem-GK-to-ddownAK-4}) 
$C_{1,1}$ is the formula $(\bt(\bmin K) \ne 0 \land\bt(\bmax K) \ne 0)$.
Using (\ref{eq-star2}) for $x=\bmin K$ and $x=\bmax K$ we obtain
\[\rt{C_{1,1}}{((\bt(\bmin K) \ge 0 \lor \bt(\bmin K) \le 0)\land (\bt(\bmax K) \ge 0 \lor \bt(\bmax K) \le 0))}.\]
Therefore, it suffices to show that $A(K)$ is implied by the formula
\[(\bt(\bmin K) \ge 0 \lor \bt(\bmin K) \le 0)\land (\bt(\bmax K) \ge 0 \lor \bt(\bmax K) \le 0).\]
The latter is equivalent to
\[(|\bmin K|\le 1/2 \lor |\bmin K| \ge 1/2)\land (|\bmax K|\le 1/2 \lor |\bmax K| \ge 1/2).\] 

If $|\bmin K|\le 1/2$ and $|\bmax K|\le 1/2$, we choose $E\Def\{0\}$.

If $|\bmin K|\le 1/2$ and $|\bmax K|\ge 1/2$, then $\bmax K\ge 1/2$, hence we choose $E\Def\{0,1\}$.

If $|\bmin K|\ge 1/2$ and $|\bmax K|\le 1/2$, then $\bmin K\le -1/2$, hence we choose $E\Def\{-1,0\}$.

If $|\bmin K|\ge 1/2$ and $|\bmax K|\ge 1/2$, then, by (\ref{eq-star2}),
$(\bmin K \ge 0 \lor \bmin K \le 0)\land (\bmax K \ge 0 \lor \bmax K \le 0)$, which, as in 
case (\ref{lem-GK-to-ddownAK-1}), implies $A(K)$. 
\end{proof}

The above statements now allow the derivation of the result we are looking for.

\begin{prop}\label{prop-comgtoc}
\[
\bG_{\bK} \subseteq \bS_{\bK}^{*}.
\]
\end{prop}
\begin{proof}
The statement follows by co-induction. We have to show that
\[
\bG_{\bK}(K) \rightarrow \itdown((\exists E \in \bP_{\fin}(\bSD))\, 
(K \subseteq \bII_E \land
(\forall d \in E)\, (K_d \not= \emptyset \land \bG_{\bK}(\bav{d}^{-1}[K_d])))).
\]
From Lemmas~\ref{lem-clavd} and \ref{lem-II0} it follows
\begin{multline*}
\bG_{\bK}(K) \rightarrow ((\exists E \in \bP_{\fin}(\bSD))\, 
(K \subseteq \bII_E \land
(\forall d \in E)\, K_d \not= \emptyset)  \\
\rightarrow 
(\exists E \in \bP_{\fin}(\bSD))\, 
(K \subseteq \bII_E \land
(\forall d \in E)\, (K_d \not= \emptyset \wedge  \bG_{\bK}(\bav{d}^{-1}[K_d])))).
\end{multline*}
By the monotonicity of $\itdown$ we thus obtain
\begin{multline*}
\bG_{\bK}(K) \rightarrow (\itdown((\exists E \in \bP_{\fin}(\bSD))\, 
(K \subseteq \bII_E \land
(\forall d \in E)\, K_d \not= \emptyset)) \\
\rightarrow 
\itdown((\exists E \in \bP_{\fin}(\bSD))\, 
(K \subseteq \bII_E \land
(\forall d \in E)\, (K_d \not= \emptyset \wedge \bG_{\bK}(\bav{d}^{-1}[K_d]))))),
\end{multline*}
where, the assumption 
$\itdown((\exists E \in \bP_{\fin}(\bSD))\, 
(K \subseteq \bII_E \land
(\forall d \in E)\, K_d \not= \emptyset))$ 
can be discharged by Lemma~\ref{lem-GK-to-ddownAKn}.
\end{proof}

The result we have obtained so far is  analogous to the number case.
\begin{thm}\label{thm-sdtogtoconcsd}
$\bS_{\bK} \subseteq \bG_{\bK} \subseteq \bS_{\bK}^{*}$.
\end{thm}

\begin{rem} 
\mytextcolor{green}
{
The definition of $\bG_{\bK}(K)$ (Definition~\ref{def-gk}) 
can be simplified to
\begin{gather*}
\bG_{\bK}(K) = \bK(K) \land  \bG(\bmin K) \land \bG'_{\bK}(K) \\
 \intertext{where}
\bG'_{\bK}(K) \overset{\nu}{=} \bG(\bmax K) \land 
 (\forall d \in  \bGC)\, (K_{d} \ne \emptyset \to  \bG'_{\bK}(\bt[K_{d}])).
\end{gather*}
This is equivalent to \ref{def-gk} since $\bG$ is closed under the function 
$\bt$ and for $d\in \bGC$ with $K_{d} \ne \emptyset$,
$\bK(K)$ implies $\bK(\bt[K_{d}])$, 
and if $d=-1$, then $\bmin \bt[K_{d}] = \bt(\bmin K)$ while 
for $d=1$, $\bmin \bt[K_{d}] = \bt(\bmax K)$.
Although the new definition looks more complicated, it leads to simpler 
realisers since in each recursion step it refers to $\bG$ only once.
The definition of $\bG_{\bK}^{*}(K)$ in the subsequent Section~\ref{sec-concgray}  can be
simplified in a similar way.
}
\end{rem}

%% file: sec-concgray-final-final.tex
Next, set
\[
\bG_{\bK}^{*}(K) \overset{\nu}{=} \bK(K) \land \bG^{*}(\bmin K) \land \bG^{*}(\bmax K) \land (\forall d \in \bGC)\, (K_d \ne \emptyset \rightarrow \bG_{\bK}^{*}(\bt[K_d])).
\]

Our next and final goal is to show that $\bS_{\bK}^{*} = \bG_{\bK}^{*}$.

\begin{lem}\label{lem-smin}
If $\bS_{\bK}^{*}(K)$ then also 
\begin{enumerate}
\item\label{lem-smin-1} $\bS_{\bK}^{*}(-K)$.
\item\label{lem-smin-2} $\bS^{*}(\bmin K)$.
\item\label{lem-smin-2.5} $\bS^{*}(\bmax K)$.
\item\label{lem-smin-3} $(\forall d \in \bGC) (K_d \ne \emptyset \to \bS_{\bK}^{*}(\bt[K_d]))$.
\end{enumerate}
\end{lem}
\begin{proof}
(\ref{lem-smin-1}) Let $P \Def \set{K}{\bS_{\bK}^{*}(-K)}$. We use co-induction to prove that $P \subseteq  \bS_{\bK}^{*}$. That is, we have to show that
\[
P(K) \rightarrow \itdown((\exists E \in \bP_{\fin}(\bSD))\, 
(K \subseteq \bII_E \land
(\forall d \in E)\, (K_d \not= \emptyset \wedge P(\bav{d}^{-1}[K_d])))).
\]
By definition of $P$ it suffices to derive
\begin{multline*}
\itdown((\exists F \in \bP_{\fin}(\bSD))\, 
(-K \subseteq \bII_F \land
(\forall d \in F)\, ((-K)_d \not= \emptyset \wedge \bS_{\bK}^{*}(\bav{d}^{-1}[(-K)_d]))))\\
\rightarrow
\itdown((\exists E \in \bP_{\fin}(\bSD))\, 
(K \subseteq \bII_E \land
(\forall d \in E)\, (K_d \not= \emptyset \wedge P(\bav{d}^{-1}[K_d])))).
\end{multline*}
Because of the monotonicity rule for $\itdown$ we thus only have to show that
\begin{multline*}
(\exists F \in \bP_{\fin}(\bSD))\, 
(-K \subseteq \bII_F \land
(\forall d \in F)\, ((-K)_d \not= \emptyset \wedge \bS_{\bK}^{*}(\bav{d}^{-1}[(-K)_d]))) \\
\rightarrow
(\exists E \in \bP_{\fin}(\bSD))\, 
(K \subseteq \bII_E \land
(\forall d \in E)\, (K_d \not= \emptyset \wedge P(\bav{d}^{-1}[K_d]))),
\end{multline*}
which has been done in the proof of Lemma~\ref{lem-cpneg}(\ref{lem-cpneg-1}).

(\ref{lem-smin-2}) The proof is an adaptation of the proof of Lemma~\ref{lem-cpneg}(\ref{lem-cpneg-1.4}). Let 
\[
R \Def \set{x \in \bII}{(\exists K \in \bS_{\bK}^{*})\, x = \bmin K}.
\]
 We have to show that
\[
R(K) \to \itdown ((\exists d \in \bSD)\, (x \in \bII_{d} \land R(\bav{d}^{-1}(x)))).
\]
Assume $R(K)$. Then there is some $K \in \bS_{\bK}^{*}$ with $x = \bmin K$. It follows that
\[
\itdown((\exists E \in \bP_{\fin}(\bSD))\, 
(K \subseteq \bII_E \land
(\forall e \in E)\, (K_{e} \ne \emptyset \land \mytextcolor{red}{\bS_{\bK}^{*}(\bav{e}^{-1}[K_{e}]))))}.
\]
Because of the monotonicity law for $\itdown$ it suffices to prove that
\begin{multline*}
(\exists E \in \bP_{\fin}(\bSD))\, 
(K \subseteq \bII_E \land
(\forall e \in E)\, (K_{e} \ne \emptyset \land \mytextcolor{red}{\bS_{\bK}^{*}(\bav{e}^{-1}[K_{e}]))) }\to \\
(\exists d \in \bSD)\, (x \in \bII_{d} \land R(\bav{d}^{-1}(x))).
\end{multline*}
Order $\bSD$ again by $-1 < 0 < 1$ and let $d$ be the least element of $E$ with respect to this order. Then $\bmin K \in K_{d} \subseteq \bII_{d}$ and $\bav{d}^{-1}(\bmin K) \in \bav{d}^{-1}[K_{d}]$. Note that $\bav{d}^{-1}$ is monotone. Therefore, $\bav{d}^{-1}(\bmin K) = \bmin \bav{d}^{-1}[K_{d}]$. Since $ \bav{d}^{-1}[K_{d}] \in \bS_{\bK}^{*}$, it follows that $\bav{d}^{-1}(\bmin K) \in R$. 

(\ref{lem-smin-2.5}) As in Lemma~\ref{lem-cpneg}, the statement is a direct consequence of Statements~\ref{lem-smin-1} and \ref{lem-smin-2} as well as Lemma~\ref{lem-neg}(\ref{lem-neg-1}).

(\ref{lem-smin-3}) Set
\[
R^d_K(K) \Def \set{K}{(\exists Z \in \bS_{\bK}^{*})\, (Z_d \not= \emptyset \wedge K = \bt[Z_d])}.
\]
We use concurrent half-strong co-induction to show that $R^d_K \subseteq \bS_{\bK}^{*}$. That is, we have to prove that
\begin{equation}\label{eq-smin-3}
\begin{split}
R^d_K(K) \to  \itdown(\itdown((\exists E \in \bP_{\fin}(\bSD)&)\, (K \subseteq \bII_E \land  \mbox{} \\
&(\forall e \in E)\, (K_e \not= \emptyset \wedge R^d_K(\bav{e}^{-1}[K_e])))) \vee \bS_{\bK}^{*}(K)).
\end{split}
\end{equation}

If $R^d_K(K)$, there is some $Z \in \bS_{\bK}^{*}$ so that $Z_d \not= \emptyset$ and $K = \bt[Z_d]$. Since $\bS_{\bK}^{*}(Z)$, it follows that
\[
\itdown((\exists F \in \bP_{\fin}(\bSD))\, 
(Z \subseteq \bII_F \land (\forall f \in F)\, (Z_f \not= \emptyset \wedge \bS_{\bK}^{*}(\bav{f}^{-1}[Z_f])))).
\]

Now, assume that there is some $F \in \bP_{\fin}(\bSD)$ such that $Z = \bigcup_{f \in F} Z_f$ and for all $f \in F$, $Z_f \not= \emptyset$ and $\bS_{\bK}^{*}(\bav{f}^{-1}[Z_f])$. As in the proof of Lemma~\ref{lem-cpneg}(\ref{lem-cpneg-2}) it follows for $d \in \bGC$ with $Z_d \not= \emptyset$ and $K = \bt[Z_d]$ that $\bS_{\bK}^{*}(K)$, if $d \in F$ or, $0 \notin F$ and $d \notin F$, and $R^d_K(\bav{1}^{-1}[K])$, if $d \notin F$, but $0 \in F$. Thus, we have that 
\[
((\exists E' \in \bP_{\fin}(\bSD))\, 
(K \subseteq \bII_{E'} \land (\forall e \in E')\, (K_e \not= \emptyset \wedge R^d_K(\bav{e}^{-1}[K_d])))) \vee \bS_{\bK}^{*}(K).
\]
With the monotonicity and the idempotency of $\itdown$ and $\itdown$-$\vee$ distribution we therefore obtain
\begin{multline*}
\itdown((\exists F \in \bP_{\fin}(\bSD))\, 
(Z \subseteq \bII_F \land (\forall f \in F)\, (Z_f \not= \emptyset \wedge \bS_{\bK}^{*}(\bav{f}^{-1}[Z_f])))) \to \mbox{} \\
\itdown(\itdown((\exists E' \in \bP_{\fin}(\bSD))\, 
(K \subseteq \bII_{E'} \land
(\forall e \in E')\, \\ (K_e \not= \emptyset \wedge R^d_K(\bav{e}^{-1}[K_e])))) \vee \bS_{\bK}^{*}(K)),
\end{multline*}
of which (\ref{eq-smin-3}) is a direct consequence.
\end{proof}

By co-induction we now obtain the first inclusion we are looking for.
\begin{prop}\label{prop-scg}
$\bS_{\bK}^{*} \subseteq \bG_{\bK}^{*}$.
\end{prop}

Let us now start with proving the converse inclusion. Again we need some technical results.

\begin{lem}\label{lem-sgmin}
$\bG_{\bK}^{*}(K) \rightarrow \bG_{\bK}^{*}(-K)$.
\end{lem}
\begin{proof}
Let $P \Def \set{K}{\bG_{\bK}^{*}(-K)}$. We prove $P \subseteq \bG_{\bK}^{*}$ by strong co-induction. That is, we must show that
\[
P(K) \to (\bG^{*}(\bmin K) \land \bG^{*}(\bmax K) \land (\forall d \in \bGC)\, (K_d \ne \emptyset \to (P(\bt[K_d]) \lor \bG_{\bK}^{*}(\bt[K_d])))).
\]

Assume that $P(K)$. Then $\bG_{\bK}^{*}(-K)$ and hence $\bD_{\bK}^{*}(-K)$ and $\bG_{\bK}^{*}(\bt[(-K)_d])$, for $d \in \bGC$ with $(-K)_d \not= \emptyset$. Note that $\bt(x) = \bt(-x)$ and $(-K)_d = -(K_{(-d)})$. Thus, $\bG_{\bK}^{*}(\bt[K_{d}])$, for $d \in \bGC$ with $K_d \not= \emptyset$. By Lemma~\ref{lem-gneg} it follows that $\bG^{*}(\bmax K)$ and $\bG^{*}(\bmin K)$. Moreover, as $\bt(x) = \bt(-x)$ and $(-K)_{d} = -(K_{(-d)})$, $\bG_{\bK}^{*}(\bt[K_{d}])$, for all $d \in \bGC$ with $K_{d} \ne \emptyset$.
\end{proof}

\begin{lem}\label{lem-savd}
For $d \in \bGC$,
\[
\bG_{\bK}^{*}(K) \rightarrow K_d \not= \emptyset \rightarrow \bG_{\bK}^{*}(\bav{d}^{-1}[K_d]).
\]
\end{lem}
The statement follows as in case of Lemma~\ref{lem-clavd}.

\begin{lem}\label{lem-starmir}
Let $K \subseteq \bII_{1}$. Then
\[
\bG_{\bK}^{*}(K) \to \bG_{\bK}^{*}((\lambda x.\, 1-x)[K]).
\]
\end{lem}
The proof proceeds as in Lemma~\ref{lem-mirr} by using Lemma~\ref{lem-min1}.

\begin{lem}\label{lem-starII0}
Let $K \in \bK$. Then 
\[
\bG_{\bK}^{*}(1/2 K) \to \bG_{\bK}^{*}(K).
\]
\end{lem}
The result follows as in Lemma~\ref{lem-II0} by applying Lemma~\ref{lem-ii0}.

\begin{lem}\label{lem-scomgc}
\[
\bG_{\bK}^{*}(K) \rightarrow 
\itdown((\exists E \in \bP_{\fin}(\bSD))\, 
(K \subseteq \bII_E \land (\forall d \in E)\, K_d \not= \emptyset)).
\]
\end{lem}
\begin{proof}
The proof follows the derivation of Lemma~\ref{lem-GK-to-ddownAKn}.
Set 
\[
A(K) \Def (\exists E \in \bP_{\fin}(\bSD))\, 
(K \subseteq \bII_E \land (\forall d \in E)\, K_d \not= \emptyset).
\] 
Then the assertion is
\[
\bG_{\bK}^{*}(K) \rightarrow \itdown(A(K)).
\]

Let again
\begin{align*}
B^{\mathrm{min}}_{0} &\Def \bmin K \ne 0, 
&B^{\mathrm{min}}_{1} &\Def \bt(\bmin K) \ne 0, \\
B^{\mathrm{max}}_{0} &\Def \bmax K \ne 0, 
&B^{\mathrm{max}}_{1} &\Def \bt(\bmax K) \ne 0.
\end{align*}
and 
\[C_{i,j} \Def B^{\mathrm{min}}_{i} \land B^{\mathrm{max}}_{j},
\]
 for $i, j \in \{0, 1\}$.
As we have seen in the proof of Lemma~\ref{lem-GK-to-ddownAKn}, it suffices to show that
\begin{enumerate}
\item\label{lem-scomgc-1}  $\rt{C_{0,0}}{\itdown(A(K))}$,

\item\label{lem-scomgc-2}  $\rt{C_{0,1}}{\itdown(A(K))}$,

\item\label{llem-scomgc-3}  $\rt{C_{1,0}}{\itdown(A(K))}$,

\item\label{lem-scomgc-4}  $\rt{C_{1,1}}{\itdown(A(K))}$.
\end{enumerate}

Assume that $\bG_{\bK}^{*}(K)$ and note for $x \in \{ \bmin K, \bmax K \}$ that $\bG^{*}(x)$ entails $\bG^{*}(\bt(x))$. From both we obtain that $\rt{x \ne 0}{\itdown(\bB(x))}$ and  $\rt{\bt(x) \ne 0}{\itdown(\bB(\bt(x)))}$.  Because of Rules~($\rest$-$\land$), (\emph{$\rest$-mon}), and (\emph{$\rest$-antimon}) it follows that
\[
\rt{x \ne 0 \land y \ne 0}{(\itdown(\bB(x)) \land \itdown(\bB(y)))},
\]
from which we obtain with Rules~(\emph{$\itdown$-$\land$-intro}) and (\emph{$\rest$-mon}) that
\begin{equation}\label{eq-scomgc}
\rt{x \ne 0 \land y \ne 0}{\itdown(\bB(x) \land \bB(y))}.
\end{equation}
As we have seen in the proof of Lemma~\ref{lem-GK-to-ddownAKn},
\[
(\bB(x) \land \bB(y)) \to \itdown(A(K)),
\]
for each choice of $x$ and $y$.
With Rules~(\emph{$\itdown$-mon}) and (\emph{$\itdown$-idem}) we thus have that
\[
\itdown(\bB(x) \land \bB(y)) \to \itdown(A(K)).
\]
Consequently,  by (\ref{eq-scomgc}) and Rule~(\emph{$\rest$-mon}), we  obtain that $\rt{C_{i,j}}{\itdown(A(K))}$, for $i, j \in \{0, 1\}$.
\end{proof}

\begin{prop}\label{prop-cgcs}
$\bG_{\bK}^{*} \subseteq \bS_{\bK}^{*}$.
\end{prop}
\begin{proof}
The statement follows by co-induction. We need to show that
\[
\bG_{\bK}^{*}(K) \rightarrow 
\itdown((\exists E \in \bP_{\fin}(\bSD)) 
(K \subseteq \bII_E \land
(\forall d \in E)\, (K_d \not= \emptyset \wedge \bG_{\bK}^{*}(\av{d}^{-1}[K_d])))).
\]

From Lemmas~\ref{lem-savd}  and \ref{lem-starII0} it follows 
\begin{multline*}
\bG_{\bK}^{*}(K) \rightarrow ((\exists E \in \bP_{\fin}(\bSD)\, 
(K \subseteq \bII_E \land
(\forall d \in E)\, K_d \not= \emptyset) \\
\rightarrow 
(\exists E \in \bP_{\fin}(\bSD)) 
(K \subseteq \bII_E \land
(\forall d \in E)\, (K_d \not= \emptyset \wedge \bG_{\bK}^{*}(\bav{d}^{-1}[K_d])))).
\end{multline*}
By monotonicity of $\itdown$ we thus obtain
\begin{multline*}
\bG_{\bK}(K) \rightarrow (\itdown((\exists E \in \bP_{\fin}(\bSD)) 
(K \subseteq \bII_E \land
(\forall d \in E)\, K_d \not= \emptyset)) \\
\rightarrow 
\itdown((\exists E \in \bP_{\fin}(\bSD)) 
(K \subseteq \bII_E \land 
(\forall d \in E)\, \mytextcolor{red}{(K_d \not= \emptyset} \wedge \bG_{\bK}^{*}(\bav{d}^{-1}[K_d]))))),
\end{multline*}
where, the assumption $\itdown((\exists E \in \bP_{\fin}(\bSD))\, 
(K \subseteq \bII_E \land 
(\forall d \in E)\, K_d \not= \emptyset))$ can be discharged by Lemma~\ref{lem-scomgc}.
\end{proof}

As a consequence of Propositions~\ref{prop-scg} and~\ref{prop-cgcs} we now obtain our central result for the compact sets case.
\begin{thm}\label{thm-maincompgc}
$\bS_{\bK}^{*} = \bG_{\bK}^{*}.$
\end{thm}

%% file: sec-concl-final-final.tex
In this paper the computational power of infinite Gray code has been re-considered and compared with the  signed digit representation which is mostly used in applications. Infinite Gray code is a redundancy-free representation of the real numbers, whereas the signed digit representation has a high degree of redundancy: every real number has infinitely many names. Instead of all real numbers only the interval [-1, 1] was considered.

The central aim was to study the relationship between both representations without having to discuss the manipulation of code words directly. To this end, for each of the two kinds of representation, co-inductive characterisations for the spaces under consideration were introduced in a formal logical system as predicates $\bG$ and $\bS$, from which the representation can be recovered via a realisability interpretation. Instead of dealing with representations directly, the predicates were compared. Computable translations between the representations can  be extracted from the formal proofs. The proofs also guarantee the correctness of the extracted programs.

As was known from earlier studies by Tsuiki~\cite{ts,tsug}, infinite Gray code can be translated into signed digit code in a sequential way; for the converse translation, however, one has to allow the computations to proceed concurrently. In \cite{bt}, Berger and Tsuiki introduced a modality $\ddown$ for concurrency. $\ddown(A)$ has no effect on the classical validity of the formula $A$, but on its realisability interpretation: two concurrent processes try to realise $A$, in case $\ddown(A)$ is realisable, at least one of them will do so. With help of this modality a predicate $\bS_{2}$ was co-inductively defined, the realisers of which are again streams of signed digits. However, they can be computed concurrently. It was shown that $\bS \subseteq \bG \subseteq \bS_{2}$. 

In the present paper the set of rules coming with the modality $\ddown$ was enlarged by two new realisable rules, and several other useful rules were derived.
Moreover, the modality was inductively extended to a modality $\itdown$ of bounded non-determinism, co-inductively leading to predicates $\bG^{*}$ and $\bS^{*}$. Proof rules for the new modality were derived, and by this way it was shown that $\bG^{*} = \bS^{*}$, thus extending  the result in \cite{bt}.

A powerful proof tool in the proof of the inclusion $\bS \subseteq \bG$ case was Archimedean
induction. Here, a similar rule was presented for the concurrent case.

In \cite{bs,sp} the present authors have given a co-inductive characterisation of the hyperspace of all non-empty compact subsets of a given digit space. Instead of streams of digits, as in the point case, extracted realisers are now finitely branching infinite trees with nodes being labelled with digits. By doing so, in particular a canonical way of lifting the signed digit representation of the real numbers in $[-1, 1]$ to a representation of the non-empty compact subsets of $[-1, 1]$  is obtained. The representation is very natural: the infinite paths of a tree representing a compact set $K$ correspond to the streams representing the elements of $K$.

A central aim of the present research was to do analogous investigations for the lifted representations as was done in the point case. The situation turned out very similar to the point case. Predicates $\bG_{K}, \bS_{K}$ and $\bS_{K}^{*}$ were defined co-inductively and the inclusions $\bS_{K} \subseteq \bG_{K} \subseteq \bS_{K}^{*}$ shown. Note, however, that for the last inclusion one had to use the stronger modality $\itdown$ in the definition of a   predicate for `concurrent' signed digit representation, whereas in the point case the use of $\ddown$ sufficed. A co-inductive predicate $\bG_{K}^{*}$  for `concurrent' Gray code was introduced as well and $\bG_{K}^{*} = \bS_{K}^{*}$ derived. 

Moreover, an Archimedean induction rule for non-empty compact subsets was obtained.

A computability-theoretic approach to representing compact sets is carried out in work by Pauly and Tsuiki~\cite{pt} who show in particular that $\KKK(\II)$ has a faithful $\TT^{\omega}$-representation $\TT^{\omega} \to \KKK(\II)$. Here, $\TT$ is the partial order $(\{ \bot, 0 , 1 \}, \sqsubseteq)$ with $\bot \sqsubseteq 0, 1$. In this study compact sets are  represented as trees as well, but then converted to bottomed sequences in such a way that for finite sets the number of bottoms in the sequence increased by 1 coincides with the cardinality of the set.
The exact relationship of this kind of Gray code for $\KKK(\II)$ with the one introduced in the present paper will have to be investigated in future work. 
Since the constructions given by Pauly and Tsuiki use coding and dove-tailing techniques, 
which correspond to a direct reference to a fixed operational semantics, it is unclear whether 
they can be recast in our abstract setting.